\documentclass{tac}

\usepackage{amsmath}
\usepackage{amssymb}
\usepackage{amsfonts}
\usepackage{amscd}
\usepackage{pgf,tikz}
\usetikzlibrary{arrows}
\usepackage{tikz-cd}  
\usepackage{bm}
\usepackage[square]{natbib}
\renewcommand{\cite}{\citep}
\usepackage[colorlinks=true]{hyperref}
\hypersetup{allcolors=[rgb]{0.1,0.1,0.4},breaklinks=true}

\author{Juan Pablo Vigneaux}

\thanks{The research presented here was developed at the Universit\'e Paris Diderot as part of my doctoral dissertation. I would like to thank Daniel Bennequin for his invaluable advice during my Ph.D. years, as well as Samson Abramsky, Philippe Elbaz-Vincent, and TAC's anonymous referee---who reviewed the material presented here at different stages---for their detailed feedback. I am also grateful to Gr\'egoire Sergeant-Perthuis, Olivier Peltre, Jean-Michel Fischer, and Daniel Juteau for many valuable discussions during the last years. }

\address{Universit\'e de Paris, Sorbonne Universit\'e, CNRS, Institut de Math\'ematiques de Jussieu-Paris Rive Gauche (IMJ-PRG), 75013, Paris, France.\\
Max-Planck-Institut f\"{u}r Mathematik in den Naturwissenschaften, Inselstra{\ss}e 22, 04013, Leipzig
}

\title {Information structures and their cohomology}

\copyrightyear{2020}


\eaddress{vigneaux@mis.mpg.de}

\thirdleveltheorems
\newtheorem{theorem}{Theorem}
\newtheorem{prop}{Proposition}
\newtheorem{lem}{Lemma}
\newtheoremrm{ex}{Example}
\newtheorem{defi}{Definition}
\newtheorem{cor}{Corollary}
\newtheoremrm{remark}{Remark}

\newcommand{\salg}[1]{\mathfrak {#1}}
\newcommand{\supp}{\operatorname{supp}}

\renewcommand{\ker}{\operatorname{ker}}
\newcommand{\coker}{\operatorname{coker}}
\newcommand{\im}{\operatorname{im}}
\newcommand{\coim}{\operatorname{coim}}
\newcommand{\End}{\operatorname{End}}

\newcommand{\Hom}{\operatorname{Hom}}
\newcommand{\Ext}{\operatorname{Ext}}

\newcommand{\id}{\operatorname{id}}

\newcommand{\Ob}{\operatorname{Ob}}

\newcommand{\cat}[1]{\mathbf{#1}}
\newcommand{\Sets}{\cat{Sets}}
\newcommand{\Mod}{\cat{Mod}}
\newcommand{\MeasSets}{\cat{Meas}}
\newcommand{\InfoStr}{\cat{InfoStr}}
\newcommand{\ObsFin}{\cat{Obs}_{\mathrm{fin}}}

\newcommand{\eq}[1]{\overset{\tiny{(#1)}}=}
\newcommand{\sm}{\setminus}

\newcommand{\Rr}{\mathbb{R}}
\newcommand{\Nn}{\mathbb{N}}
\newcommand{\Zz}{\mathbb{Z}}

\newcommand{\set}[2]{\{\,#1 \, : \, #2\,\} }
\newcommand{\bigset}[2]{\left\{\,#1 \, : \, #2\,\right\} }

\newcommand{\Sring}{\sheaf A}
\newcommand{\Sinf}{\cat{S}}
\newcommand{\Smon}{\sheaf S}
\newcommand{\Us}{\top}
\newcommand{\FProb}{\sheaf P}

\newcommand{\sheaf}[1]{\mathscr{#1}}

\newcommand{\ent}[1]{S_{#1}}

\newcommand{\keyt}[1]{\emph{#1}}

\begin{document}

\maketitle

\begin{abstract}
We introduce the category of \emph{information structures}, whose objects are suitable diagrams of measurable sets that encode the possible outputs of a given family of observables and their mutual relationships of refinement; they serve as mathematical models of contextuality in classical and quantum settings. Each information structure can be regarded as a ringed site with trivial topology; the structure ring is generated by the observables themselves and its multiplication corresponds to  joint measurement. We extend Baudot and Bennequin's definition of \emph{information cohomology} to this setting, as a derived functor in the category of modules over the structure ring, and show explicitly that the bar construction gives a projective resolution in that category, recovering in this way the cochain complexes previously considered in the literature. Finally, we study the particular case of a one-parameter family of coefficients made of functions of probability distributions. The only $1$-cocycles are Shannon entropy or  Tsallis $\alpha$-entropy, depending on the value of the parameter. 
\end{abstract}

\tableofcontents

\section{Introduction}
Entropy plays a fundamental role in several domains of mathematics and physics, and it is natural to ask why it is so important and ubiquitous. A pragmatic answer might highlight the connections to  limiting theorems in probability theory and dynamical systems; examples are the Shannon-Macmillan-Breiman theorem \cite[Thm.~16.8.1]{Cover2006} or the quantification of large deviations in terms of relative entropy \cite{Varadhan2003}. But entropy also possess remarkable algebraic properties, as Shannon already pointed out in the foundational article of information theory \protect\cite{Shannon1948}. He gave there an ``axiomatic characterization'' of entropy based on expected  (``natural'') properties of a measure of uncertainty. Since then, many authors have contributed with similar theorems, proposing alternative axiomatic characterizations of \emph{information functions} (many of them summarized in \protect\cite{Csiszar2008}; see also \cite{Khinchin1957,Otahal1994,Hatori1958}); the list includes a recent category-theoretic article \protect\cite{Baez2011} that focuses on the \emph{information loss} induced by reductions between finite probability spaces. Some of these works led to the study of functional equations uniquely solved by the entropy \protect\cite{Aczel1975}, hence to very sophisticated techniques involving real analysis.

The purpose of the present article---which complements a previous one by Pierre Baudot and Daniel Bennequin \protect\citeyearpar{Baudot2015}---is to develop a new perspective that identifies entropy with a topological invariant of a finite statistical system. In particular, entropy appears as a cohomology class and not merely as a function. We introduce here a new definition of \emph{information structures} (categories of observables) and cohomological invariants associated to certain presheaves on them, using the  framework developed by Artin, Grothendieck, Verdier and their collaborators in the SGA 4 \cite{Artin1972, Artin1972-2}. We recall that toposes were introduced there as a general foundation of topology, that allowed a unified study of several cohomological invariants involving groups, topological spaces, and schemes. Our results constitute an extension of the field of application of these ideas.

\subsection{Entropies and their algebraic characterization}\label{sec:intr:entropy}
 \citet{Shannon1948} defined the information content of a random variable $X$, taking values in a finite set ${\sheaf E}_X$, by the formula
\begin{equation}\label{eq:shannon_ent_def}
\ent 1[X](P) := -\sum_{x\in {\sheaf E}_X} P (X=x) \log   P(X=x),
\end{equation}
where $P$ denotes a probability measure (law) on ${\sheaf E}_X$. The function $S_1$ is called (Gibbs-Shannon) entropy, and quantifies the  uncertainty of a measurement.\footnote{In information theory, it is customary to write $H(X)$ instead of $S_1[X](P)$. We have decided to reserve the use of $H$ for cohomology, employing instead the letter $S$, common in physics (although our $S_1$ is adimensional). In turn, the presence of Tsallis $\alpha$-entropies $S_\alpha$ justifies the subscript $1$, as explained below. Finally, the functional equations involved in this work contain evaluations of $S_1$ at different laws, making necessary to mention the argument $P$ explicitly. }

Given two random variables $X$ and $Y$, valued respectively in sets ${\sheaf E}_X$ and ${\sheaf E}_Y$, their joint measurement $(X,Y)$ is also random variable, valued in ${\sheaf E}_{XY}\subset {\sheaf E}_X\times {\sheaf E}_Y$. Following again  \cite{Shannon1948}, a probability law $P$ on ${\sheaf E}_{XY}$  can be represented by a tree as in Fig. \ref{intro:fig:trees}-(a).  The probability of observing $X=x$ is computed as the sum of all the outputs of $(X,Y)$ that contain $x$ in the first component: $X_*P(x):=P(X=x)=\sum_{ (x,y)\in {\sheaf E}_{XY}} P(x,y).$  The probability $X_*P$ on ${\sheaf E}_X$ is called \emph{marginal law}. Instead of measuring directly $(X,Y)$ one could measure first $X$, which constitutes a first random choice; the uncertainty that remains after obtaining the result $X=x_0$ is represented by the \emph{conditional probability law}  $P|_{X=x_0}:{\sheaf E}_{XY}\to [0,1]$, given by 
\begin{equation}
P|_{X=x_0}(x,y) := \begin{cases}
\frac{P(x,y)}{X_*P(x_0)} & \text{if } x=x_0 \\
0 & \text{otherwise}
\end{cases},
\end{equation}
provided $X_*P(x_0)>0$ (it remains undefined for $x_0$ in the maximal $X_*P$-null set). 
This iterated choice/measurement can in turn be pictured as a tree, e.g.  Fig. \ref{intro:fig:trees}-(b). The function $S_1$ satisfies the so-called \emph{chain rule}
\begin{equation}\label{intro:chain_rule_1}
\ent 1[(X,Y)](P) = \ent 1[X](X_*P) + \sum_{\substack{x\in {\sheaf E}_X\\ X_*P(x) > 0}} X_*P(x) \ent 1 [Y](Y_*P|_{X=x})
\end{equation}
 Evidently, if the measurement of $Y$ is performed first, we obtain another tree, Fig. \ref{intro:fig:trees}-(c), that corresponds to
\begin{equation}\label{intro:chain_rule_2}
\ent 1[(X,Y)](P) = \ent 1[Y](Y_*P) + \sum_{\substack{y \in {\sheaf E}_Y\\ Y_*P(y)>0}} Y_*P(y) \ent 1[X](X_*P|_{Y=y})
\end{equation}

 \citet[p.~392-393]{Shannon1948} gave an algebraic characterization of  
\begin{equation}
H_n:\Delta^{n}\to \Rr,\quad (p_0,...,p_n) \mapsto -\sum_{i=0}^n p_i\log p_i
\end{equation} as the only family of continuous functions that satisfies the chain rule \eqref{intro:chain_rule_1} for any possible tree---this is, for arbitrary pairs  $(X,Y)$,  setting $S_1[X] = H_{|{\sheaf E}_X|}$ and so on---and  such that $H_n(1/n,...,1/n)$ is monotonic in $n$.
\begin{figure}
\centering
\def\radpoint{.07}
\def\xdist{3}
\def\ydist{2.5}
\begin{tikzpicture}
\node[coordinate] (o) at (0,0) {};
\node[coordinate] (a) at (\xdist,\ydist) {};
\node[coordinate] (b) at (\xdist,0) {};
\node[coordinate] (c) at (\xdist,-\ydist) {};
\draw[fill] (a) circle [radius=\radpoint];
\draw[fill] (b) circle [radius=\radpoint];
\draw[fill] (c) circle [radius=\radpoint];
\draw[fill] (o) circle [radius=\radpoint];
\draw[thick] (o) -- (a);
\draw[thick] (o) -- (b);
\draw[thick] (o) -- (c);
\node [above left] at  (0.6*\xdist,\ydist/2) {$p_{00}$};
\node [above] at  (0.6*\xdist,0) {$p_{01}$};
\node [below left] at (0.6*\xdist,-\ydist/2) {$p_{11}$};
\node at (\xdist/2,-1.3*\ydist) {(a)};

\begin{scope}[xshift=5cm]
\node[coordinate] (o) at (0,0) {};
\node[coordinate] (a) at (\xdist,\ydist) {};
\node[coordinate] (b) at (\xdist,0) {};
\node[coordinate] (c) at (\xdist,-\ydist) {};
\node[coordinate] (ab) at (0.5*\xdist,0.5*\ydist) {};
\node[coordinate] (bc) at (0.5*\xdist,-0.5*\ydist) {};
\draw[fill] (a) circle [radius=\radpoint];
\draw[fill] (b) circle [radius=\radpoint];
\draw[fill] (c) circle [radius=\radpoint];
\draw[fill] (o) circle [radius=\radpoint];
\draw[fill] (ab) circle [radius=\radpoint];
\draw[fill] (bc) circle [radius=\radpoint];
\draw[thick] (o) -- (ab);
\draw[thick] (ab) -- (a);
\draw[thick] (ab) -- (b);
\draw[thick] (o) -- (c);
\node [above left] at  (0.35*\xdist,0.25*\ydist) {$p_{00}+p_{01}$};
\node [above left] at  (0.85*\xdist,0.75*\ydist) {$\frac{p_{00}}{p_{00}+p_{01}}$};
\node [below left] at  (0.90*\xdist,0.25*\ydist) {$\frac{p_{01}}{p_{00}+p_{01}}$};
\node [below left] at  (0.35*\xdist,-0.25*\ydist) {$p_{11}$};
\node [below left] at  (0.85*\xdist,-0.75*\ydist) {$1$};
\node at (\xdist/2,-1.3*\ydist) {(b)};
\end{scope}

\begin{scope}[xshift=10cm]
\node[coordinate] (o) at (0,0) {};
\node[coordinate] (a) at (\xdist,\ydist) {};
\node[coordinate] (b) at (\xdist,0) {};
\node[coordinate] (c) at (\xdist,-\ydist) {};
\node[coordinate] (ab) at (0.5*\xdist,0.5*\ydist) {};
\node[coordinate] (bc) at (0.5*\xdist,-0.5*\ydist) {};
\draw[fill] (a) circle [radius=\radpoint];
\draw[fill] (b) circle [radius=\radpoint];
\draw[fill] (c) circle [radius=\radpoint];
\draw[fill] (o) circle [radius=\radpoint];
\draw[fill] (ab) circle [radius=\radpoint];
\draw[fill] (bc) circle [radius=\radpoint];
\draw[thick] (o) -- (a);
\draw[thick] (o) -- (bc);
\draw[thick] (bc) -- (b);
\draw[thick] (bc) -- (c);
\node [above left] at  (0.35*\xdist,0.25*\ydist) {$p_{00}$};
\node [above left] at  (0.85*\xdist,0.75*\ydist) {$1$};
\node [below left] at  (0.35*\xdist,-0.25*\ydist) {$p_{01}+p_{11}$};
\node [above left] at  (0.85*\xdist,-0.25*\ydist) {$\frac{p_{01}}{p_{01}+p_{11}}$};
\node [below left] at  (0.90*\xdist,-0.75*\ydist) {$\frac{p_{11}}{p_{01}+p_{11}}$};
\node at (\xdist/2,-1.3*\ydist) {(c)};
\end{scope}
\end{tikzpicture}
\caption[Decomposition of a choice from three possibilities (in general).]{Different groupings when ${\sheaf E}_X = {\sheaf E}_Y = \{0,1\}$ and ${\sheaf E}_{XY} =\{ (0,0), (0,1),(1,1)\}$. We denote by $p_{ij}$ the probability of the point $(i,j)\in {\sheaf E}_{XY}$. In (b) and (c), the probabilities to the left are the marginals $X_*P$ and $Y_*P$, respectively, and those to the right are the conditional laws on the appropriate subset of ${\sheaf E}_{XY}$.  }\label{intro:fig:trees}
\end{figure}

It is worth noticing that several generalizations of entropy play a role in  information theory and statistical mechanics. One of them is the structural $\alpha$-entropy, defined for each $\alpha\in ]0,\infty[\sm\{1\}$ as
\begin{equation}\label{eq:tsallis-entropy-def}
S_\alpha[X](P)= \frac 1{1-\alpha} \left(\sum_{x\in {\sheaf E}_X} P(x)^\alpha -1\right),
\end{equation}
It was introduced axiomatically in 1967 by  \citet{Havrda1967} (who characterized it up to a multiplicative constant). The use of $\alpha$-entropies in statistical mechanics was proposed by  \citet{Tsallis1988}, and the most common name for $S_\alpha$ is Tsallis $\alpha$-entropy. This function satisfies the deformed equation
\begin{equation}\label{intro:chain_rule_alpha}
\ent \alpha[(X,Y)](P) = \ent \alpha[Y](Y_*P) + \sum_{\substack{y \in {\sheaf E}_Y\\ Y_*P(y)>0}} (Y_*P(y))^\alpha \ent \alpha[X](X_*P|_{Y=y}).
\end{equation}

  \citet{Tverberg1958} was the first to deduce from the chain rule a simple functional equation that characterized Shannon entropy, called ``fundamental equation of information theory'':
 \begin{equation}\label{eq:intro:FEITH}
 f(x) + (1-x) f\left(\frac{y}{1-x}\right) = f(y) +(1-y) f\left(\frac{x}{1-y}\right),
 \end{equation}
 where $f:[0,1]\to \Rr$ is an unknown function such that $f(0)=f(1)=1$, and $x,y\in [0,1)$ are such that $x+y\in [0,1]$. The only symmetric, measurable solutions of this equation are the real multiples of $s_1(x):=-x\ln(x)+ (1-x)\ln(1-x)$ \cite{Lee1964}.    The result is  quite remarkable, because Shannon's characterization of the functions $H_n$ requires an \emph{infinite} number of equations---for \emph{any} random variable and \emph{any} possible grouping of its outcomes---along with a strong regularity of $H_n$. \citet{Daroczy1970} proposed a similar equation solved by the $\alpha$-entropy $s_\alpha(x) := x^\alpha +(1-x)^\alpha-1$.\footnote{For a detailed historical introduction and a comprehensive treatment of the subject, up to 1975, see the book by  \citet{Aczel1975}.}
 
In the same vein, if  the product $(X,Y)$ is nondegenerate (see below), then the system of functional equations \eqref{intro:chain_rule_1}-\eqref{intro:chain_rule_2}, with measurable \emph{unknowns} $\ent 1 [X]$, $\ent 1[Y]$, and $\ent 1[(X,Y)]$, is uniquely solved by the corresponding Shannon entropies  \eqref{eq:shannon_ent_def}, up to a multiplicative constant.  This holds even for the situation pictured in Figure \ref{intro:fig:trees}, that is evidently the simplest possible choice that can be broken down in two different ways (see Proposition \ref{functional equations}).

More importantly, the chain-rule-like functional equations \eqref{intro:chain_rule_1}-\eqref{intro:chain_rule_2} accept a cohomological interpretation. Let us define, for any probabilistic functional $P\mapsto f(P)$, a new functional $X.f$ given by 
\begin{equation}\label{eq:intro:action}
(X.f)(P) := \sum_{x\in {\sheaf E}_X}  X_*P(X) f(Y_*P|_{X=x}).
\end{equation}
in order to rewrite \eqref{intro:chain_rule_1} as 
\begin{equation}\label{H_cocycle}
0=  X.\ent 1[Y]-\ent 1[(X,Y)] + \ent 1[X].
\end{equation}
The notation is meant to suggest an action of random variables on probabilistic functionals, and in fact the equality $Z.(X.f)=(Z,X).f$ holds. There is an strong resemblance between \protect\eqref{H_cocycle} and a cocycle equation in group cohomology.  \citet{Baudot2015} formalized this analogy  introducing an adapted cohomology theory---information cohomology---through an explicit differential complex that recovered the equations \eqref{H_cocycle} as $1$-cocycle conditions. They used presheaves, exploiting a notion of \protect\emph{locality} specific to the problem: the entropy of a variable $X$ only depends on the marginalized version $X_*P$ of any global law $P$.\footnote{\citet{Cathelineau1988} was the first to find a cohomological interpretation for the fundamental equation \eqref{eq:intro:FEITH}: an analogue of it is involved in the computation of the homology of $SL_2$ over a field of characteristic zero, with coefficients in the adjoint action; however, this result was not explicitly connected to Shannon entropy or information theory. The first published work in this direction is a note by Kontsevich (reproduced as an appendix in \cite{Elbaz2002}), that introduces $H_p(x) = \sum_{k=1}^{p-1} \frac{x^k}{k}$ as ``a residue modulo $p$'' of entropy, being the only continuous map $f:\Zz/p\Zz\to \Zz/p\Zz$ that verifies $f(x)=f(1-x)$ and an equation equivalent to \eqref{eq:intro:FEITH}. He proves that a related function defines a cohomology class in $H^2(F,F)$, for $F=\Rr$ or $\Zz/p\Zz$. Several works connected to motives or polylogarithms have emphasized the role of the fundamental equation, for instance \cite{Cathelineau1996,Elbaz2002,Elbaz2015,Bloch2003}.}  

\subsection{Categories of observables}\label{sec:intro:infostr}

Information cohomology was introduced in \protect\cite{Baudot2015} considering presheaves on \emph{information structures}, that were either categories of partitions of a given measurable space   or categories of orthogonal decompositions of a Hilbert space. The partitions corresponded to atoms of the $\sigma$-algebras generated by measurable functions (classical observables) with finite range, and the orthogonal decompositions appeared as eigenspaces of self-adjoint operators (quantum observables) with finitely many different eigenvalues.

Inspired by \protect\cite{Gromov2012}, we wanted to approach measurements from a categorical viewpoint, describing directly the relations between their outputs and without presupposing the existence of an underlying probability space or Hilbert space.\protect\footnote{A similar approach is taken in \protect\cite{Matveev2018}.} A probability space is only necessary to represent a collection of observables by measurable functions with a common domain, as customary in ``classical'' probability theory (as opposed to ``quantum'').  However, the existence of such representation is not trivial: some collections of variables are contextual (see Section \ref{sec:models}) and therefore violate generalized Bell inequalities  \protect\cite[Prop.~III.1]{Abramsky2012}, which make them incompatible with such \emph{classical representations}. The sets of outputs can also be interpreted as the spectra of self-adjoint operators, in such a way that some contextual collections have \emph{quantum representations}.

In view of the foregoing, we introduce here a more general definition of information structure, that covers the classical and quantum cases at the same time and extends without modification to continuous random variables (see \protect\cite{Vigneaux2019-thesis}). This allows us to introduce a category of information structures and to treat the algebraic aspects of the theory (Section \ref{sec:topoi}) in a unified manner, once for all these cases. To attain this flexibility and generality, the definition decouples the combinatorial structure of joint measurements and the local models of the outputs of each individual measurement.

Let  $\MeasSets_{\mathrm{surj}}$ be the category of measurable spaces and measurable surjections between them. 

\begin{defi}\label{def:info_structure}
A \emph{conditional meet semilattice} is a poset\footnote{A partially ordered set (poset) is set $C$ with a binary relation $\leq$ that is reflexive, antisymmetric and transitive. Equivalently, it is a small category $\cat C$ such that:
\begin{enumerate}
\item for any pair of objects $A$, $B$, there is at most one morphism from $A$ to $B$, and
\item if there is a morphism from $A$ to $B$ and a morphism from $B$ to $A$, then $A=B$.
\end{enumerate} We move freely between both descriptions.   
} that satisfies the following property:
\begin{equation}\label{conditional_wedge}
    \begin{minipage}{0.8\textwidth}
    for any $X,Y,Z\in \Ob \Sinf$, if $Z\to X$ and $Z\to Y$, then the categorical product $X\wedge Y$ exists.
    \end{minipage}
  \end{equation}
  It is \emph{unital} whenever it has a terminal object, denoted $\Us$. 
  
An \keyt{information structure} is a pair $(\Sinf,  \sheaf M)$, where $\Sinf$ is a unital conditional meet semilattice and $  \sheaf M:\Sinf \to \MeasSets_{\mathrm{surj}}$ is a functor\footnote{Given  a functor $\sheaf F :\Sinf \to \Sets$, we denote its value at $X\in \Ob \Sinf$ by $\sheaf F(X)$ or $\sheaf F_X$.} (say  $ \sheaf M_X =  ({\sheaf E}_X,\salg B_X)$, for each $X\in \Ob \Sinf$)   that satisfies:
  \begin{enumerate}
 \item ${\sheaf E}_\Us \cong \{\ast\}$, with the trivial $\sigma$-algebra;
 \item for every $X\in \Ob \Sinf$ and any $x\in {\sheaf E}_X$, the $\sigma$-algebra $\salg B_X$ contains  the singleton $\{x\}$;
 \item\label{wedge_coondition} for every diagram 
 $
 \begin{tikzcd}
  X & X\wedge Y \ar[l,swap, "\pi"] \ar[r, "\sigma"] & Y 
 \end{tikzcd}
 $
 the measurable map $${\sheaf M}_{X\wedge Y} \hookrightarrow {\sheaf M}_X\times {\sheaf M}_Y, z\mapsto (x(z),y(z)):= (\sheaf M\pi(z),\sheaf M \sigma(z))$$ is an injection.
\end{enumerate}
\end{defi}

The objects of the conditional meet semilattice $\Sinf$ stand for observables and the arrows encode the relation of \emph{refinement} between them (think of refinements of $\sigma$-algebras or orthogonal decompositions). The terminal object is ``certainty'', the meet $X\wedge Y$ represents the joint measurement of $X$ and $Y$, and condition \eqref{conditional_wedge} accommodates the impossibility of doing some joint measurements. For instance, in quantum mechanics, it is only possible to jointly measure $X$ and $Y$ if they commute, in which case the observable $(X,Y)$ induces an orthogonal decomposition of the Hilbert space that refines the decompositions induced by $X$ and $Y$. In turn, the functor $\sheaf M$  represents the possible outputs of each observable. A refinement $\pi:X\to Y$ translates into a surjection $\pi_*\equiv \sheaf M\pi:\sheaf M_X\to \sheaf M_Y$ that induces an injection at the level of the algebras of events $\pi^*:\salg B_Y \to \salg B_X$ that maps $A$ to $\sheaf M\pi^{-1}(A)$;\footnote{\label{foot:surjections}Proof: If $\pi^*(A) = \pi^*(B)$, then $\pi^{-1}(A\Delta B) = \emptyset$, where $\Delta$ is the symmetric difference. Since $\sheaf M\pi$ is surjective, this implies that $A\Delta B= \emptyset$ i.e. $A=B$. Of course, one could relax the surjectivity introducing ideals of negligible sets (e.g. through a reference measure) and asking  the preimage under $\sheaf M\pi$ of any negligible set to be negligible, but---at least in the discrete case---there is no loss of generality in supposing directly that $\sheaf M\pi$ is a surjection.}compare with the \emph{extensions} of probability spaces discussed in \cite[p.~3]{Tao2012}. The set ${\sheaf E}_{X\wedge Y}$ represents the possible outputs of the joint measurement $X\wedge Y$, hence it can be identified with a subset of ${\sheaf E}_X\times {\sheaf E}_Y$ as in Section \ref{sec:intr:entropy}. When convenient, we use the notations common in probability theory: $\{X=x\}$  means ``the element $x$ contained in ${\sheaf E}_X$'' and $\{X=x,Y=y\}$ should be interpreted as \emph{the} element $z$ of ${\sheaf E}_{X\wedge Y}$ mapped to $x$ by ${\sheaf E}_{X\wedge Y}\to {\sheaf E}_X$ and to $y$ by ${\sheaf E}_{X\wedge Y}\to {\sheaf E}_Y$ (if such $z$ does not exist, $\{X=x,Y=y\}= \emptyset$).

In Section \ref{sec:general_structure} we also define the morphisms between information structures and prove that the category $\InfoStr$ thus obtained has countable products and arbitrary coproducts (Proposition \ref{prop:products_and_coproducts}). 

Probability laws come as a functor $\FProb:\Sinf \to \Sets$  that associates to each $X\in \Ob\Sinf$ the set $\FProb_X$ of measures $P$ on $\sheaf M_X$ such that $P({\sheaf E}_X) =1$. Each arrow $\pi:X\to Y$ induces a measurable surjection $   \sheaf M:\sheaf M_X\to \sheaf M_Y$, and $\FProb \pi:\FProb_X\to \FProb_Y$ is defined to be the push-forward of measures: for every $B\in \salg B_Y$,
\begin{equation}\label{defi_marginalization}
(\FProb \pi(P))(B) = P(  {\sheaf M\pi}^{-1}(B)).
\end{equation}
This operation is called \emph{marginalization}. We write $\pi_*$ or  $Y_*$ instead of $\FProb \pi$, if there is no risk of ambiguity;  this notation is compatible with that of Section \ref{sec:intr:entropy}.

\example[Simplicial information structures]\label{ex:simplicial_structures}
 If $I$ is any set, let $\cat{\Delta(I)}$ be the poset of its finite subsets, with an arrow $A\to B$ whenever $B\subset A$.  A simplicial subcomplex of $\cat{\Delta(I)}$ is a full subcategory $\cat K$ such that, for any given object of $\cat K$ (``a cell''), all its subsets are also objects of $\cat K$ (``faces''). Given a collection $\{({\sheaf E}_i,\salg B_i)\}_{i\in I}$ of masurable spaces, let $\sheaf M:\cat{\Delta(I)}\to \Sets$ be the functor that associates to each  $A\subset I$ the set ${\sheaf E}_A:=\prod_{i\in A} {\sheaf E}_i$ with the product $\sigma$-algebra $\salg B_A:= \bigotimes_{i\in A} \salg B_i$, and to each arrow in $\cat{\Delta(I)}$ the corresponding canonical projector. The pair $(\cat K,   \sheaf M|_{\cat K})$ is a \emph{simplicial information structure}.\footnote{It is worth noting that \emph{abelian} (co)presheaves on some complex $\cat K$ are cellular (co)sheaves in the sense of \cite{Curry2013}.}  
 
 A particular case of this construction---such that all the ${\sheaf E}_i$ are equal---appears in  \cite{Abramsky2011sheaf}, which introduces a sheaf-theoretic treatment of nonlocality and contextuality. There, the elements of $I$ are called \emph{measurements}, the minimal objects in the poset $\cat K$ are called \emph{maximal measurement contexts} (without loss of generality, it is supposed that the $0$-skeleton of $\cat K$ is $I$ itself), and the sheaf $\sheaf M:\cat{\Delta(I)}\to \Sets$ is called \emph{sheaf of events}. The follow-up article \cite{Abramsky2015} introduces a more general notion of possible events, allowing any subfunctor $\sheaf N:\cat{\Delta(I)}\to \Sets$ of $\sheaf M$ that is ``flasque beneath the cover''---which means that $\sheaf N\pi$ is surjective  for any arrow $\pi$ in $\cat K$---and such that any section of  $\sheaf N|_{\cat K}$ induces an element of $\sheaf N(I)$. Our definition of information structure further generalizes these ideas: it can be applied to nonsimplicial settings (see the examples in Section \ref{sec:general_structure}) and it does not make reference to global facts (which are treated in the context of representations, see Section \ref{sec:models}).\footnote{Let $\FProb_{\cat K}$ be the functor of probabilities on $\sheaf M|_{\cat K}$. The sections of $\FProb_{\cat K}$ are the elements of $\Gamma(\cat K,\FProb_{\cat K}):=\Hom_{[\cat K,\Sets]}(\ast,  \FProb_{\cat K})$, where $[\cat K,\Sets]$ is the category of $\Sets$-valued functors on $\cat K$ and $\ast$ is the functor that associates to each $X\in \Ob \cat K$ a singleton. An element $s\in \Gamma(\cat K,\FProb_{\cat K})$  is a  collection of probabilities that are mutually compatible under marginalizations; this appears in the literature as  \protect\emph{pseudo-marginals} \protect\cite{Vontobel2013} or \protect\emph{no-signaling empirical models} \cite{Abramsky2011sheaf}, among other names. Remark that, in the simplicial case, $\sheaf M$ and the sheaf  $\FProb$ of probabilities on it are naturally defined on the whole category $\cat{\Delta(I)}$, which can be seen as a larger geometrical space that contains $\cat K$. The \emph{marginal problem} consists in determining when a pseudo-marginal $s$ on $\cat K$ can be extended to a section $\tilde s$ of $\FProb$ on $\cat{\Delta(I)}$ that coincides with $s$ over each $A\in \Ob \cat K$; this extension $\tilde s$ is meant to represent a joint state of the observables indexed by $I$, compatible with the known local interactions. It is well known that such extension does not always exist: this correspond to \emph{frustration} in statistical mechanics \protect\cite{Pelizzola2005, Matsuda2001}, \protect\emph{contextuality} in quantum mechanics, and paradoxes in logic \protect\cite{Abramsky2015,Abramsky2012,Fritz2013}. The problem also appears in the context of \protect\emph{graphical models}, where the belief propagation algorithm converges to a pseudo-marginal on the smallest simplicial subcomplex that contains all the factors of the model (we are identifying here a factor node in a graphical model with its boundary, which is a subset of  the variable nodes), see \protect\cite{Pelizzola2005} and \cite[Ch.~9]{Mezard2009}. A sheaf-theoretic treatment of the marginal problem was introduced in \protect\cite{Abramsky2011sheaf}. The determination of eventual connections between information cohomology and their cohomology of contextuality is an exciting open problem.} 
\endexample

An information structure $(\Sinf,\sheaf M)$ is called finite if all the measurable sets $({\sheaf E}_X,\salg B_X)$ in the image of $\sheaf M$ are finite. Section \ref{sec:models} studies the conditions under which the observables of a finite information structure $(\Sinf,\sheaf M)$ can be represented as measurable functions on a unique sample space $(\Omega, \salg F)$. A necessary and sufficient condition is the existence of a global section $s(x)$ of $ {\sheaf E}$ compatible with any given  value  $x\in {\sheaf E}_X$ assigned to any observable $X$. For the sake of completeness, we also define quantum representations.

\subsection{Cohomological characterization of $\alpha$-entropies}\label{sec:intro:cohomology}

Let $\Sinf$ be a conditional meet semilattice, and $\sheaf S$ be the presheaf of monoids that maps each $X\in \Ob \Sinf$ to the  set $\Smon_X \equiv \Smon(X) := \set{Y\in \Ob \Sinf}{X\to Y}$ equipped with the product  $(Y,Z) \mapsto YZ:= Y\wedge Z$, and each arrow  $X\to Y$ in $\Sinf$ to the inclusion $\Smon_Y \hookrightarrow \Smon_X$. There is an associated presheaf of induced algebras $X\mapsto \Sring_X := \Rr[\Smon_X]$. The category of $\sheaf A$-modules---abelian presheaves with a natural action of $\Sring$---is abelian and has enough injective objects, hence it is possible to define right derived functors using standard tools in homological algebra \cite{Grothendieck1957, Weibel1994}. 

A particular example of $\sheaf A$-module is the space of probabilistic functionals that appear in Section \ref{sec:intr:entropy}; the $\Sring$-action was defined by  \eqref{eq:intro:action}. Another is the trivial $\Sring$-module $\Rr_{\Sinf}$, that associates to each object $X\in \Ob{\Sinf}$ the set $\Rr$ with trivial $\Sring_X$ action and to every arrow in $\Sinf$ the identity map.

\begin{defi}
The information cohomology of $\Sinf$ with coefficients in the $\Sring$-module  $\sheaf F$ is 
\begin{equation}
H^\bullet(\Sinf, \sheaf F) := \Ext^\bullet(\Rr_{\Sinf},   \sheaf F).
\end{equation} 
\end{defi}

Recurring to the (relative) bar resolution \cite[Ch.~IX]{MacLane1994}, we obtain a computable version of this cohomology. The bar construction gives a resolution   
\begin{equation}
\begin{tikzcd}
0 
& \Rr_\Sinf \ar[l] 
&   \sheaf B_0            \ar[l, "\epsilon", swap]  
&   \sheaf B_1            \ar[l, "\partial_1", swap]
&   \sheaf B_2            \ar[l, "\partial_2", swap]
& ...\ar[l, "\partial_3", swap]
\end{tikzcd},
\end{equation}
where each $\sheaf B_i$ is a \emph{relative} projective  $\Sring$-module (see the appendix). More explicitly,  for each $X\in  \Ob\Sinf$ and $ n\in \Nn$, the module $\sheaf B_n(X)$ is freely generated over $\Sring_X$ by  $\set{[X_1|...|X_n]}{X_1,...,X_n\in \Smon_X}$; in the case of $ \sheaf B_0(X)$, simply by the ``empty'' symbol $[\:]$. Proposition \ref{prop_bar_projective} proves that, due to the conditional existence of products in the definition of $\Sinf$, these $\{ \sheaf B_i\}_i$ are  projective objects in the category $\Mod(\Sring)$, which in turn implies that  $H^\bullet(\Sinf, \sheaf F)$ can be identified with the cohomology of the differential complex $(C^\bullet(\Sinf,\sheaf F),\delta)$, where $C^n(\Sinf, \sheaf F) := \Hom_{\Sring}( \sheaf  B_n,  \sheaf  F)$ and $\delta f := f\partial$. An element $f\in C^n(\Sinf,  \sheaf F)$ consists of several components $\{f_X:\sheaf B_n(X)\to \sheaf F(X)\}_{X\in \Ob \Sinf}$. Each map  $\delta\equiv \delta^n:C^{n}(\Sinf,  \sheaf F)\to C^{n+1}(\Sinf,  \sheaf F)$ is defined as follows: for each $X\in \Ob \Sinf$, $X_1,..,X_{n+1}\in \Smon_X$, and $f\in C^{n}(\Sinf,  \sheaf F)$,
\begin{multline}\label{eq:n-coboundary}
(\delta f)_X[X_1|...|X_{n+1}] = X_1.f_X[X_2|...|X_{n+1}] + \sum_{k=1}^{n} (-1)^k f_X[X_1|...|X_kX_{k+1}|...|X_n] \\+(-1)^{n+1}f_X[X_1|...|X_{n}].
\end{multline}
To simplify notation, we write $f_X[X_2|...|X_{n+1}]$ instead of $f_X([X_2|...|X_{n+1}])$, etc.

We introduce a particular family of $\Sring$-modules, made of probabilistic functionals.  Let $(\Sinf, {\sheaf E})$ be a finite information structure, and let $  \sheaf Q$ denote any subfunctor of  $\FProb$ stable under conditioning (\emph{adapted probability functor}). We represent every probability law by its density $P$ with respect to the counting measure, which can be seen as a vector in $\Rr^{\sheaf E_X}$, so that $\sheaf P_X$ is the standard (probability) simplex $\Delta(\sheaf E_X)$ in $\Rr^{\sheaf E_X}$.   Let $  \sheaf F(X)$ be the additive abelian group of measurable\footnote{If $(\Omega, \salg F)$ is a measurable space and $A$ is a subset of $\Omega$, then $\salg F' = \set{A\cap F}{F\in \salg F}$ is a $\sigma$-algebra of subsets of $A$, \emph{induced} by $\salg F$. For each $X\in \Ob \Sinf$, we identify $\sheaf P(X)$ with the standard simplex $\Delta(\sheaf E_X)$ equipped with the $\sigma$-algebra $\salg S$ induced by the Borel $\sigma$-algebra on $\Rr^{|{\sheaf E}_X|}$ (with its standard topology); equivalently, $\salg S$ is the Borel $\sigma$-algebra associated to the subspace topology on $\Delta(\sheaf E_X)$. Then $\sheaf Q(X)\subset \sheaf P(X)$ becomes a measurable space with the $\sigma$-algebra induced by $\salg S$.} real-valued functions on $  \sheaf Q(X)$ and, for any arrow $\pi:X\to Y$, let $  \sheaf F\pi:  \sheaf F(Y)\to   \sheaf F(X)$ be precomposition with marginalization: $  \sheaf F\pi(f) = f\circ \pi_*$. We obtain in this way a contravariant functor $  \sheaf F$ on $\Sinf$.

For each $Y\in \Smon_X$, $f\in   \sheaf F(X)$ and $P\in   \sheaf Q(X)$, define 
\begin{equation}
(Y.f)(P) = \sum_{\substack{y\in {\sheaf E}_Y\\ Y_*P(y) \neq 0 }} (Y_*P(y))^\alpha f(P|_{Y=y}).
\end{equation}
 
 This turns each $  \sheaf F(X)$ into an $\Sring_X$-module and this action is functorial, in such a way that $  \sheaf F$ becomes an $\Sring$-module  $\sheaf F_\alpha$ (for any  $\alpha>0$). We call \emph{probabilistic} the information cohomology with coefficients in some $\sheaf F_\alpha$. With reference to the cochain complex $(C^\bullet(\Sinf, \sheaf F_\alpha), \delta)$ introduced above, we can determine the following facts. 

A $0$-cochain is a collection  $\left\{f_X[\,]\right\}_{X\in \Ob S}$ that is local: for any $X$, $f_X[\,](P) = f_{\Us}[\,](1)$, so $f_X[\,]$ equals a constant $K\in \Rr$. The $0$-cochain $f$ is a $0$-cocycle if for any $X\to Y$ in $\Sinf$  
$ (\delta f)_X[Y] = Y.f_X[\,] - f_X[\,]$, which evaluated on a probability  $P\in   \sheaf Q(X)$ reads 
$$(\delta f)_X[Y](P)=\sum_{y\in {\sheaf E}_Y} (Y_*P(y))^\alpha K -K = \begin{cases} 0 & \text{ if }\alpha=1\\ 
K S_\alpha[Y](Y_*P) & \text{ otherwise} \end{cases}.$$
In other words:  every cochain is a $0$-cocycle if $\alpha=1$; there are no $0$-cocycles if $\alpha \neq 1$, but Tsallis entropy appears as $1$-coboundary multiplied by a global constant $K$.
 
 The $1$-cochains are characterized by collections of functionals $\left\{f[X]:\sheaf Q(X)\to \Rr\right\}_{X\in \Ob \Sinf}$, which takes into account the naturality of $f$: $f_Y[X](P)=f_X[X](X_*P)=:f[X](X_*P)$.  The $1$-cocycles additionally satisfy
 \begin{equation}
 0=X.f[Y]-f[XY]+f[X]
 \end{equation}
as functions on $\sheaf Q(XY)$, where marginalizations are implicit. As explained in Section \ref{sec:intr:entropy}, this equation and its analogue with $X$ and $Y$ permuted imply that $f[\cdot] = K_{XY} \ent \alpha[\cdot]$, for a constant  $K_{XY}\in \Rr$. This holds as long as $\sheaf Q_{XY}$ contains enough probabilities: a precise sufficient condition is stated in the definition of nondegeneracy for the product of two observables (Definition \ref{def:non-degenerate-probabilistic}). Despite being quite involved, this definition is one of the main contributions of this article. It is instrumental to write a more explicit proof of \cite[Thm.~1]{Baudot2015} and its generalization, Theorem \ref{H1-non-degenerate}. If an observable $Z$ can be written as a nondegenerate product, we say that it is nontrivially reducible. 

Remind that a category $\cat J$ is connected if any two objects $j,k\in \cat J$ can be joined by a finite sequence of arrows 
\begin{equation}
j=j_0 \leftarrow j_1 \rightarrow j_2 \leftarrow \cdots \leftarrow j_{2n-1} \rightarrow j_{2n} = k.
\end{equation}
Every category $\cat J$ is a disjoint union (coproduct in $\cat{Cat}$) of connected categories $\cat J_{k}$, called \emph{connected components} \cite[p.~90]{MacLane1998}. We denote by $\pi_0(\cat J)$ the set of connected components of a small category $\cat J$.


\begin{theorem}\label{H1-non-degenerate}
Let $(\Sinf,{\sheaf E})$ be a finite information structure, and let $ \sheaf Q$ be an adapted probability functor such that all marginalization maps are surjective. Denote by $\Sinf^\ast$ the full subcategory of $\Sinf$ generated by $\set{X\in \Ob \Sinf}{|\sheaf E_X| > 1}$. Suppose that for every $X\in \Ob \Sinf^*$ there exists a nontrivially reducible object $Z\in \Ob \Sinf$ such that $Z\to X$. Then, there is a linear isomorphism $\varphi: \Rr^{\pi_0(\Sinf^*)}\to Z^1(\Sinf,\sheaf F_\alpha( \sheaf Q))$ that maps $\Lambda=(\lambda_{\cat C})_{\cat C\in \pi_0(\Sinf^*)}$ to the natural transformation $S^\Lambda_\alpha$ defined as follows: for every $X\in \Ob \cat S$ and $Y\in \Smon_X$, 
\begin{equation}
(S_\alpha^{\Lambda})_X[Y] = \begin{cases}\lambda_{\cat C}(S_\alpha)_X[Y] & \text{if } Y \in \Ob\cat C \\
0 & \text{if } Y \notin \Ob \cat S^*
\end{cases}.
\end{equation}
Under this isomorphism,  $\delta C^0(\Sinf, \sheaf F_1(\sheaf Q))$ corresponds to $\langle 0 \rangle \subset \Rr^{\pi_0(\Sinf)}$, whereas for $\alpha \neq 1$,  $\delta C^0(\Sinf, \sheaf F_\alpha(\sheaf Q))$ corresponds to the diagonal $\Delta  \subset \Rr^{\pi_0(\Sinf)}$. Hence $H^{1}(\Sinf, \sheaf F_1(\sheaf Q))\cong \Rr^{\pi_0(\Sinf^*)}$ and $H^{1}(\Sinf, \sheaf F_\alpha(\sheaf Q))\cong \Rr^{\pi_0(\Sinf^*)}/\Delta$ when $\alpha \neq 1$.
\end{theorem}

One may say that a pair $(\Sinf, \sheaf Q)$ is \emph{strongly connected} if  $\Sinf^*$ is connected and every  object can be refined by a nontrivially reducible one. Shannon entropy is to some extent analogous to the fundamental class of an orientable connected manifold: a generator of the one-dimensional $H^1(\Sinf, \sheaf F(\sheaf Q))$ for a strongly connected pair $(\Sinf,\sheaf Q)$.

For a simplicial structure $(\cat K,   \sheaf M)$, the theorem says that
\begin{equation}
H^1(\Sinf,   \sheaf F_1(\FProb)) \cong \Rr^{\beta_0(\cat K)} \quad\text{ and }\quad  H^1(\Sinf,   \sheaf F_\alpha(\FProb)) \cong \Rr^{\beta_0(\cat K)-1} \text{ when }\alpha \neq 1, 
\end{equation}
where $\beta_0(\cat K)$ is the $0$-th Betti number of (the geometric realization of) $\cat K$.  Higher cohomology groups might be linked to the higher Betti numbers.

 When some of the minimal objects of $\Sinf$ is irreducible (it cannot be written as a nontrivial product), the group $H^1(\Sinf, \sheaf F_\alpha( \sheaf Q))$ has infinite dimension; Proposition \ref{irreducible_element_infty}  
 makes this precise.  In  \ref{sec:determination_H1} we also study other pathological examples. 

At the end, we give an interpretation of $H^0$, $H^1$ and $H^2$ in terms of invariant sections, crossed homomorphism, and extensions, respectively, following the classic arguments about Hochschild cohomology.

\section{The category of information structures}\label{sec:info_structures}

\subsection{Terminology and examples}\label{sec:general_structure}
An information structure $(\Sinf, \sheaf M)$---see Definition \ref{def:info_structure}---is said to be  \keyt{bounded} if the poset $\Sinf$ has finite height. It is \keyt{finite} if all the sets ${\sheaf E}_X$ are finite, in which case ${\sheaf E}_X$ corresponds to the atoms of $\salg B_X$ and the algebra can be omitted from the description. We denote a finite structure by $(\Sinf, {\sheaf E})$, where ${\sheaf E}$ is a covariant functor from $\Sinf$ to $\Sets$. The cohomological computations in Section \ref{sec:classical_info_cohomology} concern finite structures, but the general constructions in Section \ref{sec:topoi}---among them the definition of information cohomology---do not require this hypothesis. In fact, they only depend on the combinatorial object $\Sinf$.

\begin{example}[Concrete structures] \label{ex:classical_structure} The motivating example for the theory is the original version of information structures introduced in \cite{Baudot2015}. 

Given a set $\Omega$, let $\ObsFin(\Omega)$ be the category \emph{finite observables}; the objects of this category are finite partitions of $\Omega$, and there is an arrow $X\to Y$ whenever $X$ refines $Y$. In this case, $X$ discriminates better between the \emph{configurations} $w\in \Omega$. The category $\ObsFin(\Omega)$ has a terminal object: the trivial partition $\Us:=\{\Omega\}$. When $\Omega$ is finite, it also has an initial object: the partition by points, that we denote by $\bot$.  The categorical product $X\times Y$ of two partitions $X$ and $Y$ is the coarsest partition that refines both.  This product is commutative, associative, idempotent and unitary ($\Us\times  X = X$).  

A classical \emph{information structure} in the sense of \cite{Baudot2015} is a full subcategory $\Sinf$ of $\ObsFin(\Omega)$ such that 
\begin{itemize}
\item $\Us\in \Ob{\Sinf}$;
\item for any $X, Y, Z$ in $\Ob\cat{S}$, if $X \to Y$ and $X\to Z$, then $Y\times Z$ belongs to $\Sinf$.
\end{itemize}
We call $\Sinf$ a \emph{concrete structure}.  If $\square:\ObsFin(\Omega)\to \Sets$ denotes the ``forgetful'' functor  that maps the partition $X=\{A_1,...,A_n\}$ to the set $\{A_1,...,A_n\}$ and each arrow $X\to Y$ in $\ObsFin(\Omega)$ to the unique surjective map $\square\pi:{\sheaf E}_X\to {\sheaf E}_Y$ such that  $B=\cup_{A\in {\sheaf E}\pi^{-1}(B)} A$ for any $B\in {\sheaf E}_Y$, the pair $(\Sinf,\square)$ is a finite information structure according to Definition \ref{def:info_structure}.

Concrete structures turn out to be too restrictive. For instance, Baudot and Bennequin associate to any finite indexed collection  $\Sigma=(S_1,...,S_n)$ of partitions of $\Omega$ a ``simplicial structure'' $\Sinf(\cat K)$: a subcategory of $\ObsFin(\Omega)$ that contains $\prod_{i\in A} S_i$ for any object $A$ of a simplicial subcomplex $\cat K$ of the abstract simplex $\Delta(\{1,...,n\})$---see the notation introduced in Example \ref{ex:simplicial_structures}; by convention, the empty product gives the trivial partition. Such construction does not necessarily give an information structure (in their sense). For example:  if $n=3$, $\Omega=\{0,1\}^2$, $S_i$ is the partition induced by the projection on the $i$-th component ($i=1,2$), $S_3=\{\{(0,0)\},\{(0,0)\}^c\}$, and the maximal cells of $\cat K$ are $\{1,2\}$ and $\{3\}$, then  $S_1\times S_2$ is the atomic partition, that refines all the others, while some products (like $S_1\times S_3$) are not in $\Sinf(\cat K)$. In our framework, the category $\Sinf(\cat K)$ appears as a classical representation (cf. Section \ref{sec:models}) of the (generalized) information structure $(\cat K,{\sheaf E})$, where ${\sheaf E}:\cat K \to \Sets$ is given by ${\sheaf E}_{\{1\}}={\sheaf E}_{\{2\}} = {\sheaf E}_{\{3\}}=\{0,1\}$, ${\sheaf E}_{\{1,2\}} = {\sheaf E}_{\{1\}}\times {\sheaf E}_{\{2\}}$, the maps induced by the arrows in $\cat K$ being canonical projections. 
\end{example}

\begin{example}[Homogeneous structures]\label{ex:homogeneous}
Let $G$ be a locally compact, Hausdorff topological group.  Any collection $\mathfrak C$ of closed  subgroups of $G$ that contains $G$ and is conditionally closed under intersections (i.e. for any $M,N,O\in \mathfrak C$, if $N\subset M$ and $N\subset O$, then $M\cap O\in \mathfrak C$) defines a conditional meet semilattice $\Sinf$, whose arrows correspond to inclusions. Let $\sheaf M$ be the functor that associates to each subgroup $N$ the (Hausdorff) quotient space $G/N$ with the Borel $\sigma$-algebra induced by the quotient topology, and to each arrow $N\to M$ the canonical projection $\pi_{M,N}:G/N\to G/M$ that sends the coset $gN$ to $gM$. 
The information structures $(\Sinf, \sheaf M)$ obtained in this way are called \emph{homogeneous}, because each coset space $G/M$ is a homogeneous space for $G$.\footnote{Moreover, the  diagrams $\sheaf M(\Sinf)$ obtained for finite $G$ are exactly the \emph{minimal homogeneous diagrams} in  \cite[Sec.~2.7]{Matveev2018}, provided each quotient $G/N$ is equipped with the uniform measure. It is explained there that homogeneous diagrams approximate asymptotically any diagram of probability spaces.}

A particular example of this construction was introduced by the author in \cite[Ch.~11]{Vigneaux2019-thesis}, where $G=(\Rr^n,+)$. The resulting pairs $(\Sinf, \sheaf M)$ are called there \emph{Grassmannian structures} and play a key role in the computation of the information cohomology associated to continuous observables with gaussian laws.
\end{example}

\subsubsection{Relation with idempotent monoids} Recall that a monoid $(M,\cdot,e)$ is \emph{idempotent} if for all $m\in M$, $m\cdot m = m$. Any conditional meet semilattice $\Sinf$ induces a presheaf of idempotent monoids on it: for each $X\in \Ob\cat{S}$, set $\Smon_X:=\{Y\in \Ob\cat{S} \mid X\to Y\}$, with the monoid structure given by the product of observables in $\Sinf$: $(Z,Y)\mapsto ZY:= Z\wedge Y$; an arrow $X\to Y$ in $\Sinf$ induces an inclusion $\Smon_Y \hookrightarrow \Smon_X$. 

Furthermore, there is a well-known equivalence between idempotent monoids and meet semilattices with a terminal object. For a proof, see e.g. \cite[Prop.~2.1]{Connes2017}.
\begin{prop}\label{prop:equivalence_monoids_meetsl}
If $(M,\cdot, e)$ is an idempotent monoid, then the condition 
\begin{equation}\label{eq:order_monoid}
x\leq y \Leftrightarrow x\cdot y = x
\end{equation} defines a partial order on $M$ such that any two elements of $M$ have a meet and $e$ is the greatest element. 

Conversely, if $(E,\leq)$ is a poset with a greatest element in which any two elements $x,y\in E$ have a meet $x\wedge  y$, then $E$ endowed with the addition $(x,y)\mapsto x\wedge y$ is an idempotent monoid.

The two functors just introduced are inverses of each other.
\end{prop}
 Is there a counterpart to conditional meet semilattices with a terminal object in the theory of idempotent monoids? The following result serves as a partial answer. It involves \emph{upper sets} of an idempotent monoid: a subset $H$ of an idempotent monoid $M$---equipped with the partial order in \eqref{eq:order_monoid}---is called an \emph{upper set} if $h\in H$ and $h\leq m$ implies that $m\in H$. For example, the simplicial subcomplex $\cat K$ in Example \ref{ex:simplicial_structures} defines an upper set of $\cat{\Delta(I)}$, seen as an idempotent monoid according to Proposition \ref{prop:equivalence_monoids_meetsl}. 

\begin{prop}Let $(M,\cdot,e)$ be an idempotent monoid, and $\cat M$ its associated poset (seen as a category). 
The full subcategory of $\cat M$ defined by any nonempty upper set $H$ of $M$ is a unital conditional meet semilattice. 
\end{prop}
\begin{proof}
First, $e\in H$, because $e$ is greater than any element of $H$. Second, if $x,y,z$ be elements of $ H$ such that $z\leq x$ and $z\leq y$, then $z\leq x\wedge y$ in virtue of the universal property of $\wedge$ in $\cat M$, which in turn implies that $x\wedge y\in H$ (by definition of upper set). 
\end{proof} 
Based on this result, new examples of information structures may arise in connection with \emph{idempotent mathematics} \cite{Litvinov2007} and algebra ``over $\mathbb{F}_1$'' \cite{Connes2017}. 

\subsection{Morphisms and (co)products}

Given categories $\cat A, \cat B$ and $\cat C$, and a functor $\xi:\cat A\to \cat B$, define the pullback 
$\xi^*:[\cat B, \cat C]\to [\cat A,\cat C]$  by $\sheaf F \mapsto \sheaf F\circ \phi$. It commutes with limits and colimits \cite[Sec.~I.5]{Artin1972}.

 \begin{defi}\label{def:morph_str}
 A morphism $\xi:\Sinf_1\to \Sinf_2$ of conditional meet semilattices is a functor (i.e. a monotone map) with the following property: if $X\wedge Y$ exists in $\Sinf_1$, then $\xi(X\wedge Y) = \xi(X) \wedge \xi(Y)$.
 
A morphism  $\phi:(\Sinf,  \sheaf M) \to (\Sinf',  \sheaf M')$ between information structures is a pair $\phi=(\phi_0,\widehat \phi)$ such that $\phi_0$ is a morphism of conditional meet semilattices that maps $\Us_{\Sinf}$ to $\Us_{\Sinf'}$, and $\widehat \phi:  \sheaf M\Rightarrow   \sheaf \phi_0^* M' $ is a natural transformation. If there is no risk of ambiguity, we write $\phi$ instead of $\phi_0$.

 Given $\phi:(\Sinf,   \sheaf M)\to (\Sinf',  \sheaf M')$ and $\psi:(\Sinf',  \sheaf M') \to  (\Sinf'',  \sheaf M'')$, their composition $\psi \circ \phi$ is defined as $(\psi_0 \circ \phi_0, \widehat \psi \circ \widehat \phi:  \sheaf M\Rightarrow   \sheaf \phi_0^* \psi_0^* M'')$. 
 
 We denote by $\InfoStr$ the category of information structures obtained in this way.
\end{defi} 

 Note that, if $X\wedge Y$ exists, then $\xi(X\wedge Y) \to \xi(X)$ and $\xi(X\wedge Y) \to \xi(X)$, and thus the product $\xi(X) \wedge \xi(Y)$ exists too, in virtue of Definition \ref{def:info_structure}. 
 
 A morphism of information structures is a particular case of morphism of $\MeasSets_{\mathrm{surj}}$-valued covariant diagrams \cite[Def.~3.1]{DeSilva2019}. We want $\phi_0$ to respect the unit and the products, so that it induces a morphism between the corresponding  presheaves of  idempotent monoids, see Proposition \ref{prop:chain_map}.

The preceding definition is one of the main motivations for our generalized setting. In fact, one could imagine a correspondence between the partitions of two concrete structures (Example \ref{ex:classical_structure}) defined on different sample spaces, but in which category would that correspondence take place? Since we eliminated the explicit reference to the sample space in our definition of information structure, the introduction of morphisms becomes straightforward.  This allows the computation of products and coproducts. The connection to the sample spaces is not completely lost, but reformulated in the language of representations: as a consequence of Proposition \ref{prop:prod_and_coprod_of_models}, if $\Sinf_i$ is a concrete structure on $\Omega_i$ ($i=1,2$), then the objects of $\Sinf_1\times \Sinf_2$ can be identified with partitions of $\Omega_1\times \Omega_2$, as one would expect. 


\begin{prop}\label{prop:products_and_coproducts}
The category  $\InfoStr$ has countable products and arbitrary coproducts. 
\end{prop}
\begin{proof} Let $\mathbf 0$ be the category that has $\Us$ as the only object and $\id_{\Us}$ as the only morphism, and let $\sheaf M_{\mathbf 0}$ be the functor that associates to $\Us$ the set $\{\ast\}$ equipped with the atomic $\sigma$-algebra. Clearly $(\mathbf 0,\sheaf M_{\mathbf 0})$ is initial and terminal in the category $\InfoStr$, hence it corresponds to the empty product and coproduct respectively. 

\textbf{Nonempty products: }
Given information structures  $(\Sinf_i,\sheaf M_i)$ indexed by $i$ in an arbitrary set $I$, we introduce first the ordinary categorical product $\Sinf=\prod_{i\in I} \Sinf_i$: its objects are $I$-tuples $\langle X_i\rangle_{i\in I}$ with $X_i\in\Ob\cat{S_i}$ for each $I\in I$; there is an arrow  $\langle \pi_i\rangle_{i\in I}:\langle X_i\rangle_{i\in I} \to \langle Y_i\rangle_{i\in I}$ whenever $\pi_i :X_i\to Y_i$ in $\Sinf_i$ for each $i\in I$. Then a functor  $\sheaf M:\Sinf \to \MeasSets_{\mathrm{surj}}$ is defined as follows: for each $X=\langle X_i\rangle_{i\in I}\in \Ob \Sinf$ the measurable space $\sheaf M(X)$ is the set ${\sheaf E}(X) := \prod_{i\in I} {\sheaf E}_i(X_i)$ equipped with the product $\sigma$-algebra  $\salg B(X) := \bigotimes_{i\in I} \salg B_i(X_i)$, which is the smallest $\sigma$-algebra that makes every canonical projection $\widehat{ p^i}_{\langle X_i\rangle_{i\in I}}:{\sheaf E}(\langle X_i \rangle_{i\in I}) \to \sheaf E_i(X_i)$ measurable \cite[Sec.~5.1]{Cohn2013}; at the level of morphisms, $\sheaf M(\langle \pi_i \rangle_{i\in I}):= \prod_{i\in I} \sheaf M_i(\pi_i)$, which comes from the product in $\Sets$. 

The pair $(\Sinf, \sheaf M)$ is an information structure. It is easy to verify that $\Sinf$ is a poset with terminal object 
$\langle \Us_{\Sinf_i}\rangle_{i\in I}$.
 The conditional existence of products also holds: if  $\langle X_i\rangle_{i\in I}$, $\langle Y_i\rangle_{i\in I}$ and $\langle Z_i\rangle_{i\in I}$ are objects of $\Sinf$ such that $\langle X_i\rangle_{i\in I} \to \langle Y_i\rangle_{i\in I}$ and  $\langle X_i\rangle_{i\in I} \to \langle Z_i\rangle_{i\in I}$, then for every $i\in I$,
$\begin{tikzcd}
Y_i & X_i \ar[l, swap, "\pi_{Y_i}"] \ar[r, "\pi_{Z_i}"] & Z_i \end{tikzcd}$ in $\Sinf_i$, which in turn implies that $Y_i \wedge Z_i$ exists in $\Sinf_i$ by definition of conditional meet semilattice; the reader can verify that
$$ \langle Y_i\rangle_{i\in I} \wedge \langle Z_i\rangle_{i\in I} = \langle Y_i\wedge Z_i\rangle_{i\in I} .$$
The functor $\sheaf M$ also has the desired properties. It is clear that   ${\sheaf E}(\langle \Us_{\Sinf_i} \rangle_{i\in I}) \cong \{ \ast\}$. If $I$ is countable, then for any $(x_i)_{i\in I}\in \sheaf E(\langle X_i\rangle_{i\in I})$ the singleton  $\{(x_i)_{i\in I}\}$  belongs to $\salg B(\langle X_i\rangle_{i\in I})$, because it can be written as a countable intersection $\bigcap_{i\in I} (\widehat{ p^i}_{\langle X_i\rangle_{i\in I}})^{-1}(x_i)$. Finally, when $\sheaf M$ is applied to the product $\langle Y_i\rangle_{i\in I} \wedge \langle Z_i\rangle_{i\in I}$  and its projections, one gets
$$
\begin{tikzcd}[column sep=huge]
{\sheaf M\langle Y_i\rangle}_{i\in I} & {\sheaf M\langle Y_i\wedge Z_i\rangle}_{i\in I} \ar[l, swap, "{\sheaf M \langle \pi_{Y_i}\rangle_{i\in I}}"] \ar[r,"{\sheaf M \langle \pi_{Z_i}\rangle_{i\in I}}"]  & {\sheaf M \langle Z_i\rangle_{i\in I}}.
\end{tikzcd}
$$
The map $$ \sheaf M \langle \pi_{Y_i}\rangle_{i\in I} \times {\sheaf M \langle \pi_{Z_i}\rangle_{i\in I}}: \sheaf M\langle Y_i\wedge Z_i\rangle_{i\in I} \to \sheaf M\langle Y_i\rangle_{i\in I} \times \sheaf M \langle Z_i\rangle_{i\in I}$$
 is injective, because for any  $(y_i)_{i\in I} \in {\sheaf E}(\langle Y_i\rangle_{i\in I})$ and $(z_i)_{i\in I} \in {\sheaf E}(\langle Z_i\rangle_{i\in I})$, the elementary properties of set operations imply that 
\begin{align*}
 (\sheaf M \langle \pi_{Y_i}\rangle_{i\in I} & \times \sheaf M \langle \pi_{Z_i}\rangle_{i\in I})^{-1}((y_i)_{i\in I},(z_i)_{i\in I}) \\
 \qquad &=({\sheaf M\langle \pi_{Y_i}\rangle_{i\in I}})^{-1}((y_i)_{i\in I}) \cap ({\sheaf M \langle \pi_{Z_i }\rangle_{i\in I}})^{-1}((z_i)_{i\in I}) \\
\qquad &= \left\{ \prod_{i\in I} {\sheaf M_i\pi_{Y_i}}^{-1}(y_i) \right\}\cap \left\{  \prod_{i\in I} {\sheaf M_i\pi_{Z_i}}^{-1}(z_i)  \right\}  & \text{(by def. of }\sheaf M\text{)}  \\
\qquad &= \prod_{i\in I} \left\{ {\sheaf M_i\pi_{Y_i}}^{-1}(y_i) \cap {\sheaf M_i \pi_{Z_i}}^{-1}(z_i)  \right\},
\end{align*}
and the cardinality of each factor in the last expression is at most $1$. 

For each $i\in I$, we introduce a morphism of information structures $p^i:(\Sinf,\sheaf M) \to (\Sinf_i,\sheaf M_i)$ such that ${p^i}_0$ maps each object or morphism $\langle A_i \rangle_{i\in I}$ to $A_i$, and $$\widehat{ p^i}_{\langle X_i\rangle_{i\in I}}: \prod_{i\in I} \sheaf M_i(X_i)\to \sheaf M_i(X_i)$$ is the canonical projection (which is measurable by definition of the product $\sigma$-algebra, see above).  We claim that $\Sinf$, with the projections $p^i$ just introduced, 
is the product of $(\Sinf_i,\sheaf M_i)_{i\in I}$  in $\InfoStr$, written $\prod_{i\in I} (\Sinf_i,\sheaf M_i)$, unique up to unique isomorphism (we also use the symbol $\times$ for finite products). In fact, given an $I$-cone 
$\{f^i: (\cat R,\sheaf F) \to (\Sinf_i, \sheaf M_i)\}_{i\in I}$ in $\InfoStr$ (where $I$ is seen as a discrete category), define $\langle f^i\rangle_{i\in I} : (\cat R,\sheaf F) \to (\Sinf,\sheaf M)$ by
\begin{align*}
(\langle f^i\rangle_{i\in I})_0 : & \cat R \to \Sinf \\
& R \mapsto \langle {f^i}(R) \rangle_{i\in I}  
\end{align*}
for any object or morphism $R$; for any $X\in \Ob \cat R$, the surjection $$\widehat{\langle f^i\rangle_{i\in I}}(X):\sheaf F(X) \to \sheaf M(\langle f^i(X) \rangle_{i\in I} )=\prod_{i\in I} \sheaf M_i(f^i(X))$$ is the map $\langle \widehat {f^i}_{X} \rangle_{i\in I}$ induced by the $I$-cone $\{\widehat{ f^i}_X: \sheaf F(X)\to \sheaf M_i(f^i(X))\}_{i\in I}$ in $\cat{Sets}$, in such a way that $p^i \circ \langle f^i\rangle_{i\in I} = f^i$ for all $i\in I$.

 \textbf{Nonempty coproducts: } Given information structures $\{(\Sinf_i,\sheaf M_i)\}_{i\in I}$, define a category $\Sinf$ such that 
 \begin{itemize}
 \item $\Ob\cat{S}=\bigsqcup_{i\in I} \Ob\cat {\Sinf_i} / \sim$, where $\sim$ is the smallest equivalence relation such that $\Us_{\Sinf_i}\sim \Us_{\Sinf_j}$ for all $i,j\in I$;
 \item $A\to B$ in $\Sinf$ if and only if $A\to B$ in $\Sinf_i$ for some $i$.
 \end{itemize} Let  $\sheaf M:\Sinf\to \Sets$ be the functor that coincides with $\sheaf M_i$ on $\Sinf_i$. The pair $(\Sinf,\sheaf M)$ is an information structure: the properties in Definition \ref{def:info_structure} are verified locally on each $\Sinf_i$.

 Injections $j^i : \Sinf_i \to \Sinf$ are defined in the obvious way: $j^i_0(A) = A$ for $A\in \Ob\cat{S_i}$ or $A\in\Hom(\Sinf_i)$, and the mappings ${\widehat j^i}_X$ are identities. If $\{f^i:(\Sinf_i,\sheaf M_i)\to (\cat R,\sheaf F)\}_{i\in I}$ is an $I$-cocone,  define 
\begin{align*}
(\langle f^i\rangle_{i\in I})_0 : &   \Sinf \to \cat R \\
& A   \mapsto f^i(A)  \text{ if } A\in \Ob\cat{S_i} \text{ or } A\in \Hom(\Sinf_i)
\end{align*}
and, if $X\in \Ob\cat{\Sinf_i}$, set $(\widehat{\langle f^i\rangle_{i\in I}})_X=\widehat{f^i}_X$. By construction, $\langle f^i\rangle_{i\in I} \circ j^i = f^i$. Therefore, $(\Sinf,\sheaf M)$ is the coproduct of $\{(\Sinf_i,\sheaf M_i)\}_{i\in I}$ in $\InfoStr$, denoted $\coprod_{i\in I} (\Sinf_i,\sheaf M_i)$ (we also use  $\sqcup$ for finite coproducts), which is unique up to unique isomorphism.
\end{proof}

\begin{remark} 
If $(\Sinf_1,\sheaf M_1)$ and $(\Sinf_2,\sheaf M_2)$ are bounded structures, their product and coproduct are bounded too. In fact, if the height  of the poset $\Sinf_i$ is $N_i$ ($i=1,2$), then the height of $\Sinf_1\times \Sinf_2$ is $N_1+N_2$ and that of $\Sinf_1\sqcup \Sinf_2$ equals $\max(N_1,N_2)$. Similarly, if both structures are finite, their product and coproduct is finite too.
\end{remark}

\begin{remark}
If each measurable space $({\sheaf E}(X),\salg B(X))$ appearing in $\Sinf_1$ and $\Sinf_2$ verifies that ${\sheaf E}(X)$ is second countable topological space and $\salg B(X)$ is its Borel $\sigma$-algebra, then each algebra $\salg B(X_1) \otimes \salg B(X_2)$ on  ${\sheaf E}(X_1)\times {\sheaf E}(X_2)$ equals the Borel $\sigma$-algebra on this space \cite[Prop.~9.1]{Vigneaux2019-thesis}.  
\end{remark}


One could remove the requirement of surjectivity in Definition \ref{def:info_structure}: products, coproducts and even other (co)limits would make sense. However, the probabilistic interpretation of such structures would be different: they do not have classical or quantum representations in the sense of Definitions \ref{def:representation_str} and \ref{def:q_representation_str}, nor the cohomological results in Section \ref{sec:classical_info_cohomology} apply to them (because a product $X\wedge Y$ is always degenerate if one of the ``coordinate projections'' ${\sheaf E}_{X\wedge Y}\to {\sheaf E}_X$ or ${\sheaf E}_{X\wedge Y}\to {\sheaf E}_Y$ is not surjective). The definition of information cohomology would still hold, since it only depends on the underlying conditional meet semilattice.

 \subsection{Digression: Representations}\label{sec:models}
 
We introduce here the notion of \emph{representation} of an information structure in terms of classical observables (measurable functions) or quantum observables (self-adjoint operators), as a bridge between our categorical definitions and the more traditional models used in classical and quantum probability theory. The rest of the paper does not depend on this section.  For simplicity, we restrict ourselves to finite information structures. 
 
 Recall that $\ObsFin(\Omega)$ denotes the poset of finite partitions of a set $\Omega$, ordered by the relation of refinement, and $\square$ is  the ``forgetful'' functor from $\ObsFin(\Omega)$ into $\Sets$ introduced in Example \ref{ex:classical_structure}, that maps each partition $\{A_1,...,A_n\}$ to the set $\{A_1,...,A_n\}$  and each arrow of refinement to a surjection.
 
 \begin{defi}\label{def:representation_str}
A \keyt{classical representation} of a finite information structure $(\Sinf,{\sheaf E})$ is a pair $(\Omega, \rho)$, where $\Omega$ is a set and $\rho=(\rho_0,\widehat\rho): (\Sinf, {\sheaf E})\to (\ObsFin(\Omega),\square)$ is a morphism of information structures  such that $\widehat \rho:\sheaf E \to \rho_0^* \square$ is a natural isomorphism (i.e. the components $\widehat \rho_X:{\sheaf E}(X)\to \square \rho_0(X)$ are bijections, natural in $X$). 
\end{defi}

If $(\Omega,\rho)$ is a classical representation of $\Sinf$, each observable $X$ in $\Sinf$ can be associated with a unique function $\widetilde  X:\Omega \to {\sheaf E}_X$, in such a way that $\rho_0(X)$ is the partition induced by $\widetilde X$ and $\widehat \rho_X(x)=\widetilde X^{-1}(x)$. Since $\rho_0$ is a morphism of conditional meet semilattices, for any $X,Y\in\Ob \Sinf$ the joint  $(\widetilde X,\widetilde Y):\Omega \to {\sheaf E}_X\times {\sheaf E}_Y$ is equivalent to $\widetilde{X\wedge Y}:\Omega \to {\sheaf E}_{XY}$, in the sense that both induce the same partition of $\Omega$.

%

The next proposition points to a close link between representations and $\lim {\sheaf E}$. Recall that the limit of the functor ${\sheaf E}:\Sinf \to \Sets$ is defined as
\begin{equation}
\lim {\sheaf E} := \Hom_{[\Sinf,\Sets]}(\ast, {\sheaf E}),
\end{equation}
where $[\Sinf,\Sets]$ is the category of covariant functors from $\Sinf$ to $\Sets$, and $\ast$ is the functor that associates to each object a one-point set; equivalently
\begin{equation}\label{lim_as_subset_product}
\lim {\sheaf E} \cong \bigset{(s_Z)_{Z\in\Ob\cat{S}} \in \prod_{Z\in\Ob\cat{S}} {\sheaf E}(Z) }{ {\sheaf E}\pi_{YX}(s_X) = s_Y \text{ for all }\pi_{YX}:X\to Y},
\end{equation}
where $s_Z$ denotes $\varphi(\ast)$ for any $\varphi\in \Hom_{[\Sinf,\Sets]}(\ast,{\sheaf E})$. The requirements imposed on $(s_Z)_{Z\in\Ob\cat{S}}$ in  \eqref{lim_as_subset_product} are referred hereafter as \emph{compatibility conditions}. We denote the restriction of each projection $\pi_{{\sheaf E}(X)}: \prod_{Z\in\Ob\cat{S}} {\sheaf E}(Z) \to {\sheaf E}(X)$ to $\lim {\sheaf E}$ by the same symbol. We interpret the limit as all possible combinations of compatible outcomes. 

\begin{prop}\label{prop:model_implies_noncontextuality}
If $(\Sinf, {\sheaf E})$ has a classical representation $(\Omega,\rho,\widehat \rho)$, then for any $X\in \Ob \Sinf$ and any $x\in {\sheaf E}(X)$, there exists an element $s(x)\in \lim {\sheaf E}$ such that $\pi_{{\sheaf E}(X)}(s(x)) = x$.
\end{prop}
\begin{proof}
For each $X$ in $\Ob \Sinf$, there is a surjection $\pi_X:\Omega\to \sheaf E(X)$ obtained as the composition of $\xi_X:\Omega\to \square \rho_0(X)$, which maps $\omega\in \Omega$ to the part that contains it, and $\widehat \rho_X^{-1}:\square \rho_0(X) \to \sheaf E(X)$. The maps $\{\pi_X\}$ define a cone over $\sheaf E$ i.e. given $f:Y\to Z$ in $\Sinf$, one has a commutative diagram 
$$ 
\begin{tikzcd}
\Omega \ar[r, "\xi_Y"] \ar[d, equal] & \square \rho_0(Y) \ar[d, "\square\rho_0 f"] \ar[r, "\widehat \rho_Y^{-1}"] & {\sheaf E}(Y) \ar[d, "{\sheaf E}f"] \\
\Omega \ar[r, "\xi_Z"] & \square\rho_0(Z) \ar[r, "\widehat \rho_Z^{-1}"] & {\sheaf E}(Z)
\end{tikzcd};
$$
the commutation of the left square comes from the definition of $\ObsFin(\Omega)$ and $\square$, and the right square commutes because $\widehat \rho$ is a natural isomorphism. Therefore, there is a map $\pi:\Omega\to \lim \sheaf E$ and the desired section is obtained as the image under $\pi$ of any $\omega\in \widehat \rho_X(x)$. 

%
\end{proof}

\begin{defi}
An information structure is  \keyt{noncontextual} if, for all $X\in \Ob \Sinf$ and all $x\in {\sheaf E}(X)$, there exists an element $s(x)\in \lim {\sheaf E}$ such that $\pi_{{\sheaf E}(X)}(s(x)) = x$.
\end{defi}

Thus information structures are sufficiently flexible to model \emph{contextual} situations, which arise when data is locally consistent but globally inconsistent. This happens in different domains, notably in quantum mechanics and in database theory. In the terminology of \cite[Sec.~3]{Abramsky2015}, a structure is said to be \emph{logically contextual at a value $x\in {\sheaf E}_X$} if $x$ belongs to no compatible family of measurements, i.e. there is no section $s(x)\in \lim {\sheaf E}$ such that $\pi_{{\sheaf E}(X)}(s(x)) = x$, and \emph{strongly contextual} if ${\sheaf E}$ does not accept any global section, i.e. $\lim {\sheaf E} = \emptyset$.  

\begin{theorem}
A finite information structure $(\Sinf, {\sheaf E})$ is noncontextual if and only if it has a classical representation.
\end{theorem}
\begin{proof}
Proposition \ref{prop:model_implies_noncontextuality} already showed that existence of a classical representation implies noncontextuality. 

In the other direction, let us suppose that a finite structure $(\Sinf,{\sheaf E})$ is noncontextual. We are going to show that it can be represented by partitions of $\lim {\sheaf E}$. Define first  $\tau_0:\Sinf\to\ObsFin( \lim {\sheaf E})$ as follows: associate to $X\in \Ob\cat{S}$ the set $\tau_0(X):=\{\pi_{{\sheaf E}(X)}^{-1}(x)\}_{x\in {\sheaf E}(X)}$, which is clearly a partition of $\lim {\sheaf E}$; the associated map $\widehat\tau_X:{\sheaf E}_X \to \tau(X), \: x\mapsto \pi_{{\sheaf E}(X)}^{-1}(x)$ is a bijection, because noncontextuality means that each $\pi_{{\sheaf E}(X)}^{-1}(x)$ is nonempty.  Given $\pi_{YX}:X\to Y$, there is a corresponding arrow $\tau_0(X)\to \tau_0(Y)$ in $\ObsFin(\Omega)$ in virtue of the identity
\begin{equation}
\pi_{{\sheaf E}(Y)}^{-1}(y) = \bigcup_{x\in {\sheaf E}\pi_{YX}^{-1}(y)} \pi_{{\sheaf E}(X)}^{-1}(x),
\end{equation}
which is proved as follows: if $x\in {\sheaf E}\pi_{YX}^{-1}(y)$ and $s\in \pi_{{\sheaf E}(X)}^{-1}(x)$, then 
$$\pi_{{\sheaf E}(Y)}(s) = {\sheaf E}\pi_{YX}(\pi_{{\sheaf E}(X)}(s)) = {\sheaf E}\pi_{YX}(x) = y,$$ which means $ \cup_{x\in {\sheaf E}\pi_{YX}^{-1}(y)} \pi_{{\sheaf E}(X)}^{-1}(x) \subset \pi_{{\sheaf E}(Y)}^{-1}(y)$; to prove the other inclusion, take $s=(s_Z)_{Z\in \Ob\cat{S}}\in \pi_{{\sheaf E}(Y)}^{-1}(y)$ and note that $s_X$ must satisfy---by definition---the compatibility condition ${\sheaf E}\pi_{YX}(s_X)=s_Y=y$, thus $s_X\in {\sheaf E}\pi_{YX}^{-1}(y)$ and $s$ itself belong to $\cup_{x\in {\sheaf E}\pi_{YX}^{-1}(y)} \pi_{{\sheaf E}(X)}^{-1}(x)$. 

Finally, to prove that $\tau_0$ defines a morphism of conditional meet semilattices, consider a diagram $X \leftarrow X\wedge Y \rightarrow Y$, and an arbitrary partition $W$ of $\lim {\sheaf E}$ that refines $\tau_0(X)$ and $\tau_0(Y)$. We have to show that $W$ also refines $\tau_0(X\wedge Y)$. If $W$ refines $\tau_0(X)$, each $w\in W$ ($w$ is a subset of $\lim {\sheaf E}$) is mapped to certain $x_w$ by $\pi_{{\sheaf E}(X)}$; analogously, $\pi_{{\sheaf E}(Y)}(w)=\{y_w\}$. This means that $\pi_{{\sheaf E}(X\wedge Y)} (w)=\{z_w\}$, where $z_w$ is the only point of ${\sheaf E}(X\wedge Y)$ that satisfies ${\sheaf E}(\pi_{X(X\wedge Y)})(z_w)=x_w$, ${\sheaf E}(\pi_{X(X\wedge Y)})(z_w)=y_w$, which means that $w\subset \pi_{{\sheaf E}(X\wedge Y)}^{-1}(z_w)$. Thus, $W$ refines $\tau_0(X\wedge Y)$.
\end{proof}

\begin{ex}
In the notation of Example \ref{ex:simplicial_structures}, let $\cat K$ be simplicial subcomplex of $\Delta(\{1,2,3\})$ with minimal objects $\{1,2\}$, $\{1,3\}$, and $\{2,3\}$, which can be pictured as a triangle. We consider three possible functors $\sheaf E$. All of them associate to $\{1\}$, $\{2\}$, and $\{3\}$ the set $\{0,1\}$. We give $\sheaf E_{\{i,j\}}$ as a subset of $\sheaf E_{\{i\}} \times \sheaf E_{\{j\}}$ so that the maps $\sheaf E_{\{i,j\}}\to \sheaf E_{\{i\}}$ are the restrictions of the canonical projectors. 
\begin{enumerate}
\item If $\sheaf E_{\{i,j\}} = \{(0,1),(1,0)\}$ (for any $i,j\in \{1,2,3\}$ such that $i\neq j$), then $\lim \sheaf E = \emptyset$ i.e. the structure is strongly contextual, and it has no classical representation. This means that there is no assignment of values for the three observables associated to the vertices of the triangle that is consistent with the constraints associated to the edges.
\item If $\sheaf E_{\{i,j\}} = \{(0,0),(1,1)\}$ (for any $i,j\in \{1,2,3\}$ such that $i\neq j$), them $\lim \sheaf E \cong \{0,1\}$. The observables associated to the vertices must take  all the same value.
\item If only $\sheaf E_{\{1,2\}}= \{(0,0),(1,1)\}$, and the other two $\sheaf E_{\{i,j\}}$ equal $\sheaf E_{\{i\}}\times \sheaf E_{\{j\}}$, then $\lim \sheaf E$ has just six different elements. This corresponds to imposing a constraint that makes the observables associated to $\{1\}$ and $\{2\}$ identical.
\end{enumerate}
\end{ex}

\begin{ex}\label{inverse_limit}
Given a concrete structure $(\Sinf, \square)$ on a set $\Omega$ (see Example  \ref{ex:classical_structure}), the set $\lim \square$ may differ from the original $\Omega$, as the following example shows.  Set $\Omega=\{1,2,3,4\}$ and  $X_{i} = \{\{i\}, \Omega \sm \{i\}\}$, for $i=1,...,4$. Let $\Sinf$ be the concrete  structure  that includes only the partitions $X_1$, $X_2$, $X_3$, $X_1X_2$, and $X_2X_3$. The corresponding generalized information structure is given by a conditional meet semilattice represented by the graph 
$$
\begin{tikzcd}
 & & \Us & &\\
X_1 \ar[urr] & &X_2 \ar[u] & & X_3\ar[ull]\\
& X_1X_2 \ar[ul]\ar[ur] & & X_2X_3 \ar[ul]\ar[ur] & 
\end{tikzcd}
$$
and the functor $\square$ can pictured as
\footnotesize
$$
\begin{tikzcd}
 & & \{{\{1,2,3,4\}}\} & &\\
\{\{1\},\{2,3,4\}\} \ar[urr] & &\{{\{2\}},{\{1,3,4\}}\} \ar[u] & & \{{\{3\}},{\{1,2,4\}}\}\ar[ull]\\
& \{{\{1\}},{\{2\}},{\{3,4\}}\} \ar[ul]\ar[ur] & & \{{\{2\}},{\{3\}},{\{1,4\}}\}\ar[ul]\ar[ur] & 
\end{tikzcd}
$$
\normalsize
where each arrow corresponds to a surjection of finite sets that sends $I$ to $J$ when $I\subset J$.  In this case, $\lim {\sheaf E} \subset \{\ast\} \times {\sheaf E}(X_1) \times {\sheaf E}(X_2)\times {\sheaf E}(X_3) \times {\sheaf E}(X_1X_2)\times {\sheaf E}(X_2X_3)$ corresponds to the set
\footnotesize
\begin{align*}
\lim {\sheaf E}&= \{
({\{1,2,3,4\}},{\{1\}},{\{1,3,4\}},{\{3\}},{\{1\}},{\{3\}}), 
({\{1,2,3,4\}},{\{1\}},{\{1,3,4\}},{\{1,2,4\}},{\{1\}},{\{1,4\}}), \\
& \qquad ({\{1,2,3,4\}},{\{2,3,4\}},{\{2\}},{\{1,2,4\}},{\{2\}},{\{2\}}), 
 ({\{1,2,3,4\}},{\{2,3,4\}},{\{1,3,4\}},{\{3\}},{\{3,4\}},{\{3\}}),\\
& \qquad  ({\{1,2,3,4\}},{\{2,3,4\}},{\{1,3,4\}},{\{1,2,4\}},{\{3,4\}},{\{1,4\}})\}.
\end{align*}
\normalsize
The difference between $\Omega$ and $\lim {\sheaf E}$ is explained by the presence of 
$$({\{1,2,3,4\}},{\{1\}},{\{1,3,4\}},{\{3\}},{\{1\}},{\{3\}});$$
 this measurement is impossible in the concrete structure $\Sinf\subset \ObsFin (\Omega)$ (because the underlying sample $\omega\in \Omega$ should belong to $\{1\}$ and $ {\{3\}}$), but the observables in $( \Sinf,\square)$ cannot distinguish between the points $1$ and $3$, a sort of nonseparability. In fact, if we also include $X_1X_3$ at the beginning, we obtain $\Omega \cong \lim {\sheaf E}$.
\end{ex}

%

\begin{prop}
The product and the coproduct of two noncontextual structures is noncontextual. 
\end{prop}
\begin{proof}
Let $\Sinf_1$ and $\Sinf_2$ be noncontextual structures. We use the notations in the proof of Proposition \ref{prop:products_and_coproducts}.

Products: Consider a point $(x_1,x_2)\in {\sheaf E}(\langle X_1, X_2\rangle)$. There exist sections $$s^i(x_i)=(s^i_Z(x_i))_{Z\in \Ob\cat{S_i}}\in \lim {\sheaf E}_i \subset \prod_{Z\in \Ob\cat{S_i}} {\sheaf E}_i(Z),$$  such $\pi_{{\sheaf E}_i(X_i)}(s(x_i))=x_i$ (for $i=1,2$). Note that the vector $$s(x_1,x_2):=(s^1_{Z_1}(x_1),s^2_{Z_2}(x_2))_{\langle Z_1, Z_2\rangle \in \Ob\cat{S}} \in \prod_{\langle Z_1, Z_2\rangle  \in \Ob\cat{S}} {\sheaf E}(\langle Z_1, Z_2\rangle )$$ satisfies all the compatibility conditions and is therefore in $\lim  {\sheaf E}$. By definition, \\$\pi_{{\sheaf E} (\langle X_1, X_2\rangle )}(s(x_1,x_2))=(x_1,x_2)$.

Coproducts: given $X\in \Sinf_1$, $x\in {\sheaf E}(X)$, there exists  $s^1(x)=(s_Z(x))_{Z\in \Ob\cat{S_1}}\in \lim  {\sheaf E}_1$ satisfying $\pi_{{\sheaf E}(X)}(s^1(x))=x$, and similarly for ${\sheaf E}_2$; we can build a new vector $( s_Z(x))_{Z\in \Ob\cat{\Sinf}} \in \Ob\cat{S_1}\sqcup \Ob\cat {S_2}$ such that $s_Z = s^i_Z$ if $Z\in \Ob\cat{S_i}$; luckily, for $\Us$ there is no choice.
\end{proof}

\subsubsection{Representations of finite products and coproducts} There is a general construction of a representation for a product or coproduct of two information structures whenever each of the factors is already represented. 

Let $\Omega_1$ and $\Omega_2$ be nonempty sets. Given collections $\salg A = \{A_i\}_i$ of subsets of $\Omega_1$ and $\salg B = \{B_j\}_j$ of subsets of $\Omega_2$, denote by 
$\salg A \times \salg B$ the collection $\set{A_i \times B_j }{A_i\in \salg A \text{ and } B_j\in \salg B}$ of subsets of $\Omega_1\times \Omega_2$. If $\salg A$ and $\salg B$ are partitions, then $\salg A\times \salg B$ is a partition too. 

Let $(\Omega_i, \rho^i)$, with $\rho^i=(\rho_i,\widehat {\rho^i})$, be a classical representation of $(\Sinf_i,{\sheaf E}_i)$, for $i=1,2$. Associate to each observable $\langle X_1, X_2\rangle  \in \Ob\cat{ S_1 \times  S_2}$ the partition of $\Omega_1 \times \Omega_2$ given by
\begin{equation}
\rho^\times(\langle X_1, X_2\rangle) := \rho_0^1(X_1) \times \rho_0^2(X_2).
\end{equation}
There is a natural transformation $\widehat {\rho^\times}_{\langle X_1, X_2\rangle}:{\sheaf E}(\langle X_1, X_2\rangle)\to \square \rho^\times(\langle X_1, X_2\rangle)$ that maps $(x_1,x_2)$ to $\widehat {\rho^1}_{X_1}(x_1) \times \widehat {\rho^2}_{X_2}(x_2)$.

Analogously, for each $X\neq \Us$ in $\Ob\cat{S_1 \sqcup S_2}$, let us define the partition of $\Omega_1 \times \Omega_2$ given by
\begin{equation}
\rho^{\sqcup}(X) = \begin{cases}
\rho^1(X) \times \{\Omega_2\} & \text{if } X \in \Ob\cat{S_1}\\
\{\Omega_1\}\times \rho^2(X) & \text{if } X \in \Ob\cat{S_2}
\end{cases}.
\end{equation}
In particular, $\rho^{\sqcup}(\Us) = \{\Omega_1\times \Omega_2\}$. The maps $\widehat {\rho^{\sqcup}}_X$ are  $x\mapsto \widehat {\rho^1}(x)\times \{\Omega_2\}$ or $x\mapsto\{\Omega_1\}\times \widehat {\rho^2}(x)$ accordingly. 

\begin{prop}\label{prop:prod_and_coprod_of_models}
Let $(\Omega_i, \rho^i)$ be a classical representation of $(\Sinf_i,{\sheaf E}_i)$, for $i=1,2$. Then 
\begin{enumerate}
\item $(\Omega_1\times \Omega_2, \rho^\times)$ is a classical representation of $(\Sinf_1,{\sheaf E}_1) \times (\Sinf_2,{\sheaf E}_2)$;
\item $(\Omega_1\times \Omega_2, \rho^{\sqcup})$ is a classical representation of $(\Sinf_1,{\sheaf E}_1) \sqcup (\Sinf_2,{\sheaf E}_2)$.
\end{enumerate}
\end{prop}

The proof depends on the following lemma.

\begin{lem} \label{lemma_partitions}
\begin{enumerate}
\item If $\salg A = \{A_i\}_i$ and $\salg A' = \{A'_j\}_j$ are finite partitions of a set $\Omega$, then $\sigma(\salg A, \salg A') = \sigma(\{A_i\cap A'_j\}_{i,j})$, and the nonempty elements of $\{A_i\cap A'_j\}_{i,j}$ are the atoms of $\sigma(\salg A, \salg A')$.
\item If $\salg A = \{A_i\}_i$, $\salg A' = \{A'_j\}_j$ are two finite partitions of $\Omega_1$ and $\salg B = \{B_l\}_l$, $\salg B' = \{B_m\}_m$ two finite partitions of $\Omega_2$, then $(\salg A\times \salg B)(\salg A' \times \salg B') = \salg A \salg A' \times  \salg B \salg B'$, where juxtaposition of partitions denotes their product in $\ObsFin(\Omega)$.
\end{enumerate}
\end{lem}
\begin{proof}
\begin{enumerate}
\item On the one hand, note that each set $A_i\cap A'_j$ is contained in $\sigma(\salg A, \salg A')$, therefore $\sigma(\{A_i\cap A'_j\}_{i,j}) \subset \sigma(\salg A, \salg A')$. On the other, each generator $A_i \in \salg A$ of $\sigma(\salg A, \salg A')$ can be written as 
$$A_i = A_i \cap \Omega = A_i \cap \left(\bigcup_j A'_j\right)= \bigcup_j (A_i \cap A_j'),$$
and similarly for the generators $A'_j\in \salg A'$, which implies that $\sigma(\salg A, \salg A') \subset \sigma(\{A_i\cap A'_j\}_{i,j})$. The reader can verify that the nonempty elements of $\{A_i\cap A'_j\}_{i,j}$ are atoms.
\item The previous result can be read as $\salg A \salg A' = \{A_i\cap A'_j\}_{i,j}$. The set-theoretical identity 
\begin{equation}
(A_i \times B_l) \cap (A'_j\times B'_m) = (A_i \cap A'_j) \times (B_l \cap B'_m),
\end{equation}
implies that the atoms of $(\salg A\times \salg B)(\salg A' \times \salg B')$ and $\salg A \salg A' \times  \salg B \salg B'$ coincide.
\end{enumerate}
\end{proof}

\begin{proof}[of Proposition \ref{prop:prod_and_coprod_of_models}]
 Most verifications are almost immediate from the definitions. We simply prove that $\rho^\times(\langle X_1, X_2\rangle\wedge \langle Y_1, Y_2\rangle) = \rho^\times (\langle X_1, X_2\rangle) \rho^\times (\langle Y_1, Y_2\rangle)$. Note that  
 \begin{align*}
 \rho^\times(\langle X_1, X_2\rangle\wedge \langle Y_1, Y_2\rangle) 
 &= \rho^\times(\langle X_1\wedge Y_1, X_2\wedge Y_2 \rangle)\\
 & = \rho^1(X_1\wedge Y_1) \times \rho^2(X_2 \wedge Y_2)\\
 & = \rho^1(X_1)\rho^1(Y_1) \times \rho^2(X_2)\rho^2(Y_2)\\
 & = (\rho^1 (X_1) \times \rho^2(X_2))(\rho^1(Y_1)\times \rho^2(Y_2))\\
 &= \rho^\times (\langle X_1, X_2\rangle) \rho^\times (\langle Y_1, Y_2\rangle)
 \end{align*}
 The first equality comes from the construction of $\Sinf_1 \times \Sinf_2$; the second, from the definition of $\rho^\times$; the third, from the fact that $\rho^1=\rho_0^1$ and $\rho^2=\rho_0^2$ are morphisms of conditional meet semilattices; the fourth equality is just a consequence of Lemma \ref{lemma_partitions}, and the fifth is a rewriting of the previous one.
\end{proof}

The partitions of $\Omega_1\times \Omega_2$ in the image of $\rho^{\sqcup}$ are also in the image of $\rho^\times$. This is consistent with the existence of a canonical morphism of structures $\phi:(\Sinf_1,{\sheaf E}_1)\sqcup(\Sinf_2,{\sheaf E}_2)\to(\Sinf_1,{\sheaf E}_1)\times(\Sinf_2,{\sheaf E}_2)$, with $\phi_0$ given at the level of objects by the injection
\begin{equation}
X\mapsto \begin{cases}
\Us_{\Sinf_1\times \Sinf_2} & \text{if } X= \Us_{\Sinf_1\sqcup \Sinf_2} \\
\langle X, \Us_{\Sinf_2}\rangle & \text{if } X \in \Ob\cat{S_1} \\
\langle  \Us_{\Sinf_1}, X\rangle & \text{if } X \in \Ob\cat{S_2} 
\end{cases},
\end{equation}
and the corresponding components $\widehat \phi_X$ being the obvious bijections:  ${\sheaf E}_1(X)\to {\sheaf E}_1(X)\times \{\ast\}$ when $X\in \Ob\cat {S_1}$ or ${\sheaf E}_2(X)\to  \{\ast\}\times {\sheaf E}_2(X)$ when $X\in \Ob\cat {S_2}$. The representation $(\Omega_1\times \Omega_2, \rho^\times)$ on $(\Sinf_1,{\sheaf E}_1)\times(\Sinf_2,{\sheaf E}_2)$ may be pulled back under $\phi$, and the resulting representation $(\Omega_1\times \Omega_2, \rho^\times \circ \phi)$ on $(\Sinf_1,{\sheaf E}_1)\sqcup(\Sinf_2,{\sheaf E}_2)$ coincides with $(\Omega_1\times \Omega_2, \rho^{\sqcup})$. This is  a particular case of a more general procedure to pull back representations, valid for any morphism of structures $\phi=(\phi_0, \widehat \phi)$ such that  $\widehat \phi$ is a natural isomorphism.

\subsubsection{Quantum observables} For the sake of completeness, we also indicate how an information structure may be represented by noncommutative observables adapted to quantum mechanics.

Let $V$ be a finite dimensional Hilbert space: a complex vector space with  a positive definite hermitian form $\langle\cdot, \cdot \rangle$. In the quantum setting, random variables are generalized by endomorphisms of $V$ (operators). An operator $H$ is called hermitian if for all $u,\,v\in V$, one has $\langle u, Hv\rangle = \langle Hu , v\rangle$. A quantum observable is a hermitian operator: the result of a quantum experiment is supposed to be an eigenvalue of such operator, that is always a real number.

A fundamental result of linear algebra, the Spectral Theorem \cite[Sec.~79]{Halmos1958}, says that each hermitian operator $Z$ can be decomposed as a weighted sum of positive hermitian projectors
$
Z = \sum_{j=1}^K z_j V_j
$
where $z_1,...,z_K$ are the (pairwise distinct) \emph{real} eigenvalues of $Z$. Each $V_j$ is the projector on the eigenspace spanned by the eigenvectors of $z_j$; the dimension of this subspace equals the multiplicity of $z_j$ as eigenvalue.  As hermitian projectors, they satisfy the equation $V_j^2 = V_j$ and $V_j^* = V_j$. They are also mutually orthogonal ($V_jV_k = 0$ for integers $j,\,k$), and their sum equals the identity, $\sum_{1\leq j\leq K} V_j= \id_V.$ 

In analogy to the classical case, we consider  equivalent two hermitian operators that define the same direct sum decomposition  $\{V_j\}_j$ of $V$ by means of the Spectral Theorem, ignoring the particular eigenvalues. For us, observable and direct sum decomposition are then interchangeable terms. In what follows, we denote by $V_\alpha$ both the subspace of $V$ and the orthogonal projector on it. A decomposition $\{V_\alpha\}_{\alpha\in A}$ is said to refine $\{V'_\beta\}_{\beta \in B}$ if each $V'_\beta$ can be expressed as sum of subspaces $\{V_\alpha\}_{\alpha \in A_\beta}$, for certain $A_\beta\subseteq A$. In that case we say also that $\{V_\alpha\}_{\alpha\in A}$ divides $\{V'_\beta\}_{\beta \in B}$, and we write $\{V_\alpha\}_{\alpha\in A} \to \{V'_\beta\}_{\beta \in B}$. With this arrows, direct sum decompositions form a category denoted $\cat{DSD}(V)$; the opposite category is the poset of quantum measurement contexts studied in \cite{Constantin2012}. The reader may verify that it is a conditional meet semilattice with terminal object $\{V\}$, but not a meet semilattice. We also introduce in this case a functor $\square$ that maps each decomposition $\{V_\alpha\}_{\alpha\in A}$ to the set $\{V_\alpha\}_{\alpha\in A}$, and each arrow of refinement to a surjection.

 \begin{defi}\label{def:q_representation_str}
A \keyt{quantum representation} of a finite information structure $(\Sinf,{\sheaf E})$ is a pair $(V, \rho)$, where $V$ is a Hilbert space and $\rho=(\rho_0,\widehat\rho): (\Sinf, {\sheaf E})\to (\cat{DSD}(V),\square)$ is a morphism of information structures  such that $\widehat \rho$ is a natural isomorphism (i.e. for each $X\in\Ob\cat{S}$, the component $\widehat \rho_X:{\sheaf E}(X)\to \square \rho(X)$ is a bijection). 
\end{defi}


The cohomological computations in Section \ref{sec:classical_info_cohomology} concern classical probabilities, but the general constructions in Section \ref{sec:topoi} only depend on the conditional meet semilattice with terminal object, hence they are equally valid in the quantum case. Using those results, one may recover the explicit cochain complex that defined quantum information cohomology in \cite{Baudot2015}.

\section{Information cohomology via derived functors}\label{sec:topoi}

In this section, we define information cohomology as a derived functor, following a remark in \cite{Baudot2015}, and then reobtain the explicit cochain complex used there. The main new ingredient is Proposition \ref{prop_bar_projective}, which proves that the relative bar resolution is a projective resolution.  We suppose that the reader is familiar with abelian categories and derived functors: we use the definitions and notations in \cite[Ch. 1 \& 2]{Weibel1994}.

\subsection{Definition}\label{sec:topos_cohomology}
Let $\Sinf$ be a conditional meet semilattice with terminal object $\Us$. We view it as a site with the trivial topology, such that every presheaf is a sheaf. For each $X\in \Ob\cat{S}$, set $\Smon_X:=\{Y\in \Ob\cat{S} \mid X\to Y\}$, with the monoid structure given by the product of in $\Sinf$: $(Z,Y)\mapsto ZY:= Z\wedge Y$.  Let $\Sring_X:= \Rr[\Smon_X]$ be the corresponding monoid algebra.  The contravariant functor $X\mapsto \Sring_X$ is a sheaf of rings; we denote it by $\Sring$. The pair $(\Sinf, \Sring)$ is a ringed site. 

The category  $\Mod(\Sring)$  is abelian  \cite[\href{https://stacks.math.columbia.edu/tag/03DA}{Lemma 03DA}]{stacks-project} and has enough injective objects \cite[\href{https://stacks.math.columbia.edu/tag/01DU}{Theorem 01DU}]{stacks-project}. For a fixed object $\sheaf O$ of $\Mod(\Sring)$, the covariant functor $\Hom(\sheaf O, -)$ is always additive and left exact: the associated right derived functors are $R^n\Hom(\sheaf O,-)=:\Ext^n (\sheaf O,-)$, for $n\geq 0$.
 
 Let $\Rr_{\Sinf}(X)$ be the $\Sring_X$-module defined by the trivial action of $\Sring_X$ on the abelian group $(\Rr,+)$ (for $s \in \Smon_X$ and $r\in \Rr$, take $s\cdot r = r$). The presheaf that associates to each $X\in \Ob\cat S$ the module $\Rr_{\Sinf}(X)$, and to each arrow the identity map is denoted  $\Rr_{\Sinf}$.
 
In Section \ref{sec:intro:cohomology}, we have defined the \keyt{information cohomology} associated to the conditional meet semilattice $\Sinf$, with coefficients in $\sheaf F\in \Mod(\Sring)$, as
\begin{equation}\label{eq:defi_info_cohomology}
H^\bullet(\Sinf, \sheaf F) := \Ext^\bullet(\Rr_{\Sinf}, \sheaf F).
\end{equation}
In all the examples contained in this article, \cite{Baudot2015} or \cite{Vigneaux2019-thesis}, the sheaf  $\sheaf F$ is obtained composing the functor $\sheaf M$ of an information structure $(\Sinf, \sheaf M)$ with other functors.

Information cohomology is formally analogous to group cohomology. In this case, one begins with a multiplicative group $G$ and constructs the free abelian group $\Zz[G]$, whose elements are finite sums $\sum m_g g$, with $g\in G$ and $m_g\in \Zz$. The product of $G$ induces a  product between two such elements, and makes $\Zz[G]$ a ring, called the integral group ring of $G$. The category of $\Zz[G]$-modules is abelian and has enough injective objects. The cohomology groups of $G$ with coefficients in a $\Zz[G]$-module $A$ are defined by
\begin{equation}
H^n(G, A) = \Ext^n(\Zz, A),
\end{equation}
where $\Zz$ is the trivial module.


 Let $\cat C$ be an abelian category with enough injectives, and suppose that we are interested in computing the groups $\{\Ext^n(A,B)\}_{n\geq 0}$  for certain fixed objects $A$ and $B$. In addition, we assume that $A$ has a \emph{projective} resolution $0 \leftarrow A \leftarrow P_0 \leftarrow P_1 \leftarrow ...$. Then, Theorem 4.6.10 in \cite{Schapira2008} implies that, for all $n\geq 0$,
\begin{equation}
(R^n\Hom_{\cat C}(A,-))(B) \simeq (R^n\Hom_{\cat C}(-,B))(A).
\end{equation} 
We denote $(R^n\Hom_{\cat C}(-,B))(A)$ by $\underline{\Ext}^n(A,B)$. They are given by the formulas
\begin{align}
\underline{\Ext}^0(A,B) &= \ker(\Hom(P_0,B) \to \Hom(P_1,B)), \label{Ext_0}\\
\underline{\Ext}^i(A,B) &= \frac{\ker(\Hom(P_i,B) \to \Hom(P_{i+1},B))}{\im(\Hom(P_{i-1},B) \to \Hom(P_{i},B))}, \quad \text{for } i\geq 1. \label{Ext_n}
\end{align}

\subsection{Nonhomogeneous bar resolution}\label{sec:bar}

In this section, we introduce a projective  resolution of the sheaf of $\Sring$-modules $\Rr_{\Sinf}$: a long exact sequence
\begin{equation}\label{eq:bar_resolution}
\begin{tikzcd}
0 
& \Rr_\Sinf \ar[l] 
& \sheaf B_0            \ar[l, "\epsilon", swap]  
& \sheaf B_1            \ar[l, "\partial_1", swap] 
& \sheaf B_2            \ar[l, "\partial_2", swap] 
& ...\ar[l, "\partial_3", swap] 
\end{tikzcd}
\end{equation}
that will allow us to determine the information cohomology, in accordance with \eqref{Ext_0} and \eqref{Ext_n}.

 For any $n\geq 0$, let $\sheaf B_n(X)$ be the tensor product over $\Rr$ of $n+1$ copies of $\sheaf A_X$, i.e. $\sheaf B_n(X) = \Sring_X^{\otimes (n+1)}$, equipped with an action of  $\sheaf A_X$ given by $$(a, b_0\otimes b_1 \otimes \cdots \otimes b_n) \mapsto ab_0 \otimes b_1 \otimes \cdots \otimes b_n.$$
Equivalently, $\sheaf B_n(X)$ is the free $\Sring_X$ module generated by the symbols $[X_1|...|X_n]:= 1  \otimes X_1 \otimes \cdots \otimes X_n$, where $\{X_1,...,X_n\} \subset \Smon_X$.  Remark that $\sheaf B_0(X)$ is the free module on one generator $[\,]$. An arrow $X\to Y$ in $\Sinf$ induces an inclusion $\Smon_Y\hookrightarrow \Smon_X$, hence an inclusion $\sheaf B_n(Y)\hookrightarrow \sheaf B_n(X)$, implying that $\sheaf B_n = \Sring^{\otimes n+1}$ is a presheaf of $\Sring$-modules for each $n\geq 0$.

We introduce now $\Sring_X$-module morphisms: an augmentation $\epsilon_X:\sheaf B_0(X)\to\Rr_{\Sinf}(X)$ given by the equation $\epsilon([\,])= 1$, and boundary morphisms  $\partial: \sheaf B_n(X) \to \sheaf B_{n-1}(X)$ given by
\begin{equation}\label{eq:boundary_op}
\partial ([X_1|...|X_n]) = X_1[X_2|...|X_n] + \sum_{k=1}^{n-1} (-1)^k [X_1|...|X_kX_{k+1}|...|X_n] +(-1)^n[X_1|...|X_{n-1}].
\end{equation}
These morphisms are natural in $X$.

\begin{prop}
The complex \eqref{eq:bar_resolution} is a resolution of the sheaf $\Rr_{\Sinf}$.
\end{prop}
\begin{proof}
The construction corresponds to the relatively projective bar resolution \cite[Ch.~IX]{MacLane1994}, more specifically to the example developed at the end of Appendix \ref{sec:relative}, setting  $\sheaf R$ and $\sheaf T$ there equal to $\sheaf S$ and $ \Rr_{\Sinf}$, respectively. The resolution $\sheaf B_\bullet$ introduced above is $B_\bullet \sheaf C$, for $ \sheaf C=\Rr_{\Sinf}$. The notation can be simplified, because $\sheaf  C(X)$ is generated  freely generated by $1$ as an $\Rr$-module. Therefore, $ B_0\sheaf C$ is generated over $\Sring_X$ by the symbol $[1]$, written simply as $[\,]$. In general, $ B_n\sheaf C(X)$ is generated over $\Sring_X$ by the symbols $[X_1|...|X_n|1]$, or simply $[X_1|...|X_n]$ if we omit the 1. 
\end{proof}

Thus far we have a resolution with relatively free objects, that in general need not be projective. However, the special properties of $\Sinf$ allow us to improve the result.

\begin{prop}\label{prop_bar_projective}
For each $n\geq 0$, the sheaf $\sheaf B_n$ is a projective object in $\Mod(\Sring)$.
\end{prop}
\begin{proof}
Let $\sheaf T$ be the presheaf of sets defined by $\sheaf T(X) = \set{[X_1|...|X_n]}{X_i \in \Smon_X}$, for $X\in \Ob S$. We have $\sheaf B_n = \Sring[\sheaf T]$ i.e. the free presheaf of $\sheaf  A$-modules generated by $\sheaf T$. Like in the case of groups or modules, there is a free-forgetful adjunction \cite[\href{https://stacks.math.columbia.edu/tag/03A8}{Lemma 03A8}]{stacks-project}
\begin{equation}
{}^{\sim}:\Hom_{\Sring} (\Sring[\sheaf T], \sheaf G) \overset\sim\rightarrow \Hom_{\cat{PSh}(\Sinf)} (\sheaf T, \sheaf G).
\end{equation}

To show that $\sheaf B_n$ is projective, one should establish the existence of an arrow $\eta:\Sring[\sheaf T] \to \sheaf C$ that makes the diagram 
$$
\begin{tikzcd}
& \Sring[\sheaf T] \ar[d, "\epsilon"] \ar[dl, " \eta", swap] \\
\sheaf C \ar[r, "\sigma", twoheadrightarrow ] & \sheaf  D
\end{tikzcd}
$$
in $\Mod(\Sring)$ commute, for any epimorphism $\sigma$ and morphism $\epsilon$. By the adjuntion, it suffices to show the existence of a morphism of presheaves $\tilde \eta : \sheaf T \to \sheaf C$ such that the diagram 
$$
\begin{tikzcd}
& \sheaf T \ar[d, "\tilde \epsilon"] \ar[dl, "\tilde \eta", swap] \\
\sheaf C \ar[r, "\sigma", twoheadrightarrow ] & \sheaf  D
\end{tikzcd}
$$
in $\cat{PSh}(\Sinf)$ commutes.

To define $\tilde\eta$, one has to determine the image of every symbol $[X_1|...|X_n]$, each time it appears in a set $\sheaf T(X)$.
Remark  that 
$$[X_1|...|X_n] \in \sheaf T(X) 
\Leftrightarrow (\forall i)( X \to X_i) 
\Leftrightarrow X\to X_1\cdots X_n = \prod_{i=1}^n X_i$$
The last equivalence is true due to the definition of  $\Sinf$. To solve the lifting problem, it is enough to pick $ m \in \sigma_{\prod_{i=1}^n X_i}^{-1}(\tilde  \epsilon ([X_1|...|X_n]))$, and define $\tilde \eta_{\prod_{i=1}^n X_i}([X_1|...|X_n]):= m$. This choice gives, by funtoriality, a well defined value $\tilde\eta_X([X_1|...|X_n])=\sheaf C\pi(m)$ over each $X$ such that $\pi:X\to \prod_{i=1}^n X_i$ in $\Sinf$.
\end{proof}

The existence of this projective resolution just depends on the definition of a conditional meet semilattice  (Definition \ref{def:info_structure}). It appears in the computation of classical and quantum information cohomology (see next section and \cite{Baudot2015}): the difference between these cases lies in the coefficients.

\begin{prop}\label{prop:chain_map} Given a conditional meet semilattice $\Sinf$ (resp. $\Sinf'$), let $\Smon$ (resp. $\Smon'$) denote the associated presheaf of monoids and $\Sring$ (resp. $\Sring'$) the presheaf of algebras induced by $\Smon$ (resp. $\Smon'$). 

For every morphism $\phi:\Sinf  \to \Sinf'$ between conditional meet semilattices such that $\phi(\Us)=\Us$, there are is a morphism of presheaf of monoids $\phi_*:\Smon \to \phi^*\Smon'$ given by $\phi_*^X:\Smon_X\to \Smon'_{\phi(X)}, \: Y\mapsto \phi(Y)$, for each $X\in \Ob \Sinf$. It can be extended linearly to $\phi_*:\Sring\to \Sring'$.

The transformation $\phi_*:\Sring\to \Sring'$ induces a map $\phi^*:\Mod(\Sring')\to \Mod(\Sring)$ as follows: given an $\Sring'$-module $\sheaf M'$,  $\Sring$ acts on $\phi^* \sheaf M'$ by the formula
\begin{equation}
\Sring_X\times \sheaf M'(\phi(X)) \to \sheaf M'(\phi(X)), \: (a,m)\to \phi(a)m.
\end{equation} 

Let $\sheaf B_n$ (resp. $\sheaf B'_n$) denote $(\sheaf A)^{\otimes n+1}$ (resp. $(\sheaf A')^{\otimes n+1}$), with an action of $\sheaf A$ (resp. $\sheaf A'$) by left multiplication on the first factor.  The maps of $\Sring$-modules $\Phi_*^n:\sheaf B_n \to \sheaf \phi^* B_n'$, whose components are
\begin{equation}
\Phi_*^n(X):\sheaf B_n(X) \to \sheaf B_n'(\phi(X)),\quad [Y_1|...|Y_n]\mapsto [\phi(Y_1)|...|\phi(Y_n)],
\end{equation}
define a morphism in $\cat{Ch}(\Mod(\Sring))$, the category of chain complexes of $\Sring$-modules. 
\end{prop}
\begin{proof}
The verifications are straightforward. They are left to the reader.
\end{proof}

\subsection{Description of cocycles}\label{sec:description_cocycles_general}\label{sec:description_cocycles}
We have built the projective resolution \eqref{eq:bar_resolution} of $\Rr_{\Sinf}$ in $\Mod(\Sring)$.
For every $\Sring$-module $\sheaf F$,  the information cohomology $H^\bullet(\Sinf, \sheaf F)$ can be computed as $\underline{\Ext}^n(\Rr_{\Sinf}, \sheaf F)$, defined in formulas \eqref{Ext_0} and \eqref{Ext_n} i.e. we deal with the cohomology of the  differential complex $(C^n(\Sinf,\sheaf F), \delta)$, where 
$$C^n(\Sinf,\sheaf F) :=\Hom_{\Sring} (\sheaf B_n, \sheaf F)$$ and $\delta$ is given by \eqref{eq:n-coboundary} bellow. A morphism $f$ in $C^n(\Sinf,\sheaf F)$  is called $n$-cochain. More explicitly, an $n$-cochain $f$ consists of a collection of morphisms  $f_X \in \Hom_{\Sring_X} (\sheaf B_n(X), \sheaf F_X)$ that satisfies the following conditions:
\begin{enumerate}
\item\label{locality} $f$ is a natural transformation (a functor of presheaves): given $\pi:X\to Y$, the diagram
$$\begin{tikzcd}
 \sheaf B_n(Y) \arrow[d, hook] \arrow[r, "f_Y"] & \sheaf F_Y \arrow[d, "\sheaf F(\pi)"] \\
 \sheaf B_n(X) \arrow[r, "f_X"] & \sheaf F_X
\end{tikzcd}$$
commutes. We refer to this property as \keyt{(joint) locality}, for reasons that become evident in the following section.
\item $f$ is compatible with the action of $\Sring$: for every $X\in \Ob \Sinf$,  the diagram
$$
\begin{tikzcd}
\Sring_X\times \sheaf B_n(X)\ar[r] \ar[d, "1\times f_X"] & \sheaf B_n(X) \ar[d, "f_X"] \\
\Sring_X\times \sheaf F_X \ar[r] & \sheaf F_X
\end{tikzcd}
$$
commutes. This means that $f_X$ is \keyt{equivariant}; in particular, $f_X(Y[Z]) = Y.f_X[Z]$ whenever $Y\in \Smon_X$.
\end{enumerate} 
Since $\sheaf B_n(X)$ is a free module, $f_X$ is determined by the values on the generators $[X_1|...|X_n]$. Just to simplify notation, we write $f_X[X_1|...|X_n]$ instead of $f_X([X_1|...|X_n])$. 

The coboundary of $f\in C^n(\Sinf,\sheaf F)$ is the $(n+1)$-cochain $\delta f = f \partial : \sheaf B^{n+1} \to \sheaf F$; \eqref{eq:n-coboundary} gives a more explicit description.
As customary, a cochain $f\in C^n(\Sinf,\sheaf F)$ is called an \keyt{$n$-cocycle}  when $\delta f=0$; the submodule of all $n$-cocycles is denoted by $Z^n(\Sinf, \sheaf F)$. The image under $\delta$ of $C^{n-1}$ is another submodule of $C^n(\Sinf,\sheaf F)$, denoted $\delta C^{n-1}(\Sinf, \sheaf F)$; its elements are called \keyt{$n$-coboundaries}. By definition, $\delta C^{-1}(\Sinf,\sheaf F) = \langle 0\rangle$, the trivial module. Since $\delta^2 =0$,  $\delta C^{n-1}$ is a submodule of $Z^n$. With this notation,  $H^n(\Sinf, \sheaf F) = Z^n(\Sinf,\sheaf F)/\delta C^{n-1}(\Sinf, \sheaf F)$, for every $n\geq 0$.

\section{Probabilistic information cohomology}\label{sec:classical_info_cohomology}

In this section, all information structures are supposed to be finite (see Section \ref{sec:general_structure}). The treatment of continuous classical random variables presents many technical complications; the particular case of  gaussian laws  is the subject of \cite[Part~IV]{Vigneaux2019-thesis}.

 We introduce probabilities as a covariant functor  on an information structure; the measurable, real-valued probabilistic functions of probabilities (\emph{probabilistic functionals}) form a presheaf $\sheaf F_\alpha$ of $\sheaf A$-modules; the action of $\sheaf A$ depends on a positive parameter $\alpha$. We compute information cohomology with coefficients in $\sheaf F_\alpha$: the corresponding $\alpha$-entropy appears as the unique $1$-cocycle on each ``connected component'' of the structure.

\subsection{Probabilities}\label{sec:probas_on_structures}

Given a finite information structure $(\Sinf,{\sheaf E})$, let  $\FProb: \Sinf\to \Sets$ be the functor that associates to each $X\in \Ob\cat S$ the set
  \begin{equation}
 \FProb(X) := \bigset{p:{\sheaf E}_X\to [0,1]}{\sum_{x\in {\sheaf E}_X}p(x) = 1},
 \end{equation}
 of probability laws for $X$, and to each arrow $\pi:X\to Y$ the \emph{marginalization} map $\FProb\pi:\FProb(X)\to \FProb(Y)$, also denoted $\pi_*$, given by
  \begin{equation}\label{eq:general_marginalizations}
\forall P\in  \FProb_X ,\:\forall y\in {\sheaf E}_Y\quad \FProb\pi(P)(y) =\sum_{x\in {\sheaf E\pi}^{-1}(y) } P(x).
 \end{equation}

  We adopt the probabilistic notation, in the following sense: given an arrow $\pi_{YX}:X\to Y$ in $\Sinf$, a law   $P\in  \FProb_X$,  and $y\in {\sheaf E}(Y)$,  the notation $P(Y=y)$ means $P({{\sheaf E}\pi_{YX}}^{-1}(y))={\pi_{YX}}_*P(y)$; similarly, if 
  $\begin{tikzcd}
  Y  & X \ar[l, "\pi_{YX}",  swap] \ar[r, "\pi_{ZX}"] & Z
  \end{tikzcd}$
  is a diagram in $\Sinf$, the notation $P(Y=y,Z=z)\equiv P(\{Y=y\}\cap\{Z=z\})$ means $P({\sheaf E\pi_{YX}}^{-1}(y)\cap {\sheaf E\pi_{ZX}}^{-1}(z))$, which equals $P({\sheaf E\langle {\pi_{YX}}, {\pi_{ZX}}\rangle}^{-1}(w(y,z)))$ for the unique $w(y,z)\in {\sheaf E}_{YX}$ sent to $(y,z)\in {\sheaf E}_Y\times {\sheaf E}_Z$ by the injection in Definition \ref{def:info_structure}-\ref{wedge_coondition}. 
  
  Given an arrow $\pi_{ZX}:X\to Z$, a law $P\in\FProb(X)$, and $z\in \sheaf E_Z$ such that $P(Z=z)>0$, the conditional law $P|_{Z=z}$ is defined by
 \begin{equation}
 P|_{Z=z}(x):=\frac{P(\{X=x	\}	\cap \{Z=z\})}{P(Z=z)} = \frac{P(x\in {\sheaf E\pi_{ZX}}^{-1}(z))}{P(Z=z)}. 
 \end{equation}
  
  Conditioning commutes with marginalizations: given arrows $\pi_{YX}:X\to Y$ and $\pi_{ZY}:Y \to Z$,
  \begin{align*}
\pi^{YX}_* (P|_{Z=z}) (y) &= \sum_{x\in {\sheaf E\pi_{YX}}^{-1}(y)} \frac{P(\{x\}\cap {\sheaf E\pi_{ZX}}^{-1}(z))}{P(Z=z)} =  \frac{\sum_{x\in {\sheaf E\pi_{YX}}^{-1}(y)} P(\{x\}\cap {\sheaf E\pi_{ZX}}^{-1}(z))}{P(Z=z)}\\
&=  \frac{P({\sheaf E\pi_{YX}}^{-1}(y)\cap {\sheaf E\pi_{YX}}^{-1}({\sheaf E\pi_{YZ}}^{-1}(z)))}{P(Z=z)} =   \frac{\pi^{YX}_*P(y\cap {\sheaf E\pi_{YZ}}^{-1}(z))}{\pi^{YX}_*P(Z=z)}\\
&= (\pi^{YX}_*P)|_{Z=z}(y).
\end{align*}

  More generally, an \emph{adapted probability functor} $\sheaf Q: \Sinf \to \Sets$ on an information structure $(\Sinf,{\sheaf E})$ is a subfunctor of $\FProb$ that is stable under conditioning: for every arrow $X\to Z$ in $\Sinf$, every law  $P\in \sheaf Q_X$, and every $z\in {\sheaf E}_Z$ such that $P(Z=z)>0$, the law $P|_{Z=z}$ belongs to $\sheaf Q_X$. 
  
  For instance, a functor $\sheaf Q$ that associates to each $X$ a (geometric) simplicial subcomplex $\sheaf Q_X$ of the probability  simplex $\FProb_X \subset \Rr^{E_X}$ is adapted (see \cite[Sec.~2.1]{Baudot2015}); the restricted marginalizations are simplicial maps. In \cite[Sec.~2.4]{Baudot2015}, it is argued that such functors model exclusion rules between possible events.

\subsection{Functional module} \label{functional_module}



Let $(\Sinf,{\sheaf E})$ be an information structure, and $\sheaf Q$ an adapted probability functor. For each $X\in \Ob{\Sinf}$, let  $\sheaf F_X = \sheaf F_X(\sheaf Q)$ be the real vector space of measurable functions on $\sheaf Q_X$; we call it  \text{functional space}. For each arrow $\pi:X\to Y$ in $\Sinf$, there is a morphism $\pi^*:\sheaf F_Y \to \sheaf F_X$ defined  by $\pi^*f (P_X) = f(\pi_* P_X).$ Therefore, $\sheaf F$ is a contravariant functor from $\Sinf$ to the category of real vector spaces.

The functional space $\sheaf F_X$ admits an action of the monoid $\Smon_X$ (parameterized by $\alpha>0$): for $Y\in  \Smon_X$,  and $f \in \sheaf F_X$, the new function $Y.f$ is given by
\begin{equation}\label{action_monoid}
\forall P \in \sheaf Q_X, \quad (Y.f)(P) = \sum_{\substack{y\in {\sheaf E}_Y\\Y_*P(y)\neq 0}} (Y_*P(y))^\alpha f(P|_{Y=y}).
\end{equation}
 By Proposition \ref{iterated_action}, there is a morphism of monoids $\Smon_X \to \End(\sheaf F_X)$, given by Equation \eqref{action_monoid}, that extends by linearity to a morphism of rings $
\Lambda_\alpha(X):\Sring_X \to \End(\sheaf F_X).
$
This means that, for each $\alpha > 0$, $\sheaf F_X$ has the structure of a $\Sring_X$-module, denoted $\sheaf F_\alpha(X)$.\footnote{As $\Sring_X$ is a $\Rr$-algebra, it comes with an inclusion $f_X:\Rr \to \Sring_X, \:r\mapsto r \Us_{\Sinf}$. The composite $\Lambda_\alpha(X) \circ f_X$ gives an action of $\Rr$ over $\sheaf F_X$, that coincides with the usual multiplication of functions by scalars.}

\begin{prop}\label{iterated_action}
Given any $X\in \Ob{\Sinf}$, $Y,Z\in \Smon_X$, and $f\in \sheaf F(\sheaf Q_X)$, the identities
 $$\Us.f = f \quad \text{and} \quad (ZY).f = Z.(Y.f)$$
 hold.
\end{prop}
\begin{proof}
Since it is obvious that $\Us.f = f$, we only prove the other.  

The universal property of products gives the commutative diagram:
$$
\begin{tikzcd}[column sep=huge, row sep=large]
	& X \ar[ld, swap, "\rho_Y"] \ar[rd, "\rho_Z"] \ar[d, "{\langle \rho_Y, \rho_Z \rangle}" near end] 
 		&  \\ 
 Y 
 	&  Y Z \ar[l, "\pi_Y"] \ar[r, swap, "\pi_Z"]
 		& Z
 	\end{tikzcd}
 	$$
Equation \eqref{action_monoid} directly implies that, for any $P\in \sheaf Q_X$,
\begin{equation}
Z.(Y.f)(P)  = \sum_{\substack{z\in {\sheaf E}_Z\\ Z_*P(z)\neq 0}}  P(Z=z)^\alpha \sum_{\substack{y \in {\sheaf E}_Y\\Y_*P|_{Z=z}(y)\neq 0}} (P|_{Z=z}(Y=y))^\alpha f((P|_{Z=z})|_{Y=y})
\end{equation}
By definition of conditional probabilities,  $$P(Z=z)P|_{Z=z}(Y=y)=P(\{Y=y\}\cap\{Z=z\}).$$ The pairs $(y,z)$ that appear in the sum are such that $P(\{Y=y\}\cap\{Z=z\}) \neq 0$, so $P(Y=y)$ and $P(Z=z)$ are different from zero; in this case, for any $B\subset X$,  $$(P|_{Z=z})|_{Y=y}(B) = \frac{P|_{Z=z}(B\cap\{Y=y\})}{P|_{Z=z}(Y=y)} = \frac{P(B\cap\{Y=y\}\cap\{Z=z\})}{P(\{Y=y\}\cap\{Z=z\})} = P|_{Y=y,Z=z}(B).$$ 

Finally, Definition \ref{def:info_structure} guarantees that the nonempty sets $$\{Y=y\}\cap\{Z=z\} = {\sheaf E\rho_Y}^{-1}(y) \cap {\sheaf E\rho_Z}^{-1}(z) \subset {\sheaf E}_X$$ are the preimage under ${\sheaf E\langle\rho_Y, \rho_Z\rangle}$ of a \emph{unique} element $w_{y,z}\in {\sheaf E}_{YZ}$; moreover, for every element $w\in {\sheaf E}(Y Z)$ we find such set.  Therefore,
$$Z.(Y.f)(P)= \sum_{\substack{w_{y,z}\in {\sheaf E}_{YZ}\\ YZ_*P(w)\neq 0}} (YZ_*P(w_{y,z}))^\alpha f(P|_{Z=z,Y=y})=(ZY).f(P).$$
\end{proof}

The next proposition shows that this action is compatible with the morphisms between functional modules. Hence, the sheaf $\sheaf F_\alpha(\sheaf Q)$ belongs to $\Mod(\Sring)$, and can be used as coefficients in information cohomology. 
\begin{prop}\label{prop:functoriality_monoidal_action}
 Given $\pi_{YX}:X\to Y$ and $\pi_{ZY}:Y \to Z$, the action of $Z$ makes the following diagram commute
$$\begin{tikzcd}
 \sheaf F(\sheaf Q_Y) \arrow[d, "\pi_{YX}^*"] \arrow[r, "Z"] & \sheaf F(\sheaf Q_Y) \arrow[d, "\pi_{YX}^*"] \\
 \sheaf F(\sheaf Q_X)  \arrow[r, "Z"] & \sheaf F(\sheaf Q_X) 
\end{tikzcd}$$
\end{prop}
\begin{proof}
We must prove that, for all $f_Y\in \sheaf F(\sheaf Q_Y), P\in \sheaf Q_X$, the equality $(Z.f_Y)(\pi^{YX}_* P) = Z.(f_Y \circ \pi_*^{YX})(P)$.
On the one hand,
\begin{equation}
(Z.f_Y)(\pi^{YX}_* P) = \sum_{\substack{z\in {\sheaf E}_Z\\ \pi^{ZY}_*\pi^{YX}_* P(z)\neq 0}} \pi^{ZY}_*\pi^{YX}_* P(z) f_Y((\pi^{YX}_* P)|_{Z=z}),
\end{equation}
and on the other,
\begin{equation}
Z.(f_Y \circ \pi_*^{YX})(P) = \sum_{\substack{z\in {\sheaf E}_Z\\ \pi^{ZX}_* P(z)\neq 0}} \pi^{ZX}_* P(z) f_Y(\pi^{YX}_* (P|_{Z=z})).
\end{equation}
The two expressions coincide since marginalizations are functorial, $\pi^{ZY}_*\pi^{YX}_* = \pi^{ZX}_*$,  and commute with conditioning (cf. Section \ref{sec:probas_on_structures}).
\end{proof}

\subsection{Functoriality}\label{sec:basic_prop}

Let $\phi=(\phi_0,\widehat\phi):(\Sinf,{\sheaf E}) \to (\Sinf',{\sheaf E}')$ be a morphism between finite information structures, and let $\sheaf Q$ be a probability functor on $\Sinf$ and $\sheaf Q'$, a probability functor on $\Sinf'$ (not necessarily adapted).  Given a $X\in \Ob\cat S$ and a law $P\in \sheaf Q_X$, define a law $\phi_*^XP$ on ${\sheaf E}'_{\phi_0(X)}$ by the equation
\begin{equation}\label{def_m}
\forall x' \in {\sheaf E}_{\phi_0(X)}, \quad \phi^X_*P(x') = \sum_{x\in {\widehat\phi_X}^{-1}(x')} P(x).
\end{equation}
We call this operation \emph{external marginalization}.

\begin{lem}
Provided that  for all $X\in \Ob\cat S$ and all $P\in \sheaf Q_X$ the law $\phi^X_*P$ belongs to $\sheaf Q'_{\phi_0(X)}$,  \eqref{def_m} defines a natural transformation  $\phi_*:\sheaf Q \to \phi_0^* \sheaf Q'  $  i.e. (internal) marginalization and external marginalization commute.
\end{lem}
\begin{proof}
For every arrow $\pi:X\to Y$ and $y'\in {\sheaf E}_{\phi(Y)}$, 
\begin{align*}
(\sheaf Q'(\phi_0\pi)(\phi_*^XP))(y') &\eq{\text{def. } \sheaf Q'}  \sum_{x'\in {(\sheaf E}'\phi_0\pi)^{-1}(y')} (\phi_*^X(P))(x')\\
 & \eq{\text{def. } \phi_*^X} \sum_{x'\in ({\sheaf E}'\phi_0\pi)^{-1}(y')} \sum_{x\in {\widehat\phi_X}^{-1}(x')} P(x)\\
 &= \sum_{x\in ({\sheaf E}\phi_0\pi \circ {\widehat \phi_X})^{-1}(y')} P(x) \\
  &= \sum_{x\in ( {\widehat \phi_Y} \circ {\sheaf E}\pi )^{-1}(y')} P(x) \\
  &= \sum_{y \in {\widehat\phi_Y}^{-1}(y')} \sum_{ x\in {\sheaf E}\pi^{-1}(y)} P(x) \\
  &\eq{\text{def. } \phi_*^Y} \phi_*^Y(\sheaf Q(\pi)(P)).
\end{align*}
The fourth equality comes from the naturality of $\widehat \phi$, cf. Definition \ref{def:morph_str}.
\end{proof}

Using the previous lemma, it is easy to verify that $\phi_*$ induces a natural transformation $\phi^*:\phi_0^* \sheaf F'   \to \sheaf F$---where $\sheaf F_\alpha = \sheaf F(\sheaf Q)$ and $\sheaf F' = \sheaf F_\alpha(\sheaf Q')$---that maps $f'\in \sheaf F'_{\phi_0(X)}$ to $f'\circ \phi_*^X$. 

Some morphisms of information structures induce morphisms at the level of cohomology.

\begin{prop}\label{prop:functoriality}
Let $\phi=(\phi_0,\widehat\phi):(\Sinf,{\sheaf E}) \to (\Sinf',{\sheaf E}')$ be a morphism of information structures; let $\sheaf Q$ (resp. $\sheaf Q'$) be an adapted probability functor on $\Sinf$ (resp. $\Sinf'$).  Suppose that
\begin{enumerate}
\item for all $X\in \Ob\cat S$, the map $\widehat\phi_X$ is a bijection, and
\item for all $X\in \Ob\cat S$ and all $P\in \sheaf Q_X$, the law $\phi_*^XP$ belongs to $\sheaf Q'_{\phi(X)}$. 
\end{enumerate} 
Then, there exist a cochain map 
\begin{equation}
\Phi^*_\bullet: (C^\bullet(\Sinf',\sheaf F'_\alpha), \delta)\to (C^\bullet(\Sinf,\sheaf F_\alpha), \delta),
\end{equation}
where $\sheaf F_\alpha = \sheaf F_\alpha(\sheaf Q)$ and $\sheaf F'_\alpha = \sheaf F_\alpha(\sheaf Q')$, given by the formula
\begin{equation}\label{eq:functoriality_under_S_morphism}
\forall f'\in  C^n(\Sinf',\sheaf F'_\alpha), \quad (\Phi_n^* f')_Y[X_1|...|X_n](P) := f'_{\phi_0(Y)}[\phi_0(X_1)|...|\phi_0(X_n)](\phi_*^Y(P)).
\end{equation}
The chain map induces a morphism of graded vector spaces in cohomology
\begin{equation}\label{eq:functoriality_morphism_coh_general}
\Phi^*_\bullet: H^\bullet(\Sinf', \sheaf F_\alpha')\to H^\bullet(\Sinf, \sheaf F_\alpha).
\end{equation}
\end{prop}
\begin{proof}
From Proposition \ref{prop:chain_map}, we know that $\phi$ induces a morphism $$\{\Phi_*^n:\sheaf B_n\to \sheaf B'_n \circ \phi_0\}_{n}$$ of chain complexes of $\Sring$-modules. 

Let $f'$ be an element of $\Hom_{\Sring'}(\sheaf B_n', \sheaf F'_\alpha)$. 

First, we claim that $\Phi_n^* f'$ is a natural transformation. Given an arrow $\pi:X\to Y$ in $\Sinf$, the diagram 
$$
\begin{tikzcd}
\sheaf B_n(Y) \ar[r, "\Phi_*^n(Y)"]\ar[d, hook] & \sheaf B'_n(\phi_0(Y)) \ar[d, hook] \ar[r, "f'_{\phi_0(Y)}"] & \sheaf F'(\phi_0(Y)) \ar[d, "\sheaf F'\phi_0\pi"] \ar[r, "\phi^*_Y"] & \sheaf F(Y) \ar[d, "\sheaf F\pi"]\\
\sheaf B_n(X) \ar[r, "\Phi_*^n(X)"] & \sheaf B'_n(\phi_0(X)) \ar[r, "f'_{\phi_0(X)}"] & \sheaf F'(\phi_0(X)) \ar[r, "\phi^*_X"] & \sheaf F(X)
\end{tikzcd}
$$
in $\cat{PSh}(\Sinf)$ commutes, because $\Phi_*^n$, $f'$, and $\phi^*$ are natural transformations. Since the composition of three successive horizontal arrows give the components $(\Phi_n^* f)_X$ and $(\Phi_n^* f)_Y$, this implies the claim. 

Second,  $\Phi_n^* f'$ is a morphism of $\Sring$-modules, which means that for every $X\in \Ob \Sinf$ the diagram
$$
\begin{tikzcd}
\Sring_X \times \sheaf B_n(X) \ar[r]\ar[d, "1\times (\Phi_n^* f')_X"] & \sheaf B_n(X) \ar[d, "(\Phi_n^* f')_X"] \\
\Sring \times \sheaf F(X) \ar[r] & \sheaf F(X)
\end{tikzcd}
$$
commutes i.e. for all $X_0,...,X_n\in \Smon_X$,
\begin{equation}\label{eq:equivariance_f'}
\phi^*_X(f'_{\phi_0(X)}(\phi_0(X_0)[\phi(X_1)|...|\phi(X_n)])) = X_0.(\phi^*_Xf'_{\phi_0(X)}([\phi(X_1)|...|\phi(X_n)]))
\end{equation}
as functions on $\sheaf Q_X$. For any $P$ in $\sheaf Q_X$, one has
\begin{align}
(\phi^*_X&(f'_{\phi_0(X)}(\phi_0(X_0)[\phi(X_1)|...|\phi(X_n)])))(P)\nonumber \\ &= (f'_{\phi_0(X)}(\phi_0(X_0)[\phi(X_1)|...|\phi(X_n)]))(\phi_*^XP) \nonumber \\
&= (\phi_0(X_0).f'_{\phi_0(X)}([\phi(X_1)|...|\phi(X_n)]))(\phi_*^XP) \nonumber \\
&= \sum_{x'_0\in {\sheaf E}'_{\phi_0(X_0)}} \phi_*^XP(\phi_0(X_0)=x_0) f'_{\phi_0(X)}([\phi(X_1)|...|\phi(X_n)])((\phi_*^XP)|_{\phi_0(x_0)= x_0'}),\label{proof:sum_A-equivariance}
\end{align}
where the first equality is implied by the definition of $\phi^*$, the second by the $\Sring'$-equivariance of $f'$, and the third by the definition of the action of $\Smon_X'$ on $\sheaf F'$. Since $\widehat \phi_X$ is a bijection, each  $x_0'\in {\sheaf E}'_{\phi_0(X_0)}$ is in correspondence with an element $x_0 = x_0(x_0')$ of ${\sheaf E}_{X_0}$, and  $\phi_*^XP(\phi_0(X_0)=x_0')=P({\widehat \phi_X}^{-1}(({\sheaf E}'\phi_0\pi_{X_0X})^{-1}(x_0')))$ equals $P(X_0=x_0)=P(({\sheaf E}\pi_{X_0X})^{-1}({\widehat\phi_{X_0}}^{-1}(x_0')))$ as a simple consequence of the naturality of $\widehat \phi$, which applied to $\pi_{X_0X}:X\to X_0$ gives the commutative diagram
$$
\begin{tikzcd}
{\sheaf E}_X \ar[r, "\widehat\phi_X"]  \ar[d, "{\sheaf E}\pi_{X_0X}", swap] & {\sheaf E}_{\phi_0(X)} \ar[d, "{\sheaf E}'\phi_0\pi_{X_0X}"]\\
{\sheaf E}_{X_0} \ar[r, "\widehat\phi_{X_0}"] & {\sheaf E}_{\phi_0(X_0)}
\end{tikzcd}.
$$
So, provided that conditioning commutes with external marginalization, i.e. \begin{equation}\label{eq:commm_ext_marg_cond}
(\phi_*^XP)|_{\phi_0(x_0)= x_0'}= \phi_*^X(P|_{X_0=x_0}),
\end{equation} the sum in \eqref{proof:sum_A-equivariance} can be rewritten as
\begin{equation}
\sum_{x_0\in {\sheaf E}_{X_0}} P(X_0=x_0) f'_{\phi_0(X)}([\phi(X_1)|...|\phi(X_n)])(\phi_*^X(P|_{X_0=x_0})),
\end{equation}
which is precisely the right-hand side of \eqref{eq:equivariance_f'} evaluated at $P$. Hence $\Phi_n^* f'$ is  a morphism of $\Sring$-modules.

To establish \eqref{eq:commm_ext_marg_cond}, remark that for any $z'\in {\sheaf E}_{\phi_0(X)}$ there is a unique $z\in {\sheaf E}_X$ such that $\widehat \phi_X(z) = z'$, and
\begin{align*}
(\phi_*^XP)|_{\phi_0(X_0)=x_0'} (z') &= \frac{\phi_*^XP(z'\in ({\sheaf E}'\phi_0\pi_{X_0X})^{-1}(x_0'))}{\phi_*^XP(\phi_0(X_0)=x_0')} \\
&= \frac{P(z\in {\widehat\phi_X}^{-1}(({\sheaf E}'\phi_0\pi_{X_0X})^{-1}(x_0')))}{P(X_0=x_0)} \\
&= \frac{P(z\in ({\sheaf E}\pi_{X_0X})^{-1}({\widehat \phi_{X_0}}^{-1}(x'_0)))}{P(X_0=x_0)},
\end{align*}
where the first equality follows from the definition of conditioning, the second from the definition of external marginalization, and the third from the naturality of $\widehat \phi$. Since ${\widehat \phi_{X_0}}^{-1}(x'_0)=x_0$, the last expression equals $P|_{X_0=x_0}(z)=(\phi_*^X(P|_{X_0=x_0}))(z')$.

Finally,  the commutativity of the diagram   
$$
\begin{tikzcd}
\sheaf B_{n+1}(X) \ar[r, "(\partial_{n+1})_X"]\ar[d, "(\Phi_*^{n+1})_X"]\arrow[rr, bend left=20, "(\delta^n \Phi_n^*f')_X"]
 & \sheaf B_n(X) \ar[r, "(\Phi_n^* f')_X"] \ar[d, "(\Phi_*^{n})_X"] & \sheaf F(X) \\
\sheaf B'_{n+1}(\phi_0(X)) \ar[r, "(\partial_{n+1})_X"] \arrow[rr, bend right=20, swap, "(\delta^n f')_{\phi_0(X)}"] & \sheaf B'_n(\phi_0(X)) \ar[r, "f'_{\phi_0(X)}"] & \sheaf F' (\phi_0(X)) \ar[u, "\phi_X^*"],
\end{tikzcd}
$$
---which follows from the definitions and the naturality of $\Phi_*^n$---entails
$$(\delta^n f')_X = \phi_X^* \circ (\delta^n f')_{\phi_0(X)} \circ (\Psi_*^{n+1})_X =: (\Phi_{n+1}^*(\delta^n f))_X.$$
In other words,  $\Phi_\bullet^*$ and $\delta$ commute, thus $\Phi_\bullet^*$ is a cochain map and as such it induces a morphism at the level of cohomology.
\end{proof}

\begin{remark}
TAC's referee has suggested that the pullback functor $\phi^*:\Mod(\Sring')\to \Mod(\Sring)$ in Proposition \ref{prop:functoriality_monoidal_action} might be exact and induce a transformation $H^\bullet (\Sinf', \sheaf F')\to H^\bullet(\Sinf, \phi^*\sheaf F)$ for every $\sheaf A'$-module $\sheaf F'$. They conjecture that the composition of such transformation with another map coming from the functoriality of $H^\bullet (\Sinf,-)$ may then reproduce the morphism $\Phi_\bullet^*$ in \eqref{eq:functoriality_morphism_coh_general}, avoiding explicit computations with cocycles. 
\end{remark}

\begin{cor}\label{cor:invariance_under_iso}
If $\phi:\Sinf \to \Sinf'$ is an isomorphism of information structures, and $\sheaf Q'(\phi(X)) = \phi_*^X(\sheaf Q_X)$ for every $X\in \Ob\cat S$, then $\Phi^*: H^\bullet(\Sinf', \sheaf F_\alpha(\sheaf Q'))\to H^\bullet(\Sinf, \sheaf F_\alpha(\sheaf Q))$ is an isomorphism too.
\end{cor}
\begin{proof}
For every $X\in \Ob\cat S$, one has $\widehat {\phi^{-1}}_{\phi(X)} \circ \widehat\phi_X = \id_X$ and $  \widehat\phi_X \circ \widehat {\phi^{-1}}_{\phi(X)} = \id_{\phi(X)}$, thus $\widehat\phi_X$ and $\widehat {\phi^{-1}}_{\phi(X)}$  are bijections. Proposition \ref{prop:functoriality} ensures the existence of $\Phi^*_\bullet: H^\bullet(\Sinf', \sheaf F_\alpha(\sheaf Q'))\to H^\bullet(\Sinf, \sheaf F_\alpha(\sheaf Q))$ and $(\Phi^{-1})^*_\bullet: H^\bullet(\Sinf, \sheaf F_\alpha(\sheaf Q))\to H^\bullet(\Sinf', \sheaf F_\alpha(\sheaf Q'))$ that are inverse to each other, which can be easily verified using   \eqref{eq:functoriality_under_S_morphism}.
\end{proof}

We  recover from Proposition \ref{prop:functoriality} a functorial property for concrete information structures.

\begin{prop}\label{Hmap_contravariant}
Consider  concrete information structures $\Sinf\subset \ObsFin(\Omega)$ and  $\Sinf'\subset \ObsFin(\Omega')$. Let $\sheaf Q$ (resp. $\sheaf Q'$) be an adapted probability functor defined on $\Sinf$ (resp. $\Sinf'$). Let $\mathcal \sigma:\Omega \to \Omega'$ be a \emph{surjective}  function. Suppose that for all $X\in \Ob{\Sinf}$, there exists $X'\in \Ob\cat {S'}$ such that $\sigma$ descends to a bijection $\sigma_X:\Omega/X \overset{\sim}{\to} \Omega'/X'$, \footnote{Every partition $X$ defines an equivalence relation and $\Omega/X$ denotes the corresponding quotient.}  and that for every $P\in  \sheaf Q_X$, the marginalization ${\sigma_X}_*P $ is in $ \sheaf Q'_{X'}$.

Then the observable $X'$ is uniquely determined by $X$, the correspondence $\phi_0:X\mapsto X'$ defines a morphism of information structures $\phi:(\Sinf,\square)\to (\Sinf',\square)$---see the notation in Example \ref{ex:classical_structure}---, and there exists a morphism of graded vector spaces
\begin{equation}
\sigma^* : H^\bullet(\Sinf',\sheaf F_\alpha(\sheaf Q')) \to H^\bullet(\Sinf, \sheaf F_\alpha(\sheaf Q)),
\end{equation}
defined at the level of cochains by \eqref{eq:functoriality_under_S_morphism}, \emph{mutatis mutandis}.
\end{prop}
\begin{proof}
Since $\sigma$ is surjective, $\sigma^{-1}(A) \neq \sigma^{-1}(B)$ unless $A=B$, cf. footnote \ref{foot:surjections}. If $X', X''$ of $\Omega'$ satisfying $F:\Omega/X \overset{\sim}{\to} \Omega'/X'$ and $G:\Omega/X \overset{\sim}{\to} \Omega'/X''$, then for every $A\in X$,
$A = \sigma^{-1}(F(X)) =\sigma^{-1}(G(X))$, therefore $X'=X''$. The correspondence $X\mapsto X'$ defines a functor  $\phi_0:\Sinf \to \Sinf'$: given an arrow $\pi:X\to Y$ in $\Sinf$, there is a corresponding surjection $\square\pi:\Omega/X \to \Omega/Y$---using an obvious identification between each part and the corresponding class in the quotient---and $\sigma_Y \circ \square \pi \circ \sigma_X^{-1}: \Omega'/\phi_0(X) \to \Omega'/\phi_0(Y)$ is also a surjection, that corresponds to an arrow $\phi(X) \to \phi(Y)$ in $\Sinf'$. We take as $\widehat \phi_X:X\to \phi(X)$ the bijection of partitions induced by $\sigma_X$.  The condition on the probabilities implies that $\phi_*: \sheaf Q \to \sheaf Q'\circ \phi_0$ is a natural transformation. Proposition \ref{prop:functoriality} gives the desired result.
\end{proof}

\begin{prop}\label{Hmap_covariant}
Consider  concrete information structures $\Sinf\subset \ObsFin(\Omega)$ and  $\Sinf'\subset \ObsFin(\Omega')$. Let $\sheaf Q$ (resp. $\sheaf Q'$) be an adapted probability functor defined on $\Sinf$ (resp. $\Sinf'$). Let $\mathcal \eta:\Omega \to \Omega'$ be a  function. Suppose that for all $X'\in \Ob{\Sinf}$, there exists $X\in \Ob\cat {S}$ such that $\eta$ descends to a bijection $\eta_X:\Omega/X \overset{\sim}{\to} \Omega'/X'$,  and that for every $P'\in \sheaf Q'_{X'}$, there exists $P\in  \sheaf Q_{X}$ with $P' = {\eta_X}_* P$.

Then the observable $X$ is uniquely determined by $X'$, the correspondence $\phi_0:X'\mapsto X$ defines a morphism of information structures $\phi: (\Sinf',\square)\to (\Sinf,\square)$, and there exists a  morphism of graded vector spaces
\begin{equation}
\eta_*: H^m(\Sinf, \sheaf F_\alpha( \sheaf Q)) \to H^m(\Sinf', \sheaf F_\alpha( \sheaf Q')),
\end{equation}
defined at the level of cochains by \eqref{eq:functoriality_under_S_morphism}, \emph{mutatis mutandis}.
\end{prop}
\begin{proof}
In view of the component-wise bijection, $X$ is the partition $\set{\eta_X^{-1}(B)}{B\in X'}$. The rest of the proof is analogous to the last one. 
\end{proof}

\subsection{Determination of $H^0$} 
Each $0$-cochain $f[\,]\equiv f$ corresponds to a collection of functions $f_X(P_X) \in \sheaf F_\alpha(\sheaf Q_X)$, for each $X\in \Ob{\Sinf}$, that satisfy $f_Y(Y_* P_X) = f_X(P_X)$ for any arrow $X\to Y$ in $\Sinf$. Since  $\Us\in \Ob{\Sinf}$, this means that $f$ is constant. Given an arrow $X\to Y$, and a $0$-cochain $f$ such that $f_X(P) = K$, 
\begin{align*}
(\delta f)_X[Y](P) &= Y.f_X(P) - f_X(P) = \sum_{y\in {\sheaf E}_Y} P(Y=y)^\alpha f(P|_{Y=y}) - f(P) \\
& = K \left( \sum_{y\in {\sheaf E}_Y} P(Y=y)^\alpha -1 \right)=0.
\end{align*}
Thus $Z^0(\Sinf, \sheaf F_1(\sheaf Q)) = C^0(\Sinf,\sheaf F_1(\sheaf Q))\cong \Rr$ and $Z^0(\Sinf,\sheaf F_\alpha(\sheaf Q)) = \langle 0 \rangle$ when $\alpha \neq 1$ (as long as some $\sheaf Q_Y$ contains a nonatomic probability). Therefore, $H^0(\Sinf, \sheaf F_1(\sheaf Q)) \cong \Rr$, and $H^0(\Sinf, \sheaf F_\alpha(\sheaf Q))\cong \langle 0 \rangle$ when $\alpha\neq 1$.

\subsection{Local structure of 1-cocycles}\label{sec:entropy}


Now we turn to  $C^1(\Sinf,\sheaf F_\alpha(\sheaf Q))$. The $1$-cochains are families $\{f_X[Y] \mid X\in\Ob{\Sinf}\}$ such that for all $Z \to X \to Y$, the equality $f_X[Y](X_*P_Z) = f_Z[Y](P_Z)$ holds (this is the \emph{locality} in Section \ref{sec:description_cocycles}). This means that it is sufficient to know $f_Y[Y](Y_*P)$ to recover $f_X[Y](P)$, for any $X\to Y$; in this sense, we usually omit the subindex and just write $f[Y]$. It follows that there is a bijective correspondence between $1$-cochains $f$ and collections of measurable functions $\{ f[X]:\sheaf Q_X \to \Rr\}_{X\in \Ob \Sinf}$.

 The computations in the previous section imply that $\delta C^0(\Sinf,\sheaf F_1(\sheaf Q)) = \langle 0 \rangle$, whereas $\delta C^0(\Sinf,\sheaf F_\alpha(\sheaf Q))\cong\Rr$ when $\alpha \neq 1$: in this case $1$-coboundaries are  multiples of the  $1$-cochain defined by $ S_\alpha[X]$ in \eqref{eq:tsallis-entropy-def}. We write $
\delta C^0(\Sinf,\sheaf F_\alpha(\sheaf Q)) \cong \Rr \cdot S_\alpha.
$

We refer to elements of $Z^1(\Sinf, \sheaf F_\alpha(\sheaf Q))$ as $1$-cocycles of type $\alpha$. By equation \eqref{eq:n-coboundary} and commutativity of the product, every $1$-cocycle must satisfy the following symmetric equation
\begin{equation}\label{symmetry_cocycle}
f[XY] = f[Y] + Y.f[X] = f[X] + X.f[Y].
\end{equation}

\begin{prop}\label{lemma_certitude}
Let $f$ be a $1$-cocycle. 
\begin{enumerate}
\item \label{lemma_certitude_1} For every $X\in \Ob \Sinf$, if $|\sheaf E_X|=1$, then $f[X] \equiv 0$. In particular, $f[\Us] \equiv 0$.
\item \label{lemma_certitude_2} For every $X\in \Ob{\Sinf}$ and $x\in {\sheaf E}_X$,  the equality $f[X](\delta_{x}) = 0$ holds.
\end{enumerate}
\end{prop}
\begin{proof}
Statement \eqref{lemma_certitude_1} is a particular case of \eqref{lemma_certitude_2}; we prove the later.  From $f[XX] = f[X] + X.f[X]$, we conclude  that $X.f[X](P) = \sum_{x\in {\sheaf E}_X \mid P(x)\neq 0} P(x)^\alpha f[X](P|_{X=x}) = 0$, for any $P\in \sheaf Q_X$. Setting $P=\delta_x$, one obtains $f[X](\delta_x) = 0$.
\end{proof}

\begin{ex}\label{ex:computation_divergent_S_binary}
We compute  $H^1(\Sinf, \sheaf F_\alpha( \FProb))$, taking  $\Sinf$ equal to $\bot \to \Us$, and  ${\sheaf E}(\bot)=\{a,b\}$.  Proposition  \ref{lemma_certitude} implies that $f[\bot](1,0) = f[\bot](0,1) = 0$, as a consequence of $f[\bot] = f[\bot] + \bot. f[\bot]$.  All the other relations derived from the cocycle condition \eqref{symmetry_cocycle} become tautological. Therefore, $1$-cocycles are in correspondence with measurable functions $f$ on arguments $(p_a,p_b)$ such that $f(1,0)=f(0,1) = 0$. We conclude that $H^1(\Sinf, \sheaf F_\alpha( \sheaf P))$ has infinite dimension. For a more general condition under which $\dim H^1$ diverges, see Proposition \ref{irreducible_element_infty}.
\end{ex}

The functions $S_\alpha[X]$ introduced in \eqref{eq:shannon_ent_def} and \eqref{eq:tsallis-entropy-def} define a $1$-cochain, that is a $1$-cocycle according to the following proposition.

\begin{prop}[Chain rule]\label{conditional_entropy}
Let $(\Sinf, {\sheaf E})$ be a finite information structure, $\sheaf Q$ an adapted probability functor,  $X$ an element of $\Ob{\Sinf}$, and $ Y,Z$  elements of $\Smon_X$. Then, for all $\alpha >0$, the $1$-cochain  $S_\alpha\in \Hom_{\Sring}(\sheaf B_1, \sheaf F_\alpha(\sheaf Q))$ satisfies 
\begin{equation}
S_\alpha[YZ] = S_\alpha[Y] + Y.S_\alpha[Z].
\end{equation}
This means that $S_\alpha$ belongs to $Z^1(\Sinf, \sheaf F_\alpha(\sheaf Q))$.
\end{prop}
\begin{proof}
Let $P$ be a probability in $\sheaf Q_X$. We further simplify the notation, writing $P(y)$ instead of $P(Y=y)=Y_*P(y)$, and $P(z|y)$ in place of $P(Z=z|Y=y)$. We label the points in ${\sheaf E}(YZ)$ by their image under the injection $\iota:{\sheaf E}(YZ)\to {\sheaf E}(Y)\times {\sheaf E}(Z)$, writing  $w(y,z) \in {\sheaf E}(YZ) $.  
\begin{enumerate}
\item Case $\alpha = 1$: by definition
$$-S_1[YZ](P) = \sum_{w(y,z)\in {\sheaf E}(YZ)} P(y,z) \log  P(y,z)$$
and in fact we can extend this to a sum over the whole set ${\sheaf E}(X)\times {\sheaf E}(Y)$, setting $P(y,z)=0$ whenever $(y,z)\notin \im \iota$ (following the convention $0\log 0 = 0$). We rewrite the previous expression using the conditional probabilities
\begin{align*}
-S_1[YZ](P) 
& =  \sum_{\substack{y\in {\sheaf E}_Y\\P(y)>0}}  \sum_{z\in {\sheaf E}_Z}P(z|y) P(y) ( \log P(y) + \log  P(z|y) )\\
& = \sum_{\substack{y\in {\sheaf E}_Y\\P(y)>0}}  P(y) \log P(y) \sum_{z\in {\sheaf E}_Z} P(z|y)  + \sum_{\substack{y\in {\sheaf E}_Y\\P(y)>0}} P(y) \sum_{z\in {\sheaf E}_Z} P(z|y)\log P(z|y).
\end{align*}
This gives the result, because 
 $\sum_{z\in {\sheaf E}_Z} P(z|y)=1$, and $\sum_{z\in {\sheaf E}_Z} P(z|y)\log P(z|y) = S_1[Z](P|_{Y=y})$. Cf. \cite{Khinchin1957}.
\item Case $\alpha \neq 1$: The result is a consequence of $\delta^2=0$, but can be proved by a direct computation.
\begin{align*}
(1-\alpha)(S[Y]+Y.S[Z]) &= \left(\sum_{y\in {\sheaf E}_Y} P(y)^\alpha -1\right) + \sum_{\substack{y\in {\sheaf E}_Y\\P(y)>0}} P(y)^\alpha  \left( \sum_{z\in {\sheaf E}_Z} P(z|y)^\alpha - 1 \right) \\
&= \sum_{\substack{y\in {\sheaf E}_Y\\P(y)>0}} \sum_{z\in Z} P(z|y)^\alpha P(y)^\alpha -1 \\
&= (1-\alpha)S[XY].
\end{align*}
 The last equality comes from $P(z|y)P(y)=P(z,y)$, and the fact that we can restrict the sum to ${\sheaf E}(YZ)$, neglecting  terms that vanish. 
\end{enumerate}
\end{proof}

We shall see that any nontrivial $1$-cocycle of type $\alpha$ is locally a multiple of $S_\alpha$.  Proposition \ref{functional equations} presents the solution to a functional equation that comes from the cocycle condition. Then, Proposition \ref{H_for_nondegenerate_products} determines the local form of a cocycle. Finally, Theorem \ref{H1-non-degenerate} determine $H^1$ under appropriate nondegeneracy hypotheses on the information structure $\Sinf$ and the probability functor $\sheaf Q$.

For convenience, we introduce the functions
\begin{align}
s_1(p) &:= -p\log p - (1-p)\log(1-p);\\
s_\alpha(p) & := \frac{1}{1-\alpha} (p^\alpha + (1-p)^\alpha - 1) \qquad (\text{for }\alpha\neq 1), 
\end{align}
both defined for $p\in[0,1]$.

\begin{theorem}[Generalized FEITH]\label{functional equations}
Let $f_1,f_2:\Delta^2 \to \Rr$ be two unknown measurable functions satisfying
\begin{enumerate}
\item $f_i(0,1) = f_i(1,0) = 0$ for $i=1,2$.
\item for all $(p_0,p_1,p_2)\in \Delta^2$ such that $p_1>0$ and $p_2>0$,
\begin{align}\label{cocycle-condition-omega3}
(1-p_2)^\alpha f_1\left(\frac{p_0}{1-p_2}, \frac{p_1}{1-p_2}\right) &- f_1(1-p_1,p_1) \\
&=(1-p_1)^\alpha f_2\left(\frac{p_0}{1-p_1}, \frac{p_2}{1-p_1}\right) - f_2(1-p_2,p_2). \nonumber
\end{align}
\end{enumerate}
Then, $f_1 = f_2$ and there exists $\lambda\in \Rr$ such that $f_1(p) = \lambda s_\alpha(p)$.
\end{theorem}
\begin{proof}
The restriction to $p_0 = 0$ (with $p_1 = x, p_2=1-x$) implies $f_2(x,1-x) = f_1(1-x,x)$. This can be used to rewrite  \eqref{cocycle-condition-omega3} in terms of $u(x):=f_1(x,1-x)$. Setting $p_1 = x$, $p_2 = y$ and $p_0 = 1-x-y$, we obtain the functional equation
\begin{equation}\label{functional_entropy}
u(1-x)+(1-x)^\alpha u\left(\frac{y}{1-x}\right) = u(y) + (1-y)^\alpha u\left(\frac{1-x-y}{1-y}\right).
\end{equation}
This functional equation is related to the so-called ``fundamental equation of information theory'' (FEITH), which first appeared in  \cite{Tverberg1958}. Every measurable solution of \eqref{functional_entropy} with $\alpha = 1$ has the form $u(x) = \lambda s_1(x)$, with $\lambda\in \Rr$ \cite{Kannappan1973}.  Analogously, \cite{BennequinVigneaux2019} shows that the general solution in the case $\alpha \neq 1$ is $u(x) = \lambda s_\alpha(x)$, with $\lambda\in \Rr$; this is directly connected to a generalization of the fundamental equation introduced in \cite{Daroczy1970}.
\end{proof}

\begin{ex} Let $\Sinf$ be the poset represented by
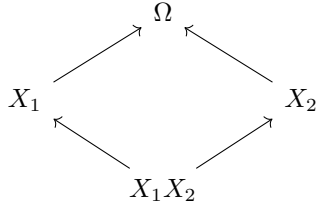
\begin{equation}\label{basic_example_S}
\begin{tikzcd}
& \Us & \\
X_1 \ar[ur] & & X_2 \ar[ul] \\
& X_1X_2 \ar[ul] \ar[ur]& 
\end{tikzcd}
\end{equation}
and ${\sheaf E}$ be the functor defined at the level of objects by ${\sheaf E}(X_1) =\{{\{1\}}, {\{0,2\}}\}$, ${\sheaf E}(X_2)=\{{\{2\}}, {\{0,1\}}\}$, and ${\sheaf E}(X_1X_2) =\{{\{0\}},{\{1\}},{\{2\}}\}$; for each arrow $\pi:X\to Y$, the map $\pi_*:{\sheaf E}(X)\to {\sheaf E}(Y)$ sends $I\to J$ iff $I\subset J$. The pair $(\Sinf, {\sheaf E})$ is an information structure (it comes from a concrete one). Consider $f\in Z^1(\sheaf F_\alpha(\FProb))$: the $1$-cocyle condition means that, as functions on $\FProb(X_1X_2)$,
\begin{equation}\label{eq:1-cocycle_little_ex}
f[X_1X_2] = X_1.f[X_2]+f[X_1] \quad \text{and} \quad  
f[X_1X_2] = X_2.f[X_1]+f[X_2]. 
\end{equation}
We write $f_1$, $f_2$ and $f_{12}$ instead of $f[X_1] $, $f[X_2]$ and $f[X_1X_2]$, respectively. The determination of $f_1$ and $f_2$ such that $X_1.f_2+f_1 = X_2.f_1+f_2$ fix $f_{12}$ completely. In terms of a probability $(p_0,p_1,p_2)$ in $\FProb(X_1X_2)$ this equation is exactly \eqref{cocycle-condition-omega3}. We conclude that every cocycle is a multiple of the corresponding  $\alpha$-entropy: there exists a unique constant $\lambda\in \Rr$ such that $f[Z](P) = \lambda S_\alpha[Z](P)$, for every observable $Z\in \Ob\cat S$ and every probability law $P\in \FProb(X_1X_2)$. This establishes  that $Z^1(\Sinf, \sheaf F_\alpha( \FProb)) \cong \Rr$. Hence $H^1(\Sinf, \sheaf F_1( \FProb))\cong \Rr$, and $H^1(\Sinf, \sheaf F_\alpha( \FProb))\cong\langle 0 \rangle$. The hypotheses are minimal: on the one hand, if we remove $X_1$ or $X_2$, Proposition \ref{irreducible_element_infty} shows that $\dim H^1 = \infty$; on the other, if $ \sheaf Q_{X_1X_2}$ does not contain the interior of $\Delta^2$, the equations \eqref{eq:1-cocycle_little_ex} translate into a ``degenerate'' system of functional equations
\begin{gather*}
f_{12}(p_0,1-p_0,0) = f_1(1-p_0, p_0), \\
f_{12}(p_0,0,1-p_0) = f_2(1-p_0, p_0), \\
f_{12}(0,p_1,1-p_1) = f_1(p_1, 1-p_1) = f_2(1-p_1,p_1),
\end{gather*}
thus one function remains arbitrary. 
\end{ex}

We want to generalize the previous computations to any product $XY$. It is already clear that ${\sheaf E}_{XY}={\sheaf E}_X\times {\sheaf E}_Y$ is not necessary; at the same time, if one is ``too far'' from ${\sheaf E}_{XY}={\sheaf E}_X\times {\sheaf E}_Y$ and $\sheaf Q = \FProb$, then the system of functional equations implied by the $1$-cocycle condition degenerates. A precise sufficient condition is introduced in the following definition of \emph{nondegenerate product}, which is better understood looking at the proof of Proposition \ref{H_for_nondegenerate_products}. To determine the function $f[XY]$, for given observables $X$ and $Y$, one first obtains the recursive formulas \eqref{recurrence_X} and \eqref{recurrence_Y} for the functions $f[X]$ and $f[Y]$: given a total order of the sets ${\sheaf E}_X$ and ${\sheaf E}_Y$, the different steps of the recursion are coded by a path in $\Zz^2$. Both formulas are a simplification of the symmetric equation \eqref{eq_XY} for particular laws $\tilde P$ given by Definition \ref{def:non-degenerate-probabilistic}-\ref{probabilistic-ng:lifting} that make one of the terms trivial. The recursive formulas involve a term where $f[X]$ and $f[Y]$ have only two nonzero arguments, and both are related  by the FEITH-like functional equation  \eqref{cocycle-condition-omega3} in Proposition \ref{functional equations};  that this equation holds for any $(p_0,p_1,p_2)\in \Delta^2$ such that $p_1>0$ and $p_2>0$ is ensured by Definition \ref{def:non-degenerate-probabilistic}-\ref{probabilistic-ng:2-dim-cell-nondegenerate}.

\begin{defi}\label{def:non-degenerate-probabilistic}
Let $X$ and $Y$ be two objects of $\Sinf$, such that $|{\sheaf E}_X|=k$ and $|{\sheaf E}_Y|=l$. Let $\iota$ be the inclusion ${\sheaf E}_{XY}\hookrightarrow {\sheaf E}_X\times {\sheaf E}_Y$. We call the product $XY$ \emph{nondegenerate} if $k,l\geq 2$ and there exist enumerations  $\{x_1,...,x_k\}$ of ${\sheaf E}_X$ and  $\{y_1,...,y_l\}$ of ${\sheaf E}_Y$, together with a North-East (NE) lattice path\footnote{A North-East (NE) lattice path  on $\Zz^2$ is a sequence of points  $(\gamma_i)_{i=1}^m\subset \Zz^2$ such that $ \gamma_{i+1} -\gamma_i\in \{(1,0),(0,1)\}$ for every $i\in \{1,...,m-1\}.$}   $(\gamma_i)_{i=1}^m$  on $\Zz^2$  going from $(1,1)$ to $(k,l)$, such that:
\begin{enumerate}  
\item\label{probabilistic-ng:lifting} If $\gamma_i=(a,b)$ and $\gamma_{i+1} -\gamma_i = (1,0)$, then for every  $P\in \sheaf Q_X$ such that $\supp P \subset \set{x_j}{ a \leq j \leq k} $, there exists  $\tilde P \in \sheaf Q_{XY}$ whose support is contained in
$$
\iota^{-1}(\{(x_a,y_{b+1})\}\cup \set{(x_j,y_b)}{a+1 \leq j \leq k})$$ or in $$\iota^{-1}(\{(x_a,y_{b})\}\cup \set{(x_j,y_{b+1})}{a+1 \leq j \leq k})$$ 
and such that $P= X_*\tilde P$. (Remark that, for such values of $XY$, the value of the $X$-projection completely determine the $Y$-projection.)

Analogously, if $\gamma_{i+1} -\gamma_i = (0,1)$, then for every $P\in \sheaf Q_Y$ such that  $\supp P \subset\set{y_j}{b\leq j \leq l}$, there exists a counting function $\tilde P \in \sheaf Q_{XY}$ whose support is contained in
$$
\iota^{-1}(\{(x_{a+1},y_{b})\}\cup \set{(x_a,y_j)}{b+1 \leq j \leq l})$$ or in $$\iota^{-1}(\{(x_a,y_{b})\}\cup \set{(x_{a+1},y_j)}{b+1 \leq j \leq k})$$
and such that $P = Y_*\tilde P$. 

\item\label{probabilistic-ng:2-dim-cell-nondegenerate} For each $\gamma_i=(a,b)$ such that $a<k$ and $b<l$, there are elements $z_1,z_2,z_3$ in $$\iota^{-1}\set{(x_m,y_n)}{ a\leq m \leq {a+1} \text{ and } b \leq n \leq {b+1}} $$
such that $[\delta_{z_1},\delta_{z_2},\delta_{z_3}]$ (the convex hull of the corresponding Dirac measures) is a subset of $\sheaf Q_{XY}$ (which in turn is a subset of the simplex $\sheaf P(X)$ with vertices $\{\delta_{z}\}_{z\in \sheaf E_{XY}}$). 
\end{enumerate}
 \end{defi}
 
 In particular, the first condition implies (when $a=1$ or $b=1$) that $\sheaf Q_{XY} \to \sheaf Q_X$ is surjective, and similarly for $Y$. The product of a observable with itself is always degenerate.

\begin{prop}\label{H_for_nondegenerate_products}
Let $(\Sinf,{\sheaf E})$ be a finite information structure, $\sheaf Q$ an adapted probability functor, and $X$, $Y$ two different observables in $\Ob{\Sinf}$ such that $XY\in \Ob{\Sinf}$. Let $f$ be a $1$-cocycle of type $\alpha$, i.e. an element of $Z^1(\Sinf, \sheaf F_\alpha(\sheaf Q))$. If $XY$ is nondegenerate, there exists $\lambda\in \Rr$ such that
$$f[X]=\lambda S_\alpha[X], \quad f[Y]=\lambda S_\alpha[Y], \quad f[XY]=\lambda S_\alpha[XY].$$
\end{prop}
\begin{proof}
As $f$ is a $1$-cocycle, it satisfies the two equations derived from \eqref{eq:n-coboundary} 
\begin{align}
Y.f[X] & =  f[XY] - f[Y] \label{proof_rec_1} \\
X.f[Y] & = f[XY] - f[X] \label{proof_rec_2}
\end{align}
and therefore the symmetric equation
\begin{equation}\label{eq_XY}
X.f[Y]-f[Y]=Y.f[X]-f[X].
\end{equation}
For a law $P$, we write  $$\left(\begin{array}{cccc}
s & t & u & \ldots \\
p & q & r & \ldots
\end{array}\right)$$ if $P(s) = p$, $P(t) = q$, $P(u) = r$, etc. and the probabilities of the unwritten parts are zero.

Fix enumerations $(x_1,...,x_k)$ and $(y_1,...,y_l)$ that satisfy the definition of nondegenerate product, and let $\{\gamma_{i}\}_{i=1}^m$ be the corresponding NE path. Write $\gamma_i=(a,b)$. If $\gamma_{i+1}-\gamma_i = (1,0)$, we shall show that the following recursive formula holds:

\begin{multline}\label{recurrence_X}
f[X]\left(\begin{array}{ccc}
x_{a} & \ldots & x_k \\
\mu_{a} & \ldots & \mu_k
\end{array}\right)=  (1-\mu_a)^\alpha f[X]\left(\begin{array}{ccc}
x_{a+1} & \ldots & x_k \\
\mu_{a+1}/(1-\mu_a)& \ldots & \mu_k/(1-\mu_a)
\end{array}\right) \\+ f[X]\left(\begin{array}{cc}
x_{a} & x_{a+1} \\
\mu_{a} & 1-\mu_a
\end{array}\right).
\end{multline}
Analogously, if $\gamma_{i+1}-\gamma_i = (0,1)$,
\begin{multline}\label{recurrence_Y}
f[Y]\left(\begin{array}{ccc}
y_{b} & \ldots & y_l \\
\nu_b & \ldots & \nu_l
\end{array}\right)= (1-\nu_b)^\alpha f[Y]\left(\begin{array}{ccc}
y_{b+1} & \ldots & y_l \\
\nu_{b+1}/(1-\nu_c)& \ldots & \nu_{l}/(1-\nu_b)
\end{array}\right) \\+ f[Y]\left(\begin{array}{cc}
y_b & y_{b+1} \\
\nu_b & 1-\nu_b
\end{array}\right).
\end{multline}

Suppose that $\gamma_{i+1}-\gamma_i = (1,0)$.  Let $$p=\left(\begin{array}{cccc}
x_a & \ldots & x_k  \\
\mu_a & \ldots & \mu_k 
\end{array}\right)$$ be a probability in $\sheaf Q_X$ supported on $\set{x_i}{a\leq i \leq k}$. We know it has a preimage $\tilde p$ under marginalization $X_*$ as in Definition \ref{def:non-degenerate-probabilistic}-\eqref{probabilistic-ng:lifting}. For such law,   knowledge of $X$ implies knowledge of $Y$ with certainty, therefore $X.f[Y](\tilde p)=0$; by equation \eqref{proof_rec_2}, $f[XY](\tilde p) = f[X](X_*\tilde p)=f[X](p)$. Equation \eqref{proof_rec_1} then reads 
\begin{multline}\label{auxiliary1_proof_recurrence_prob}
(1-\mu_a)^\alpha f[X]\left(\begin{array}{ccc}
x_{a+1} & \ldots & x_k \\
\mu_{a+1}/(1-\mu_a)& \ldots & \mu_k/(1-\mu_a)
\end{array}\right) =\\ 
f[X]\left(\begin{array}{ccc}
x_{a} & \ldots & x_k \\
\mu_{a} & \ldots & \mu_k
\end{array}\right)
-f[Y]\circ \tau\left(\begin{array}{cc}
y_b & y_{b+1} \\
1-\mu_a &  \mu_a
\end{array}\right),
\end{multline}
where $\tau$ is the identity or the transposition of the arguments. In any case, we can set  $\mu_{a+1} = 1-\mu_a$ and $\mu_{a+2}=\ldots = \mu_k = 0$ (because this probability appears as the image of some of the probabilities on $XY$ given by Definition \ref{def:non-degenerate-probabilistic}-\eqref{probabilistic-ng:2-dim-cell-nondegenerate}), to conclude that
\begin{equation}
f[X]\left(\begin{array}{cc}
x_{a} & x_{a+1} \\
\mu_{a} & 1-\mu_a
\end{array}\right)=f[Y]\circ\tau\left(\begin{array}{cc}
y_b & y_{b+1} \\
1-\mu_a &  \mu_a
\end{array}\right),
\end{equation}
which combined with \eqref{auxiliary1_proof_recurrence_prob} implies \eqref{recurrence_X}. The identity \eqref{recurrence_Y} can be obtained analogously.

We proceed to the determination of 
$$\phi_a(z):= f[X]\left(\begin{array}{cc}
x_{a} & x_{a+1} \\
z & 1-z
\end{array}\right) \quad \text{and} \quad \psi_b(z):= f[Y]\left(\begin{array}{cc}
y_b & y_{b+1} \\
z & 1-z
\end{array}\right), \quad \text{for } z\in[0,1].$$
Let $z_1,z_2,z_3$ be the three elements of ${\sheaf E}_{XY}\subset {\sheaf E}_X\times {\sheaf E}_Y$ such that $[\delta_{z_1},\delta_{z_2},\delta_{z_3}]$ is the subset of $\sheaf Q_{XY}$ given by Definition \ref{def:non-degenerate-probabilistic}-\eqref{probabilistic-ng:2-dim-cell-nondegenerate} when $\gamma_i = (a,b)$. For any $\mu=(\mu_0,\mu_X,\mu_Y)\in \Delta^2$, set $P(z_i):=\mu_Y$, where $z_i$ is the component that differs from the others w.r.t. the $Y$-projection, in such a way that $Y_*\mu$ is $(\mu_Y,1-\mu_Y)$. Similarly, set $P(z_j):=\mu_X$, where $z_j$ is the component that differs from the others w.r.t. the $X$-projection. With this assignment,  \eqref{eq_XY} reads
\begin{multline}
(1-\mu_X)^\alpha f[Y]\circ \sigma\left(\begin{array}{cc}
y_b & y_{b+1} \\
\mu_0/(1-\mu_X) & \mu_Y/(1-\mu_X)
\end{array}\right) - f[Y]\circ \sigma\left(\begin{array}{cc}
y_b & y_{b+1} \\
1-\mu_Y & \mu_Y
\end{array}\right) \\= (1-\mu_Y)^\alpha f[X]\circ\tau\left(\begin{array}{cc}
x_a & x_{a+1} \\
\mu_0/(1-\mu_Y) & \mu_Y/(1-\mu_Y)
\end{array}\right)-f[X]\circ \tau \left(\begin{array}{cc}
x_a & x_{a+1} \\
1-\mu_X & \mu_X
\end{array}\right),
\end{multline}
\normalsize
where $\sigma$, $\tau$ are the identity or the transposition of both nontrivial arguments. In any case, this leads to the  functional equation in Proposition \ref{functional equations}, hence $\phi_a(z) = \psi_b(z) = \lambda s_\alpha(z)$ for certain $\lambda\in \Rr$ (the solution is symmetric in the arguments).

When considering $\gamma_{i+1}$, one finds the functions $\phi_a$ and $\psi_{b+1}$, or the functions $\phi_{a+1}$ and $\psi_b$, since the difference $\gamma_{i+1}-\gamma_i$ is either   $(0,1)$ or $(1,0)$. This ensures that the constant $\lambda$ that appears for each $\gamma_i$ is always the same.

Repeat the process above with every $\gamma_i$ ($1\leq i \leq m$). The system of equations \eqref{recurrence_X} obtained in this way, together with the functions already determined $$f[X]\left(\begin{array}{cc}
x_{a} & x_{a+1} \\
z & 1-z
\end{array}\right) = \lambda s_{\alpha}(z,1-z)\quad \text{ for }1\leq a\leq k-1,$$ entails that
\begin{equation}
f[X](\mu_1,\ldots, \mu_k)= \lambda  \sum_{i=0}^{k-1} \left(1-\sum_{j=1}^i \mu_j\right)^\alpha s_\alpha \left(\frac{\mu_{i+1}}{ \left(1-\sum_{j=1}^i \mu_j\right)}\right).
\end{equation}
Set $T_i := 1- \sum_{j=1}^i \mu_j$. An elementary computation shows that, when $\alpha = 1$,
\begin{equation}
\sum_{i=0}^{k-1}  \left(1-T_i\right)  s_1\left(\frac{\mu_{i+1}}{ \left(1-T_i\right)}\right) = \sum_{i=1}^{k} \mu_i \log \mu_i,
\end{equation}
\normalsize
and when $\alpha \neq 1$,
\begin{equation}
\sum_{i=0}^{k-1} \left(1-T_i\right)^\alpha  s_\alpha \left(\frac{\mu_{i+1}}{ \left(1-T_i\right)}\right) 
=\sum_{i=1}^k \mu_i^\alpha - 1.
\end{equation}
Therefore, for any $\alpha > 0$, we have $f[X] = \lambda S_\alpha[X]$. Analogously, $f[Y]=\lambda S_\alpha[Y]$.
\end{proof}

\subsection{Determination of $H^1$}\label{sec:determination_H1}

In this section, we prove Theorem \ref{H1-non-degenerate}, that determines completely $H^1$ in the probabilistic case under suitable hypotheses. We also discuss some pathological cases.

We call a observable $Z$ \keyt{reducible} if there exist $X,Y\in \Ob\cat S \sm \{1,Z\}$ such that $Z=XY$, and \keyt{irreducible} otherwise.  It is \keyt{nontrivialy reducible} if the product $XY$ is nondegenerate.  


\begin{proof}[of Theorem \ref{H1-non-degenerate}]
Let $f$ be an element of $Z^1(\Sinf,\sheaf F_\alpha(\sheaf Q))$. We are going to determine $f_X[Y]$ for every $X\in \Ob \Sinf$ and $Y\in \Smon_X$. 

If $X$ satisfies $|\sheaf E_X|=1$, then $|\sheaf E_Y|=1$ for any $Y\in \Smon_X$. We conclude that $f_X[Y] = 0$ in virtue of Proposition \ref{lemma_certitude}. 

If $X$ satisfies $|\sheaf E_X| > 1$, then it belongs to $\Sinf^*$. Let $\cat C(X)$ denote the connected component of $\Sinf^*$ that contains $X$. By hypothesis, there exists a nontrivially reducible object $Z$ such that $Z\to X\to Y$. Then,
\begin{equation}\label{eq:fY_from_fZ}
f[Y](Z_*P)= f_Z[Y](P) = f_Z[Z](P) - Y.f_Z[Z](P) = \begin{cases} 0 & \text{if } |\sheaf E_Y|=1 \\ 
\lambda_Z S_\alpha[Y](Z_*P) &  \text{if } |\sheaf E_Y|>1 
\end{cases},
\end{equation}
because $f_Z[Z]$ equals $\lambda_Z S_\alpha[Z]$, for some $\lambda_Z\in \Rr$, according to Proposition \ref{H_for_nondegenerate_products}. Since $Z_*$ is supposed to be surjective, this fully determines $f[Y]$. By naturality, $f_X[Y]$ is the map $P\mapsto f[Y](X_*P)$. Remark that $Y$ also belongs to $\cat C(X)$; if $Z'$ is another nontrivially reducible observable such that $Z'\to Y$ then $Z\in \cat C(X)$ too. Repeating the previous argument for $Z'$, we conclude that $f[Y]=\lambda_Z S_\alpha[Y] = \lambda_{Z'}S_\alpha[Y]$, therefore $\lambda_Z = \lambda_{Z'} =: \lambda_{\cat C(X)}$. 

More generally, if $Z$, $Z'$ are any two nontrivially reducible objects in the same component $\cat C$ of $\Sinf^*$, then there is a zig-zag diagram in $\Sinf^*$ of the form 
$$Z \rightarrow Y_1 \leftarrow X_1 \rightarrow Y_2 \leftarrow X_2 \rightarrow\cdots \leftarrow X_k \rightarrow Y_{k+1} \leftarrow Z'$$
for certain $k\in \Nn$, where $Y_i, X_i\in \Ob \Sinf^*$. By hypothesis, each $X_i$ is refined by a nontrivially reducible object $Z_i$, in such a way that we also have a diagram 
$$Z \rightarrow Y_1 \leftarrow Z_1 \rightarrow Y_2 \leftarrow Z_2 \rightarrow\cdots \leftarrow Z_k \rightarrow Y_{k+1} \leftarrow Z'$$  
in $\Sinf^*$.   The repeated application of the argument in the previous paragraph implies that $\lambda_Z = \lambda_{Z_1} = \cdots = \lambda_{Z'} =: \lambda_{\cat C}$. 

Summarizing, $f$ is such that $f_X[Y]=0$ if $|\sheaf E_Y|=1$ and $f_X[Y](P) = \lambda_{\cat C} S_\alpha[Y](X_*P)$ if $Y$ belongs to a connected component $\cat C$ of $\Sinf^*$ (hence $X$ too). Therefore, there is a linear bijective correspondence between cocycles $f$ and sequences of $(\lambda_{\cat C})_{\cat C} \in \Rr^{\pi_0(\Sinf^*)}$. (The set $\Hom_{\Sring}(\sheaf B_1, \sheaf F_\alpha)$ has a vector space structure given by ``object-wise'' addition: for every $\phi, \psi\in \Hom_{\Sring}(\sheaf B_1, \sheaf F_\alpha)$, $\alpha,\beta\in \Rr$, $X\in \Ob \Sinf$, and $Y\in \Smon_X$, one has  $(\alpha\phi+\beta\psi)_X[Y] =\alpha \phi_X[Y] + \beta \psi_X[Y]$.)

We saw in Section \ref{sec:entropy} that $\delta C^0(\Sinf, \sheaf F_1(\sheaf Q)) \cong \langle 0\rangle$ and $\delta C^0(\Sinf,\sheaf F_\alpha(\sheaf Q)) =\set{\lambda  S_\alpha}{\lambda\in \Rr}$. 
\end{proof}
%

As a byproduct of the previous  proof, we also obtain the following proposition.
\begin{prop}
Let $\{(\Sinf_i, {\sheaf E}_i, \sheaf Q_i)\}_{i\in I}$ be a collection of triples that satisfy separately the hypotheses stated in Theorem \ref{H1-non-degenerate}. Let $\bigsqcup_{i\in I}  \sheaf Q$ denote the adapted functor of probabilities on $\bigsqcup_{i\in I} \Sinf$ that coincides with $\sheaf Q_i$ on $\Sinf_i$. Then,
$$ Z^1\left(\bigsqcup_{i\in I} \Sinf,\sheaf F(\bigsqcup_{i\in I}  \sheaf Q)\right) \cong \prod_{i\in I} Z^1(\Sinf_i,\sheaf F(\sheaf Q_i)). $$
\end{prop}
\begin{proof}
The category $(\bigsqcup_{i=1}^n \Sinf)^*$ is the disjoint union of the categories $\Sinf_i^*$, for $i\in I$. 
\end{proof}

\subsubsection{Degenerate cases} 
Let us focus now on the particular case of bounded, finite information structures $(\Sinf,\sheaf E)$, with $\sheaf E$ conservative (e.g. obtained from a concrete structure) and $\sheaf Q= \sheaf P$. 
The hypothesis of Theorem \ref{H1-non-degenerate} is satisfied if and only if every minimal object of the poset $\Sinf$ is nontrivially reducible. Hence there are two kind of cases not covered by it: 
\begin{enumerate}
\item there is an irreducible minimal object;
\item all minimal objects are reducible, but some of them cannot be written as nondegenerate products.
\end{enumerate}
In the latter, all kinds of behaviors are possible, as the examples at the end this section show. 
 
In the Example \ref{ex:computation_divergent_S_binary}, we proved that $H^1(\Sinf, \sheaf F_\alpha(\FProb))$ has infinite dimension when $\Sinf \cong \ObsFin(\{0,1\})$; in this case, there is only one nontrivial observable and it is obviously irreducible. Now we proceed to the generalization of this result. 


\begin{prop}\label{irreducible_element_infty}
Let $(\Sinf,{\sheaf E})$ be a bounded, finite information structure such that ${\sheaf E}$ is conservative. Suppose $\Sinf$ has an irreducible minimal object. Then, $\dim H^1(\Sinf,\sheaf F_\alpha(\FProb)) = \infty$.
\end{prop} 
\begin{proof}
Let $M$ be an irreducible minimal object. 
Remark that $$T:=\set{X}{M\to X \text{ and } X\neq M}$$ has a unique minimal element $\tilde X$ for the partial order $\Sinf$ (if $X_1$ and $X_2$ were two different minimal elements, the diagram $X_1 \leftarrow M \rightarrow X_2$  would imply that $M = X_1 X_2$).

Recall that a $1$-cochain $f$ is uniquely determined by a collection of measurable  functions $\{f[X]:\sheaf Q_X\to \Rr\}_{X\in \Ob \Sinf}$. Let us set  $f[X] = 0$ for every $X\in \Ob \Sinf\sm\{M\}$ and show that $f[M]$ can be chosen arbitrarily so that $f$ is a $1$-cocycle.  

The $1$-cocycle condition  implies in particular that 
\begin{equation}\label{eq:conditioned_minimal_object}
f[M](P) = \tilde X.f[M](P):=\sum_{x\in {\sheaf E}_{\tilde X}} \tilde X_*P(x)^\alpha f[M](P|_{\tilde X=x}).
\end{equation}
 Given this functional equation, the others  become redundant, since for any $X$ coarser than $M$ (hence coarser than $\tilde X$),
$$X.f[M]=X.(\tilde X.f[M]) = (X\tilde X).f[M] = \tilde X.f[M].$$

Since ${\sheaf E}$ is conservative, the induced surjection ${\sheaf E}\pi:{\sheaf E}_{M}\to {\sheaf E}_{\tilde X}$ is not a bijection. Therefore, the set ${\sheaf E}^*:= \set{x\in {\sheaf E}_{\tilde X}}{|({\sheaf E}\pi_{\tilde XM})^{-1}(x)|\geq 2}$ is nonempty. For each element of $x\in {\sheaf E}^*$, one can introduce an arbitrary function $g_x$ on $\Delta^{|({\sheaf E}\pi_{\tilde XM})^{-1}(x)|-1}$ that vanishes on the vertices, so that $f[M](P|_{X=x}) = g_x(P|_{X=x})$ and the value of $M$ on an arbitrary $P\in \FProb_X$ is given by \eqref{eq:conditioned_minimal_object}.
\end{proof}

To close this section, we make some remarks about the case of reducible minimal objects that cannot be written as nondegenerate products. If the product is degenerate, multiple constants can appear or the dimension of $H^1(\Sinf, \sheaf F_\alpha(\FProb))$ can become infinite, as the following examples show.

\begin{ex}
Consider the information structure $(\Sinf,{\sheaf E})$ given by the poset $\Sinf$ represented by
$$
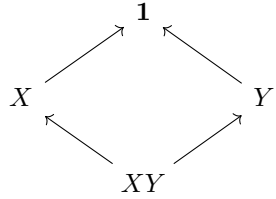
\begin{tikzcd}
 & \Us & \\
 X \ar[ur] & & Y   \ar[ul] \\
 & XY \ar[ul]\ar[ur] &
\end{tikzcd}
$$
and the assignment ${\sheaf E}(X)=\{x_1,x_2,x_3,x_4\}$, ${\sheaf E}(Y)=\{y_1,y_2,y_3,y_4\}$, and $${\sheaf E}(XY)=(\{x_1,x_2\}\times \{y_1,y_2\})\cup (\{x_3,x_4\}\times \{y_3,y_4\});$$ the surjections are the terminal maps and the restrictions of the canonical projectors.  To determine a $1$-cochain $f\in C^1(\Sinf, \sheaf F(\FProb))$ it suffices to specify 
\begin{equation*}
f[X]\left(\begin{array}{cc}
x_1 & x_2 \\
p & 1-p
\end{array}\right),  f[X]\left(\begin{array}{cc}
x_3 & x_4 \\
p & 1-p
\end{array}\right),  f[Y]\left(\begin{array}{cc}
y_1 & y_2 \\
p & 1-p
\end{array}\right) \text{ and } f[Y]\left(\begin{array}{cc}
y_3 & y_4 \\
p & 1-p
\end{array}\right),
\end{equation*}
for arbitrary $p\in [0,1]$.  Proposition \ref{functional equations} allows us to conclude that 
\begin{equation}
f[X]\left(\begin{array}{cc}
x_1 & x_2 \\
p & 1-p
\end{array}\right)=\lambda_1 s_\alpha(p), \quad
f[Y]\left(\begin{array}{cc}
y_1 & y_2 \\
p & 1-p
\end{array}\right)=\lambda_1 s_\alpha(p).
\end{equation}
We can use this proposition a second time to show that 
\begin{equation}
f[X]\left(\begin{array}{cc}
x_3 & x_4 \\
p & 1-p
\end{array}\right)=\lambda_2 s_\alpha(p), \quad
f[Y]\left(\begin{array}{cc}
y_3 & y_4 \\
p & 1-p
\end{array}\right)=\lambda_2 s_\alpha(p).
\end{equation}
However, it is impossible to find a relation between $\lambda_1$ and $\lambda_2$  using \eqref{proof_rec_1} and \eqref{proof_rec_2}. So $Z^1( \sheaf F_\alpha(\sheaf Q))\cong \Rr^2$.
\end{ex}

\begin{ex}
If in the previous example
\begin{equation}
{\sheaf E}_{XY}=(\{x_1,x_2\}\times \{y_1,y_2\})\cup\{(x_3,y_3),(x_4,y_4)\},
\end{equation}
it is still true that
\begin{equation}
f[X]\left(\begin{array}{cc}
x_1 & x_2 \\
p & 1-p
\end{array}\right)=\lambda_1 s_\alpha(p), \quad
f[Y]\left(\begin{array}{cc}
y_1 & y_2 \\
p & 1-p
\end{array}\right)=\lambda_1 s_\alpha(p),
\end{equation}
whereas \eqref{proof_rec_1} and \eqref{proof_rec_2} imply that
\begin{equation}
f[XY] \left(\begin{array}{cc}
(x_3,y_3) & (x_4,y_4) \\
p & 1-p
\end{array}\right) = f[X]\left(\begin{array}{cc}
x_3 & y_3 \\
p & 1-p
\end{array}\right) = f[Y]\left(\begin{array}{cc}
y_3 & y_4 \\
p & 1-p
\end{array}\right).
\end{equation}
We conclude that any measurable  $g:\Delta^1\to \Rr$ such that $g(0,1)=g(1,0)=0$ defines a solution, therefore  $\dim Z^1(\sheaf F_\alpha(\sheaf Q)) = \infty$.
\end{ex}

\subsection{Interpretation: Crossed homomorphisms}
 The extension of the results in \cite[Ch.~X]{MacLane1994} concerning the Hochschild cohomology of algebras to the case of presheaves of algebras is straightforward, taking into account that the relative bar construction in the appendix gives the explicit description of both.  
 
 For instance, $H^0(\Sinf,\sheaf N)$ corresponds to global sections of $\sheaf N$ that are invariant under the action of $\sheaf A$. When $\sheaf N=\sheaf F_\alpha(\sheaf Q)$ the only sections are constant functionals; they are all invariant when $\alpha = 1$, and none of them is invariant when $\alpha\neq 1$.  
 
 A \emph{crossed homomorphism} of $\Sring$ to $\sheaf N$ is a morphism of presheaves of $\Rr$-modules $f:\Sring \to \sheaf N$ satisfying the identity 
 \begin{equation}
 f[X_1X_2] = X_1 f[X_2] + f[X_1].
 \end{equation}
 The principal crossed homomorphisms have the form $f_n[X] = Xn - nX = Xn - n$ for some section $n$ of $\sheaf N$. Therefore $H^1(\Sinf, \sheaf N)$ is the $\Rr$-module of crossed homomorphisms modulo the principal ones. 
 
 When $\sheaf N=\sheaf F_\alpha(\sheaf Q)$ and the hypotheses in Theorem \ref{H1-non-degenerate} are satisfied, the \emph{only} crossed homomorphism is the corresponding $\alpha$-entropy (up to multiplicative constants). Thus each entropy represents the unique nontrivial way of transforming a multiplicative operation involving partitions or $\sigma$-algebras into an additive operation on probabilistic functionals, introducing an appropriate ``twist''. In the degenerate case of Proposition \ref{irreducible_element_infty}, the algebra $\Sring_M$ is ``too poor'' to determine a unique crossed homomorphism.  Moreover, the naturality of $S_\alpha$ implies that $S_\alpha[X](P)$ depends only on $X_*P$.  This turns out to be the appropriate notion of locality and justifies the introduction of presheaves, and in this sense one could consider information cohomology as a ``localized'' version  of Hochschild cohomology. 
 
 Finally, $H^2(\Sinf,\sheaf N)$ classifies $\Rr$-split singular algebra extensions of $\sheaf N$ by $\Sring$, see \cite[Thm.~3.1]{MacLane1994}.

\section{Final remarks}

We have seen that Shannon entropy can be identified with a cohomology class in information cohomology, a topological invariant associated to a finite statistical system. 
It determines the ``natural'' way to turn the product of partitions/$\sigma$-algebras into an additive operation of probabilistic functionals that is compatible with variations of the probability laws.
 In this framework, each entropy $S_\alpha$ spontaneously appears indexed by the observables---via the bar construction---, and it is ``local'' in the sense that $S_\alpha[X](P)$ only depends on $X_*P$, the marginalization of the probability  $P$ on ${\sheaf E}_X$.  This good notion of locality--- encoded by the naturality of the involved functors---and the 1-cocycle condition replace completely the axioms invoked by the usual algebraic approaches (one does not need to suppose symmetry, convexity/concavity or a given asymptotic behavior). 

Whereas the fact that Shannon or Tsallis entropies are 1-cocycles is related to well-known facts (because the 1-cocycle condition is the corresponding chain rule), the cohomological restatement gives new insights: it is possible to change the coefficients of the cohomology and interpret the same 1-cocycle condition for completely different objects. In \cite{Vigneaux2019-thesis}, it is showed that when the coefficients are certain ``combinatorial'' functionals ($\Rr_{>0}$-valued functions of statistical frequencies), the 1-cocycles are generalized multinomial coefficients and the chain rule corresponds to the multiplicative relations between them. The combinatorial and probabilistic cocycles are  asymptotically related. The same reference introduces information structures of continuous observables with gaussian laws (cf. Example \ref{ex:homogeneous}): in this setting, the dimension of the support appears as a 1-cocycle,  the 1-cocycle condition being the nullity-rank theorem; similarly, the log of the determinant of the covariance matrix is a 1-cocycle, and the 1-cocycle condition is Schur's determinantal formula. It would be difficult to recognize all these identities as manifestations of the same thing without the general framework. The isolation of the common combinatorial structure in all these examples (the conditional meet semilattices) is instrumental to develop a definition of information cohomology that works in general.

It seems to us that the \emph{connections} between the different points of view on entropy---algebraic, probabilistic, combinatorial, dynamical---are not yet systematically understood. We are convinced that the language of (pre)sheaves and their homological invariants opens many new directions of research, being rich enough to integrate constructions coming from these diverse domains. Further applications of our approach are related to several open problems, for instance:
\begin{enumerate}
\item the computation of cocycles of higher degrees, conjectured in \cite{Baudot2015} to be new measures of mutual information of all orders;
\item the possible reformulation of Shannon's coding theorems as cohomological obstruction problems;
\item  the computation of cohomology for categories of symplectic manifolds and reductions, giving adapted measures of information;
\item a functorial relation between classical and quantum information cohomology (and the corresponding concentration theorems) through geometric quantization, and
\item a categorical formulation of \emph{Ruzsa's dictionary} \cite{Ruzsa2009}, that relates inequalities for cardinalities and entropies.
\end{enumerate} 
 Only the first steps are made: we hope that other authors, with new ideas, will be able to go further.

\appendix

\section{Relative bar resolution}\label{sec:relative}

\subsection{General results} In this subsection, we summarize the construction of a relative bar resolution from \cite[Ch.~IX]{MacLane1994}. The purpose is to find the analogue of a free resolution of modules, but in the general context of abelian categories. Capital Latin letters $A,B,C...$ denote objects and Greek letters $\alpha, \beta...$ morphisms.

A \keyt{relative abelian category} is a pair of abelian categories $\cat A$ and $\cat M$ and a covariant functor $\square: \cat A \to \cat M$ which is additive, exact and faithful  (we write $\square (X) = X_\square$, for objects and morphisms). 

\begin{ex}\label{relative_modules}
The simple example to have in mind are $R$-modules and $S$-modules, when $S$ is a subring of $R$ with the same unit (write $\iota:S\to R$ for the injection). Every $R$-module $A$ can be seen as an $S$-module ${}_\iota A$ by \emph{restriction of scalars}; every $R$-module morphism $\alpha :A\to B$ is also a $S$-module morphism ${}_\iota \alpha:{}_\iota A\to {}_\iota B$. The assignment $\square A := {}_\iota A$ and $\square \alpha := {}_\iota \alpha$ defines  a relative abelian category.
\end{ex}

A short exact sequence $\chi\| \sigma$ in $\cat A$ is relatively split ($\square$-split) if $\chi_\square\| \sigma_\square$  splits in $\cat M$. A monomorphism $\chi$ is called allowable if $\chi \|\sigma$ is $\square$-split for some $\sigma$; this is the case if and only if $\chi \|(\coker \chi)$ is $\square$-split. 
Dually, an epimorphism is called allowable if $(\ker \sigma) \| \sigma$ is $\square$-split.
%
The following conditions on a morphism $\alpha$ are equivalent:
\begin{enumerate}
\item $\im \alpha$ is an allowable monomorphism and $\coim \alpha$ is an allowable epimorphism;
\item $\ker \alpha$ is an allowable monomorphism and $\coker \alpha$ is an allowable epimorphism;
\end{enumerate}
A morphism is called allowable when it satisfies any of these conditions (see \cite[p.~264]{MacLane1994}).


A \keyt{relative projective object} $P$ is any object of $\cat A$ such that, for every \textit{allowable} epimorphism $\sigma:B\to C$, each morphism $\epsilon :P \to C$ of $\cat A$ can be factored through $\sigma$ as $\epsilon = \sigma \epsilon'$ for some $\epsilon' :P \to A$.  

A \keyt{resolvent pair} is a relative abelian category $\square: \cat A \to \cat M$ together with a covariant functor $F:\cat M \to \cat A$ left adjoint to $\square$.

\begin{prop}\label{resolvent_pair}
 Let  $\square: \cat A \to \cat M$ be a relative abelian category. The following conditions are equivalent:
\begin{enumerate}
\item\label{resolvent_pair:cond1} there exists a covariant functor $F:\cat M \to \cat A$ left adjoint to $\square$;
\item\label{resolvent_pair:cond2}  there exist a covariant functor $F:\cat M \to \cat A$, and a natural transformation $e:1_{\cat M} \to \square F$ (where $1_{\cat M}$ is the identity functor), such that every $u:M\to A_\square$ in $\cat M$ has a factorization $u= \alpha_\square e_M$, with $\alpha:F(M) \to A$ unique.
\end{enumerate} 
\end{prop}
Remark that $e$ is the unit of the adjunction.

\begin{ex}[continuation of \ref{relative_modules}]
Take $F(M) = R\otimes_S M$ and $e_M = 1 \otimes m \in F(M)$. Given a map of $S$-modules $u:M\to A_\square$ , define $\alpha:FM \to A$ by $\alpha(1\otimes m) = u(m)$.
\end{ex}

 

A complex $\epsilon:X\to A$ over $A$ (in $\cat A$) is a sequence of $\cat A$-objects and $\cat A$-morphisms
$...X_n \to X_{n-1} \to ... \to X_1 \to X_0 \xrightarrow{\epsilon} A \to 0,$
such that the composite of any two successive morphisms is zero.  This complex is called a \emph{resolution} of $C$ if the sequence is exact, \emph{relatively free} if each $X_n$ has the form $F(M_n)$ for certain $M_n$ in $\cat M$, and \emph{allowable} if all its morphisms are allowable.
 
Each object $C$ of $\cat A$ has a canonical relatively free resolution. Writing $\tilde F C$ for $F\square C$, and $\tilde F^n$ for its $n$-fold iteration, construct the objects $
B_n(C) = \tilde F^{n+1} C$, for each $n\in \Nn.$ Define $\cat M$-morphisms $s_\bullet$ between the corresponding objects 
\begin{equation}
\begin{tikzcd}
\square C \ar[r, "s_{-1}"] 
& \square B_0(C) \ar[r, "s_{0}"]
& \square B_1(C) \ar[r, "s_{1}"]
& \square B_2(C) \ar[r, "s_{2}"]
& \ldots
\end{tikzcd}
\end{equation}
as $s_{-1}:= e(\square C)$ and $s_n := e(\square B_n(C))$ (here $e$ is the natural transformation in Proposition \ref{resolvent_pair}).

\begin{prop}[see \protect{\cite[p.~268]{MacLane1994}}]\label{general_bar_resolution}
There are unique $\cat A$-morphisms 
$$\epsilon: B_0(C) \to C, \quad \partial_{n+1}:B_{n+1}(C) \to B_{n}(C) \quad \text{for }n\in \Nn,$$
which make $B(C):=\{B_n(C)\}_n$ a relatively free allowable resolution of $C$ with $s$ as contracting homotopy in $\cat M$. This resolution, with its contracting homotopy, is a covariant functor of $C$.
\end{prop}
\begin{proof}
We simply quote here the construction of $\epsilon$ and $\partial_n$. They form the following diagram (solid arrows belong to $\cat A$, and dashed arrows belong to $\cat M$):
\begin{equation}
\begin{tikzcd}
0 
& C      \ar[l] \ar[r, "s_{-1}" below, shift right, dashed] 
& B_0(C) \ar[l, "\epsilon" above, shift right]  \ar[r, "s_{0}" below, shift right, dashed] 
& B_1(C) \ar[l, "\partial_1" above, shift right]  \ar[r, "s_{1}" below, shift right, dashed] 
& B_2(C) \ar[l, "\partial_2" above, shift right]  \ar[r, "s_{2}" below, shift right, dashed] 
& {\ldots} \ar[l, "\partial" above, shift right]
\end{tikzcd}
\end{equation}
By Proposition \ref{resolvent_pair}, $1_{\square C}$ factors through a unique $\epsilon:B_0(C) \to C$; the formula $1_{\square C} = \epsilon_\square e_C$ shows that $\epsilon$ is allowable (note that $\epsilon$ is an epimorphism). Boundary operators are defined by recursion so that $s$ will be a contracting homotopy. Given $\epsilon$, the morphism $ 1_{\square B_0}-s_{-1}\epsilon_\square$ factors uniquely as ${\partial_1}_\square s_0$, for some $\partial_1:B_1(C) \to B_0(C)$. Similarly, $  1_{\square B_n}-s_{n-1}{\partial_n}_\square : \square B_n(C) \to \square B_n(C)$ determines $\partial_{n+1}$ given $\partial_n$, as the unique $\cat A$-morphism such that ${\partial_{n+1}}_\square s_n =  1_{\square B_n}-s_{n-1}{\partial_n}_\square$.
\begin{equation}
\begin{tikzcd}
\square B_{n+1} \ar[r, "{\partial_1}_\square"] & \square B_{n} \\
\square B_{n} \ar[u, "s_n"] \ar[ur, "1_{\square B_{n}}" below, outer sep=5pt] &
\end{tikzcd}
\end{equation}
\end{proof}

The resolution $B(C)$ is called the (unnormalized) bar resolution. 

\subsection{Example: Presheaves of modules}\label{sec:exo_relative}
We develop now the particular case relevant to our theory. Let $\Sinf$ be a category, and  $\sheaf R,\sheaf T:\Sinf^{\text{op}}\to \cat{Rings}$ presheaves, such that $\sheaf T_X$ is a subring of $\sheaf R_X$ with the same unit, for every $X\in \Ob\cat \sheaf T$. Take $\cat A=\Mod(\sheaf R)$, the category of presheaves of $\sheaf R$-modules, and $\cat M=\Mod(\sheaf T)$, the category of presheaves of $\sheaf T$-modules. A relative abelian category is obtained when $\square: \cat A \to \cat M$ is the forgetful functor over each $X$, as defined in Example \ref{relative_modules}. The functor $F:\cat M \to \cat A$ sends a presheaf $\sheaf P$ to the new presheaf $X\mapsto  \sheaf R_X \otimes_{\sheaf T_X} \sheaf P_X $,\footnote{This is a left $\sheaf R$-module with action defined by $r(r'\otimes p)=(rr')\otimes g$. For iterated tensor products, this definition is not canonical; for example, when considering $\sheaf R_X \otimes \sheaf R_X \otimes \sheaf P_X$, the element $(sr)\otimes r' \otimes g$ does not equal $r\otimes (sr') \otimes g$ (for $s\in \sheaf T$), unless $\sheaf T$ is in the center of $\sheaf R$. As in this work we only use commutative rings and algebras, these differences do not pose any problem.} and each morphism of $\sheaf T$-presheaves (in short, $\sheaf T$-morphism)  $f:M\to N$ to the $\sheaf R$-morphism defined by 
\begin{equation}
\forall X\in \Ob \Sinf, \forall m\in M(X), \quad Ff(X)(1\otimes m) = 1 \otimes f(m), \quad \text{for }X\in \Sinf.
\end{equation}
 The natural transformation $e$ mentioned in Proposition \ref{resolvent_pair} corresponds to a collection of $\sheaf T$-morphisms $e_\sheaf P:\sheaf P \to \square F (\sheaf P)$, one for each presheaf $\sheaf P$ of $\sheaf T$-modules; given $X$ in $\Sinf$, we define $e_\sheaf P(X)(m) = 1\otimes m$ for each $m\in M(X)$. \footnote{Of course, one has to prove that $e$ is in fact a natural transformation and satisfies the properties required by Proposition \ref{resolvent_pair}. This proof is rather trivial but complicated to write, and we omit it.}

Fix now a presheaf $\sheaf C$ in $\Mod(\sheaf R)$. We denote by $X$ a generic element in $\Ob\cat S$. Then, $B_0 \sheaf C(X) :=F\square \sheaf C(X) = \sheaf R_X \otimes_{\sheaf T_X} (\square \sheaf C(X))$; this $\sheaf R_X$-module is formed by finite $\sheaf R_X$-linear combinations of tensors $1\otimes c$, with $c\in \sheaf C$. Generally, an element of $B_n \sheaf C(X) = \sheaf R_X \otimes \square B_{n-1}\sheaf C(X)$, for $n\geq 1$, is a finite $\sheaf R_X$-linear combination of tensors $1\otimes r_1\otimes r_2 \otimes ...\otimes r_n \otimes c$. The ring $\sheaf R_X$ acts on $B_n \sheaf C(X)$ by multiplication on the first factor of the tensor product; to highlight this fact, people usually write  $r[ r_1 | r_2 | ...| r_n | c]$ instead of $r\otimes r_1\otimes r_2 \otimes ...\otimes r_n \otimes c$. This notation explains the name ``bar resolution'' adopted above. The definition of $e$ implies that
\begin{equation}
s_{-1}^X: \square \sheaf C(X) \to \square B_0 \sheaf C(X), \quad c \mapsto 1 \otimes c = [c],
\end{equation}
and
\begin{equation}
s_{n}^X: \square B_n \sheaf C(X) \to \square B_{n+1}\sheaf C(X), \quad r[r_1| r_2 | ...| r_n | c] \mapsto [r|r_1| r_2 | ...| r_n | c] \quad \text{for } n\in \Nn.
\end{equation}
These equalities determine $s_\bullet$, since these functions are $\sheaf T_X$-linear.

Now $\epsilon$ is the \emph{unique} $\sheaf R$-morphism such that $1_{\square \sheaf C} = \epsilon_\square e_\sheaf C$; this is clearly the case if $\epsilon^X([c]) = c$. Similarly, $\partial_1$ is the unique $\sheaf R$-morphism from $B_1 \sheaf C$ to $B_0 \sheaf C$ that satisfies 
\begin{equation}\label{recurrence_bord_1}
 {\partial_1}_\square s_0 = 1-s_{-1}^X\epsilon_\square
\end{equation}
Since $B_1 \sheaf C(X)$ is generated as a $\sheaf R_X$-module by the elements $[r|c]$, and $s_0^X(r[c]) = [r|c]$, the equation \eqref{recurrence_bord_1} defines $\partial_1$ completely. Just remark that $\epsilon(r[c])= r\epsilon([c])=r\epsilon(1\otimes c) = rc$ and $s_1(rc)=[rc]$. We conclude that 
\begin{equation}
\partial_1([r|c]) = r[c] - [rc].
\end{equation}
It can be proved by recursion that (cf. \cite[p.~281]{MacLane1994})
\begin{equation}
\partial [r_1|...|r_n|c] = r_1[r_2|...|r_n|c] + \sum_{k=1}^{n-1} (-1)^k [r_1|...|r_kr_{k+1}|...|r_n|c] +(-1)^n[r_1|...|r_{n-1}|r_nc].
\end{equation}
In virtue of Proposition \ref{general_bar_resolution}, we obtain in this way a free allowable resolution of $\sheaf C$.

\end{document}